%% file: main.tex
\documentclass[onecolumn,11pt, journal]{IEEEtran}

\input{Math_Template_2019.tex}

\input{Quantum_header.tex}

\usepackage{stmaryrd }
\usepackage{stackengine}
\def\delequal{\mathrel{\ensurestackMath{\stackon[1pt]{=}{\scriptscriptstyle\Delta}}}}

\usepackage{relsize}
\usepackage{upgreek}

\renewcommand*\footnoterule{}

\usepackage{acronym}
\newacro{qc}[q-c]{quantum-to-classical }

\def\PiA{\Pi_{\rho_A}}
\def\PiB{\Pi_{\rho_B}}
\def\PihatA{\hat{\Pi}^A}
\def\PihatB{\hat{\Pi}^B}
\def\PiuA{\Pi_{u^n}^A}
\def\PivB{\Pi_{v^n}^B}
\def\lambdauA{\lambda^A_{u^n}}
\def\lambdavB{\lambda^B_{v^n}}
\def\lambdaAB{\lambda^{AB}_{u^n,v^n}}
\def\rhohatuA{\hat{\rho}^A_{u^n}}
\def\rhohatvB{\hat{\rho}^B_{v^n}}
\def\rhohatuAvB{\hat{\rho}^{AB}_{u^n,v^n}}
\def\LambdauA{\Lambda^A_{u^n}}
\def\LambdavB{\Lambda^B_{v^n}}

\def\TDeltan{\mathcal{T}_{\delta}^{(n)}}
\def\IndiU1{\mathbbm{1}_{\left\{U^{n,(\mu_1)}(l)=u^n\right\}} }
\def\IndiV1{\mathbbm{1}_{\left\{V^{n,(\mu_2)}(k)=v^n\right\}} }
\def\TDeltaN{\mathcal{T}_{\delta}^{(n)}}

\settoggle{Blue_revision}{false}

\makeatletter
\def\footnoterule{\kern-3\p@
  \hrule \@width 6.5in \kern 2.6\p@} 
\makeatother
\begin{document}
	
	
\setlength{\baselineskip}{16pt}

	\title{  \huge Faithful Simulation of Distributed Quantum Measurements with Applications in Distributed Rate-Distortion Theory}
	
	\author{Touheed Anwar Atif \IEEEauthorrefmark{1}, Mohsen Heidari \IEEEauthorrefmark{2} and  S.\ Sandeep Pradhan \IEEEauthorrefmark{1},\\
      {\IEEEauthorrefmark{1} 
EECS Dept., University of  Michigan, Ann Arbor}, MI. \\
 {\IEEEauthorrefmark{2}  CS Dept., Purdue University, West Lafayette, IN. \\
 \IEEEauthorrefmark{1}\tt \{touheed, pradhanv\}@umich.edu}, \IEEEauthorrefmark{2}\tt mheidari@purdue.edu

\thanks{This work was presented in part at IEEE International Symposium on Information Theory (ISIT) 2019.}}

%

		
		\maketitle

		\begin{abstract}

  We consider the task of faithfully simulating a distributed quantum measurement, wherein we provide a protocol for the three parties, Alice, Bob and Eve, to simulate a repeated action of a distributed quantum measurement using a pair of non-product approximating measurements by Alice and Bob, followed by a stochastic mapping at Eve. The objective of the protocol is to utilize minimum resources, in terms of classical bits needed by Alice and Bob to communicate their measurement outcomes to Eve, and the common randomness shared among the three parties, while faithfully simulating independent repeated instances of the original measurement. To achieve this, we develop a mutual covering lemma and a technique for random binning of distributed quantum measurements, and, in turn, characterize a set of sufficient communication and common randomness rates required for 
  asymptotic simulatability in terms of single-letter quantum information quantities. Furthermore, using these results we address a distributed quantum rate-distortion problem, where we characterize the achievable rate distortion region through a single-letter inner bound.
  Finally, via a technique of single-letterization of multi-letter quantum information quantities, we provide an outer bound for the rate-distortion region.

		\end{abstract}
		

	
	
	

{\hypersetup{
colorlinks=true, %
 pdfstartview={FitH},
    linkcolor=black,
    citecolor=blue, 
    urlcolor={blue!80!black}
}
	\tableofcontents}
	\setcounter{tocdepth}{2}
	\input{introduction2.tex}

	\section{Preliminaries } \label{sec:prelim}
We here establish all our notations, briefly state few necessary definitions, and also provide Winter's theorem on measurement compression. \\
\noindent \textbf{Notation:} Given any natural number $M$, let the finite set $\{1, 2, \cdots, M\}$ be denoted by $[1,M]$. Let $\mathcal{B(H)}$ denote the algebra of all bounded linear operators acting on a finite dimensional Hilbert space $\mathcal{H}$. Further, let $ \mathcal{D(H)} $ denote the set of all unit trace positive operators acting on $ \mathcal{H} $. Let $I$ denote the identity operator. The trace distance between two operators $A$ and $B$ is defined as $\|A-B\|_1\deq \tr|A-B|$, where for any operator $\Lambda$ we define $|\Lambda|\deq \sqrt{\Lambda^\dagger \Lambda}$. The von Neumann entropy of a density operator $\rho \in \mathcal{D}(\mathcal{H})$ is denoted by $S(\rho)$.  The quantum mutual information for a bipartite density operator $\rho_{AB} \in \mathcal{D}(\mathcal{H}_A \otimes \mathcal{H}_B)$ is defined as 
\begin{align*}
I(A;B)_{\rho}&\deq S(\rho_{A})+S(\rho_{B})-S(\rho_{AB}).
\end{align*} 
Given any ensemble $\{p_i, \rho_i\}_{i\in [1,m]}$, the Holevo information, as in \cite{holevo}, is defined as 
\begin{align*}
    \chi \big( \{p_i, \rho_i\}\big) \deq S\Big(\sum_i p_i \rho_i\Big) - \sum_i p_i S(\rho_i).
\end{align*}
 A positive-operator valued measure (POVM) acting on a Hilbert space $\mathcal{H}$ is a collection $M\deq \{\Lambda_x\}_{x \in \mathcal{X}}$ of positive operators in $\mathcal{B}(\mathcal{H})$ that form a resolution of the identity:
\begin{align*}
\Lambda_x\geq 0, \forall x \in \mathcal{X}, \qquad \sum_{x \in \mathcal{X}} \Lambda_x=I.
\end{align*}
where $\mathcal{X}$ is a finite set.
If instead of the equality above, the inequality $\sum_x \Lambda_x \leq I$ holds, then the collection is said to be a sub-POVM. A sub-POVM $M$ can be completed to form a POVM, denoted by $[M]$, by adding the operator $\Lambda_{0}\deq (I-\sum_x \Lambda_x )$ to the collection.  Let $\Psi^\rho_{RA}$ denote a purification of a density operator $\rho \in D(\mathcal{H}_A)$. Given a POVM $M\deq \{\Lambda^A_x\}_{x \in \mathcal{X}}$ acting on  $\rho$, the post-measurement state of the reference together with the classical outputs is represented by 
\begin{equation}\label{eq:1}
     (\text{id} \tensor M)(\Psi^\rho_{RA})\deq \sum_{x\in \mathcal{X}} \ketbra{x}\tensor \tr_{A}\{(I^R \tensor \Lambda_x^A ) \Psi^\rho_{RA} \}.
\end{equation} 
Consider two POVMs $M_A=\{\Lambda^A_x\}_{x \in \mathcal{X}}$ and $M_B=\{\Lambda^B_y\}_{y \in \mathcal{Y}}$ acting on $\mathcal{H}_A$ and $\mathcal{H}_B$, respectively. Define $M_A\tensor M_B \deq \{\Lambda^A_x\tensor \Lambda^B_y\}_{x\in \mathcal{X},y\in \mathcal{Y}}$
With this definition, $M_A\tensor M_B$ is a POVM acting on $\mathcal{H}_A\tensor \mathcal{H}_B$. By $M^{\tensor n}$ denote the $n$-fold tensor product of the POVM $M$ with itself.  
\begin{definition}[Joint Measurements]\label{def:Joint Measurements}A POVM $M_{AB} = \{\Lambda^{AB}_z\}_{z \in \mathcal{Z}}$, acting on the joint state $\rho_{AB} \in \mathcal{D}(\mathcal{H}_{A}\tensor \mathcal{H}_B)$, is said to have a separable decomposition with stochastic integration if there exist POVMs $\bar{M}_A=\{\bar{\Lambda}^A_{u}\}_{u \in \mathcal{U}}$ and $\bar{M}_B=\{ \bar{\Lambda}^B_{v}\}_{v \in \mathcal{V}}$ and a stochastic mapping $P_{Z|U,V}: \mathcal{U}\times\mathcal{V}\rightarrow\mathcal{Z}$ such that 
    \begin{equation*}
    \Lambda^{AB}_{z}=\sum_{u,v} P_{Z|U,V}(z|u,v) \bar{\Lambda}^A_{u}\tensor \bar{\Lambda}^B_{v}, \quad  \forall z \in \mathcal{Z},
    \end{equation*} 
where $\mathcal{U}, \mathcal{V}$ and $\mathcal{Z}$ are some finite sets.
Further, if the mapping $P_{Z|U,V}$ is a deterministic function then the POVM is said to have a separable decomposition with deterministic integration. 
\end{definition}

\noindent{ \textbf {Measurement Compression Theorem:}} \label{sec:Winter_compression}
Here, we provide a brief overview of the measurement compression theorem \cite{winter}.
A key concern, from an information-theoretic standpoint, is to quantify the amount of  “relevant information” conveyed by a measurement about a quantum state.  Winter quantified “relevant information” by measuring the minimum amount of classical information bits needed to “simulate” the repeated action of a measurement $M$ on a quantum state $\rho$. In this context, an agent (Alice) performs an approximating measurement \(\tilde{M}^{(n)}\) on a quantum state \(\rho^{\tensor n}\) and sends a set of classical bits to a receiver (Bob). In addition, Alice and Bob share some amount of common randomness. Bob intends to faithfully recover the outcomes of the original  measurement $M$ without having access to the quantum state based on the bits received from Alice and the common randomness. The objective is to minimize the rate of classical bits under the constraint that the recovered and the original outcomes be statistically indistinguishable. This is formally defined in the following.



\begin{definition}[Faithful simulation \cite{wilde_e}]\label{def:faith-sim}
Given a POVM ${M}\deq \{\Lambda_x\}_{x\in \mathcal{X}}$ acting on a Hilbert space $\mathcal{H}_A$ and a density operator $\rho\in \mathcal{D}(\mathcal{H}_A)$, a sub-POVM $\tilde{M}\deq \{\tilde{\Lambda}_{x}\}_{x\in \mathcal{X}}$ acting on $\mathcal{H}_A$ is said to be $\epsilon$-faithful to $M$ with respect to $\rho$, for $\epsilon > 0$, if the following holds: 
\begin{equation}\label{eq:faithful-sim-cond-1_2}
\sum_{x\in \mathcal{X}} \Big\|\sqrt{\rho} (\Lambda_{x}-\tilde{\Lambda}_{x}) \sqrt{\rho}\Big\|_1+\tr\left\{(I-\sum_{x} \tilde{\Lambda}_{x})\rho\right\}   \leq \epsilon.
\end{equation}
\end{definition}
The above trace norm constraint can be equivalently expressed in terms of a purification of state $\rho $ using the following lemma.
\begin{lem}\emph{\cite{wilde_e}}\label{lem:faithful_equivalence} 
For any state $\rho\in \mathcal{D}(\mathcal{H})$ with any purification $\Psi^{\rho}_{RA}$, and any pair of POVMs $M$ and $\tilde{M}$ acting on $\mathcal{H}$, the following identity holds 
\begin{equation}\label{eq:faithful_equivalence}
\|(\emph{\id} \tensor M)(\Psi^{\rho}_{RA}) - (\emph{\id}\tensor \tilde{M})(\Psi^{\rho}_{RA}) \|_1=\sum_{x} \|\sqrt{\rho} (\Lambda_{x}-\tilde{\Lambda}_{x}) \sqrt{\rho}\|_1,
\end{equation}
where $\Lambda_x$ and $\tilde{\Lambda}_x$ are the operators associated with $M$ and $\tilde{M}$, respectively.
\end{lem}

\begin{theorem}\emph{\cite{winter}}\label{extrinsic}
For any ${\epsilon} >0 $, any density operator $\rho\in \mathcal{D}(\mathcal{H}_A)$ and any POVM $M$ acting on the Hilbert space $\mathcal{H}_A$, there exist a collection of POVMs $\tilde{M}^{(n,\mu)}$  for $\mu \in [1,N]$, each acting on $\mathcal{H}_A^{\otimes n}$, and having at most $2^{nR}$ outcomes such that $\tilde{M}^{(n)} \deq\frac{1}{N}\sum_{\mu}\tilde{M}^{(n,\mu)}$ is  $\epsilon$-faithful to $M^{\tensor n}$ with respect to $\rho^{\tensor n}$ if
\begin{align*}
    & R \geq I(U;R)_\sigma + \delta(\epsilon),~~ ~~ \frac{1}{n}\log_2{N}+R \geq S(U)_\sigma + \delta(\epsilon), 
\end{align*}
where  $\sigma_{RU}\deq (\emph{\id}\tensor M)(\Psi^{\rho}_{RA}) $, and $\delta(\epsilon)\ssearrow 0$ as $\epsilon\ssearrow 0$.
\end{theorem}

\begin{remark}
A strong converse of the above result is also provided in \cite{winter}.
\end{remark}

	\input{POVM_Appx.tex}

    \input{POVM_Appx_Proof.tex}


	\section{Q-C Distributed Rate Distortion Theory } \label{sec:QC_dist scr coding}
	
As an application of  faithful simulation of distributed measurements (Theorem \ref{thm: dist POVM appx}), we consider the distributed extension of \ac{qc} rate distortion coding \cite{qtc}. This problem is a quantum counterpart of the classical distributed source coding.  In this setting, consider a memoryless bipartite quantum source, characterized by $\rho_{AB}\in \mathcal{D}(\mathcal{H}_{A}\tensor \mathcal{H}_B)$. Alice and Bob have access to sub-systems $A$ and $B$, characterized by $\rho_A \in \mathcal{D}(\mathcal{H}_A)$  and $\rho_A \in \mathcal{D}(\mathcal{H}_A)$, respectively, where $\rho_A = \Tr_B\{\rho_{AB}\}$ and $\rho_B = \Tr_A\{\rho_{AB}\}$.
They both perform a measurement on $n$ copies of their sub-systems and send the classical bits to Eve. Upon receiving the classical bits sent by Alice and Bob, a reconstruction state is produced by Eve. The objective of Eve is to produce a reconstruction of the source $\rho_{AB}$ within a targeted distortion threshold which is measured by a given distortion observable.  
\subsection{Problem Formulation}
We first formulate this problem as follows. For any quantum information source, characterized by $\rho_{AB} \in \mathcal{D}(\mathcal{H}_A\tensor\mathcal{H}_B)$, denote its purification by $\phirho$.
\begin{definition} \label{def:qc source coding setup}
A \ac{qc} source coding setup is characterized by a triple $(\phirho,\mathcal{H}_{\hat{X}},\Delta)$, where $\phirho \in \mathcal{D}(\mathcal{H}_R\tensor \mathcal{H}_A\tensor \mathcal{H}_B)$ is a purified quantum state,  $\mathcal{H}_{\hat{X}}$ is a reconstruction Hilbert space, and $\Delta\in \mathcal{B}(\mathcal{H}_R \tensor \mathcal{H}_{\hat{X}})$, which satisfies $\Delta\geq 0$, is a distortion observable.
\end{definition}
Next, we formulate the action of Alice, Bob and Eve by the following definition. 
\begin{definition} \label{def:qc source code}
An $(n,\Theta_1,\Theta_2)$ \ac{qc} protocol for a given input and reconstruction Hilbert spaces $(\mathcal{H}_A\tensor \mathcal{H}_B, \mathcal{H}_{\hat{X}})$ is defined by POVMs $M_A^{(n)}$ and $M_B^{(n)}$ acting on $\mathcal{H}^{\tensor n}_A$ and $\mathcal{H}^{\tensor n}_B$ with $\Theta_1$ and $ \Theta_2$ number of outcomes, respectively, and a set of reconstruction states $S_{i,j}\in \mathcal{D}(\mathcal{H}_{\hat{X}}^{\tensor n})$ for all $i\in [1:\Theta_1], j\in [1:\Theta_2]$.
\end{definition}
The overall action of Alice, Bob and Eve, as a \ac{qc} protocol, on a quantum source $\rho_{AB}$ is given by the following operation 
\begin{equation}\label{eq:N_AB action}
\mathcal{N}_{A^nB^n\mapsto \hat{X}^n}: \rho^{\tensor n}_{AB} \mapsto  \sum_{i,j}  \tr\{(\Lambda^A_i\tensor \Lambda^B_j) \rho^{\tensor n}_{AB}\}~ S_{i,j}, 
\end{equation}
where $\{\Lambda^A_i\}$ and $\{\Lambda^B_j\}$ are the operators of the POVMs $M_A^{(n)}$ and $M_B^{(n)}$, respectively. With this notation and given a q-c source coding setup as in Definition \ref{def:qc source coding setup}, the distortion of a $(n=1, \Theta_1, \Theta_2)$ \ac{qc} protocol is measured as
 		\begin{equation}\label{eq:protocol distortion}
	{d}(\rho_{AB},\mathcal{N}_{AB\mapsto \hat{X}} ) \deq \tr\left\{\Delta \left((\id_R \tensor \mathcal{N}_{AB\mapsto \hat{X}}) (\Psi_{RAB}^{\rho_{AB}})\right)\right\} \nonumber.
	\end{equation}
%
For an $n$-letter protocol, we use symbol-wise average distortion observable defined as
	\begin{align}
	\Delta^{(n)} & = \cfrac{1}{n}\sum_{i=1}^{n}\Delta_{R_i\hat{X}_i}\otimes I_{R\hat{X}}^{\otimes [n]\backslash i},
	\end{align}
	where $\Delta_{R_i\hat{X}_i}$ is understood as the observable $\Delta$ acting on the $i$th instance space $\mathcal{H}_{R_i}\tensor \mathcal{H}_{\hat{X}_i}$ of the $n$-letter space $\mathcal{H}^{\tensor n}_{R}\tensor \mathcal{H}^{\tensor n}_{\hat{X}}$.
	With this notation, the distortion for an $(n, \Theta_1, \Theta_2)$ \ac{qc} protocol is given by
	\begin{equation*}\label{avg_dis3}
	\bar{d}(\rho_{AB}^{\otimes n},  \mathcal{N}_{A^nB^n\mapsto \hat{X}^n})\deq \tr\Big\{\Delta^{(n)}
	(\id \tensor \mathcal{N}_{A^nB^n\mapsto \hat{X}^n} )({\Psi_{R^nA^nB^n}^{\rho_{AB}}})
	\Big\}, 
	\end{equation*}
	where $\Psi_{R^nA^nB^n}^{\rho_{AB}}$ is the $n$-fold tensor product of  $\Psi_{RAB}^{\rho_{AB}}$ which is the given purification of the source.
	    
 The authors in \cite{qtc}   studied the point-to-point version of the above formulation. They considered a special distortion observable of the form $\Delta = \sum_{\hat{x} \in \hat{\mathcal{X}}}\Delta_{\hat{x}} \otimes\ketbra{\hat{x}},$
 where $\Delta_{\hat{x}}\geq 0$  acts on the reference Hilbert space and $\hat{\mathcal{X}}$ is the reconstruction alphabet. In this paper, we allow $\Delta$ to be any non-negative and bounded operator acting on the appropriate Hilbert spaces. Moreover, we allow for the use of any c-q reconstruction  mapping as the action of Eve. 

\begin{definition} \label{eq:qc achievable}
For a \ac{qc} source coding setup $(\phirho,\mathcal{H}_{\hat{X}},\Delta)$, a rate-distortion triplet ($R_1, R_2, D$) is said to be achievable, if for all $ \epsilon > 0 $ and all sufficiently large $n$, there exists an $ (n, \Theta_1, \Theta_2) $ \ac{qc} protocol satisfying
	\begin{align*}
	\frac{1}{n}\log_2 \Theta_i  &\leq R_i + \epsilon,  \quad i=1,2, \\
	\bar{d}(\rho_{AB}^{\otimes n}, \mathcal{N}_{A^nB^n\mapsto \hat{X}^n})
	&\leq D + \epsilon,	
    \end{align*}	
    where $\mathcal{N}_{A^nB^n\mapsto \hat{X}^n}$ is defined as in \eqref{eq:N_AB action}. The set of all achievable rate-distortion triplets $(R_1, R_2, D)$ is called the achievable rate-distortion region. 
\end{definition}
Our objective is to characterize the achievable rate-distortion region using single-letter information quantities. 

\subsection{Inner Bound}
%
We provide an inner bound to the achievable rate-distortion region which is stated in the following theorem. We employ a q-c protocol based on a randomized faithful simulation strategy involving a time sharing classical random variable $Q$ that is independent of the quantum source. This can be viewed as a conditional version of the faithful simulation problem considered in Section \ref{sec:Appx_POVM}.   

\begin{theorem}\label{thm:q-c dist scr}
	 For a \ac{qc} source coding setup $(\phirho,\mathcal{H}_{\hat{X}},\Delta)$, 
any rate-distortion triplet $ (R_1,R_2, D)$ satisfying the following inequalities is achievable
	 \begin{align*}
R_1 &\geq I(U;RB|Q)_{\sigma_1}-I(U;V|Q)_{\sigma_3},\\
R_2 &\geq I(V;RA|Q)_{\sigma_2}-I(U;V|Q)_{\sigma_3},\\
R_1+R_2 &\geq I(U;RB|Q)_{\sigma_1}+I(V;RA|Q)_{\sigma_2}-I(U;V|Q)_{\sigma_3},\\
D&\geq {d}(\rho_{AB},\mathcal{N}_{AB\mapsto \hat{X}}),
\end{align*} 
for  POVM  of the form $M_{AB} = \sum_{q \in \mathcal{Q}} P_Q(q) M_A^q \tensor M_B^q$, where for every $q\in \mathcal{Q}$,  $\add{M^q_A\deq}\{\Lambda^{A,q}_u\}_{u\in \mathcal{U}}$ and $\add{M^q_B}\deq\{\Lambda^{B,q}_v\}_{v\in \mathcal{V}}$ are POVMs acting on $\mathcal{H}_A\tensor \mathcal{H}_B$, and reconstruction states $\{S_{u,v,q}\}$ with each state in  $\mathcal{D}(\mathcal{H}_{\hat{X}})$, and some finite sets $\mathcal{U},\mathcal{V}$ and $\mathcal{Q}$. The quantum mutual information quantities are computed according to the auxiliary states
\add{$\sigma_{1}^{RUBQ} \deq \sum_{q \in \mathcal{Q}} P_Q(q) (\emph{id}_R\tensor M^q_A \tensor \emph{id}_B) ( \Psi^{\rho_{AB}}_{R A B})\tensor \ketbra{q},~ \sigma_{2}^{RAVQ} \deq \sum_{q \in \mathcal{Q}} P_Q(q) (\emph{id}_R\tensor \emph{id}_A \tensor M^q_B)\tensor \ketbra{q},$ and $
\sigma_{3}^{RUVQ} \deq \sum_{q \in \mathcal{Q}} P_Q(q) (\emph{id}_R\tensor M^q_A \tensor M^q_B) ( \Psi^{\rho_{AB}}_{R A B})\tensor \ketbra{q},$
}
where $(U,V)$ represents the output of $M_{A B}$, and $\mathcal{N}_{AB \mapsto \hat{X}}: \rho_{AB} \mapsto  \sum_{u,v,q} P_Q(q) \tr\{(\Lambda^{A,q}_u\tensor \Lambda^{B,q}_v) \rho_{AB}\}~ S_{u,v,q}$.
\end{theorem}

\begin{remark}
Note that for the auxiliary states $\sigma_i, i=1,2,3$, we have $I(R;Q)=0$.
\end{remark}	

\begin{proof}
In the interest of brevity, we provide the proof for the special case, when the time sharing random variable is trivial, i.e., $\mathcal{Q}$ is empty. An extension to the more general case is straightforward but tedious. For the special case, the proof follows from Theorem \ref{thm: dist POVM appx}. Fix POVMs $(M_{A}, M_B)$ and reconstruction states $S_{u,v}$ as in the statement of the theorem. Let $\mathcal{N}_{AB\mapsto \hat{X}}$ be the mapping corresponding to these POVMs and the reconstruction states. Then, 
${d}(\rho_{AB},\mathcal{N}_{AB\mapsto \hat{X}})\leq D. $
 According to Theorem \ref{thm: dist POVM appx}, for any $\epsilon>0$,  there exists an $(n, 2^{nR_1},2^{nR_2}, N)$ distributed protocol for $\epsilon$-faithful simulation of $M_A^{\tensor n}\tensor M_B^{\tensor n}$ with respect to $\rho_{AB}^{\tensor n}$ such that $(R_1,R_2)$ satisfies the inequalities in \eqref{eq:dist POVm appx rates} for $\bar{M}_A=M_A$ and $\bar{M}_B=M_B$.  Let $\tilde{M}_A^{(\mu)}, \tilde{M}_B^{(\mu)}, \mu \in [1:N]$ and $f^{(\mu)}$ be the POVMs and the deterministic decoding functions of this protocol with $\mathcal{Z}=\mathcal{U}\times \mathcal{V}$. 
%
%
We use these POVM's and mappings to construct a \ac{qc} protocol for distributed quantum source coding.

For each $\mu\in [1,N]$, consider the \ac{qc} protocol with parameters $\Theta_i=2^{nR_i}, i=1,2$, and POVMs $\tilde{M}_A^{(\mu)}, \tilde{M}_B^{(\mu)}$. Moreover, we use n-length reconstruction states $S_{i,j}\deq \sum_{u^n,v^n} \11\big\{f^{(\mu)}(i,j)= (u^n,v^n)\big\} S_{u^n,v^n}$, where $S_{u^n,v^n}=\tensor_{i}S_{u_i, v_i}$. Further, let the corresponding mappings be denoted as $\tilde{\mathcal{N}}^{(\mu)}_{A^nB^n\mapsto \hat{X}^n} $. 
With this notation, for the average of these random protocols, the following bounds hold:
{  \begin{align*}
\frac{1}{N}\sum_\mu  &\bar{d}(\rho_{AB}^{\tensor n}, \tilde{\mathcal{N}}^{(\mu)}_{A^nB^n\mapsto \hat{X}^n})=\frac{1}{N}\sum_\mu \tr\Big\{ \Delta^{(n)} (\id\tensor \tilde{\mathcal{N}}^{(\mu)}_{A^nB^n\mapsto \hat{X}^n} )\phirhon \Big\}\\
&=\tr\Big\{ \Delta^{(n)} (\id\tensor \mathcal{N}^{\tensor n}_{AB\mapsto \hat{X}})\phirhon \Big\}+\tr\Big\{ \Delta^{(n)} (\id\tensor (\mathcal{N}^{\tensor n}_{AB\mapsto \hat{X}}- \tilde{\mathcal{N}}_{A^nB^n\mapsto \hat{X}^n}))\phirhon \Big\}\\
&\leq \tr\left\{\Delta \left((\id_R \tensor \mathcal{N}_{AB\mapsto \hat{X}}) (\phirho)\right)\right\}+\| \Delta^{(n)} (\id\tensor (\mathcal{N}^{\tensor n}_{AB\mapsto \hat{X}}- \tilde{\mathcal{N}}_{A^nB^n\mapsto \hat{X}^n}))\phirhon\|_1\\
&\leq D +\|\Delta^{(n)}\|_{\infty} \|(\id\tensor (\mathcal{N}^{\tensor n}_{AB\mapsto \hat{X}}- \tilde{\mathcal{N}}_{A^nB^n\mapsto \hat{X}^n}))\phirhon\|_1\\
&\leq D +\|\Delta^{(n)}\|_{\infty} \|(\id\tensor (M_A^{\tensor n}\tensor M_B^{\tensor n} - \tilde{M}_{AB}))\phirhon\|_1\\
&\leq D+\epsilon \|\Delta\|_{\infty},
\end{align*} }
where $\tilde{\mathcal{N}}_{AB\mapsto \hat{X}}$ is the average of $\tilde{\mathcal{N}}^{(\mu)}_{AB\mapsto \hat{X}}$, and $\tilde{M}_{AB}$ is the overall POVM of the underlying distributed protocol as given in \eqref{eq:overall POVM opt}. The first inequality holds by the fact that $|\tr\{A\}|\leq \norm{A}_1$. The second inequality follows by Lemma \ref{lem:norm1 of AB} given in the sequel. The third inequality is due to the monotonicity of the trace-distance \cite{Wilde_book} with respect to the quantum channel given by $\id\tensor \mathcal{L}_{UV\mapsto \hat{X}}^{\tensor n}$, where 
\begin{equation*}
 \mathcal{L}_{UV\mapsto \hat{X}}(\omega)\deq \sum_{u,v} \bra{u,v}\omega \ket{u,v} S_{u,v}.  \end{equation*}
The last inequality follows by Theorem \ref{thm: dist POVM appx}, and the fact that $\|\Delta^{(n)}\|_{\infty}\leq \|\Delta\|_{\infty}$. This completes the proof of the theorem, since $\Delta$ is a bounded operator.
\end{proof}
\begin{lem}\label{lem:norm1 of AB}
For any operator $A$ and $B$ acting on a Hilbert space $\mathcal{H}$ the following inequalities hold. 
\begin{equation*}
\|B A\|_1\leq \|B\|_{\infty} \|A\|_1, \qquad \text{and} \qquad \|AB \|_1\leq \|B\|_{\infty} \|A\|_1.
\end{equation*}
\end{lem}
\begin{proof}
See Exercise 12.2.1 in \cite{Wilde_book}.
\end{proof}

 One can observe that the rate region in Theorem \ref{thm:q-c dist scr} matches in form with the classical Berger-Tung region when $\rho_{AB}$ is a mixed state of a collection of orthogonal pure states. Note that the  rate region is an inner bound for the set of all achievable rates. The single-letter characterization of the set of achievable rates is still an open problem even in the classical setting. Some progress has been made recently on this problem which provides an improvement over Berger-Tung rate region \cite{Farhad_Dist}.

\subsection{Outer Bound}
In this section, we provide an outer bound for the achievable rate-distortion region. 
\begin{theorem}\label{thm:dist str outer}
Given a \ac{qc} source coding setup $(\phirho,\mathcal{H}_{\hat{X}},\Delta)$,
if any triplet $(R_1, R_2, D)$ is achievable, then the following inequalities must be satisfied
 \begin{subequations}\label{eq:dist POVm appx outer}
	 \begin{align}
R_1 &\geq I(W_1;R|W_2,Q)_{\sigma},\\
R_2 &\geq I(W_2;R|W_1,Q)_{\sigma},\\
 R_1+R_2 &\geq I(W_1,W_2;R|Q)_{\sigma},\\
D&\geq \tr\{\Delta\}\sigma^{R\hat{X}}\},
\end{align}
\end{subequations}
for some state $\sigma^{W_1W_2RQ\hat{X}}$ which can be written as 
\begin{align*}
 \sigma^{W_1W_2QR\hat{X}}= (\id\tensor \mathcal{N}_{AB \mapsto W_1W_2Q\hat{X}})(\phirho),
 \end{align*} 
  where $Q$ represents an auxiliary quantum state, and  $ \mathcal{N}_{AB \mapsto W_1W_2Q\hat{X}}$ is a quantum test channel with $I(R;Q)_\sigma=0$.
\end{theorem}

\begin{proof}
Suppose the triplet $(R_1, R_2, D)$ is achievable. Then, from Definition \ref{eq:qc achievable}, for all $\epsilon>0$, there exists an $(n, \Theta_1, \Theta_2)$ q-c protocol satisfying the inequalities in the definition. Let $M_A\delequal \{\Lambda^A_{l_1}\}$, $M_B\delequal \{\Lambda^B_{l_2}\}$ and $S_{l_1, l_2}\in \mathcal{D}(\mathcal{H}_{\hat{X}}^{\tensor n})$ be the corresponding POVMs and reconstruction states. Let $L_1, L_2$ denote the outcomes of the measurements. Then, for Alice's rate, we obtain
\begin{align*}
n (R_1+\epsilon) &\geq H(L_1)\geq H(L_1|L_2)\\
&\geq I(L_1;R^n|L_2)_\tau\\
  &=\sum_{j=1}^n I(L_1;R_j|L_2,R^{j-1})_\tau.
\end{align*}
where the state $\tau$ is defined as 
\begin{equation*}
     \tau^{L_1L_2R^n\hat{X}^n}\deq  \sum_{l_1,l_2} \ketbra{l_1,l_2}\tensor 
\tr_{A^nB^n} \Big\{(\id\tensor {\Lambda}^{A}_{l_1}\tensor{\Lambda}^{B}_{l_2})\phirhon\Big\} \tensor S_{l_1,l_2}.
\end{equation*}
Note that for each $j$ the corresponding mutual information above is defined for a state in the Hilbert space $\mathcal{H}_{L_1}\tensor \mathcal{H}_{L_2}\tensor \mathcal{H}_{R}^{\tensor j}$. Next, we convert the above summation into a single-letter quantum mutual information term. For that we proceed with defining a new Hilbert space using direct-sum operation. 

Let us recall the direct-sum of Hilbert spaces \cite{Conway1985}. Consider a tuple of Hilbert spaces $\mathcal{H}_k$, $k=1,2,\ldots,n$ with inner products $\braket{\cdot}{\cdot}_k$. Define $\bigoplus_{k=1}^n \mathcal{H}_k$ as the collection of tuples of vectors 
$(\ket{x}_1,\ket{x}_2,\ldots,\ket{x}_n)$. The inner 
product of two tuples 
$(\ket{x}_1,\ket{x}_2,\ldots,\ket{x}_n)$ and 
$(\ket{y}_1,\ket{y}_2,\ldots,\ket{y}_n)$
is given by the sum of inner products of the components, i.e.,
$\sum_{k=1}^n \braket{x_k}{y_k}_k$. 
A linear operator in this space is a tuple of operators given by  $(A_1,A_2,\ldots,A_n)$, where $A_k$ operates on 
$\mathcal{H}_k$, and $\Tr(A)=\sum_{k=1}^n \Tr(A_i)$.
A state in $\bigoplus_{k=1}^n \mathcal{H}_k$ 
is denoted conventionally as 
$\bigoplus_{k=1}^n \ket{x}_k$.
Similarly, a linear operator in this space is
written in the form $A=\bigoplus_{k=1}^n A_k$.

With this definition, consider the following 
single-letterization:
\begin{align*}
\sum_{j=1}^n I(L_1;R_j|L_2,R^{j-1})_\tau=n I(L_1;R|L_2,Q)_\sigma,
\end{align*}
where the state $\sigma$ is defined below 
\begin{equation}\label{eq:sigma_L_1L_2RX}
    \sigma^{L_1L_2RQ\hat{X}}\deq \hspace{-5pt}
     \sum_{l_1,l_2}\frac{ \ketbra{l_1,l_2}}{n}\tensor 
\Big(\bigoplus_{j=1}^n    \big(\tr_{R_{j+1}^nA^nB^n} \Big\{(\id\tensor {\Lambda}^{A}_{l_1}\tensor{\Lambda}^{B}_{l_2})\phirhon\Big\} \tensor  \ketbra{j}\tensor \tr_{\hat{X}^n{\sim j}}\{S_{l_1,l_2}\}\big)\Big),
\end{equation}
where $\tr_{\hat{X}^{n}\sim j}$ denotes tracing over $(\hat{X}^{\tensor j-1}\tensor \hat{X}_{j+1}^{\tensor n})$, and $Q\deq (R^{J-1},J)$, and $J$ is an averaging random variable which is uniformly distributed over $[1,n]$. We elaborate on the Hilbert space associated with $Q$ as follows. 

Suppose $\{\ket{\phi_i}\}_{i\in \mathcal{I}}$ is an orthonormal basis for $\mathcal{H}_R$. Then, a basis for $\mathcal{H}_R^{\tensor k}$ is given by $$\ket{{\phi}_{\mathbf{i}^k}}\deq \ket{\phi_{i_1}} \tensor \ket{\phi_{i_2}}\tensor \cdots \tensor \ket{\phi_{i_k}},$$
for all $\mathbf{i}^k\in \mathcal{I}^k$. Consider the direct-sum of the Hilbert spaces $\bigoplus_{k=1}^n \mathcal{H}_R^{\tensor k}$. Consider the Hilbert space $\mathcal{H}_J\tensor (\bigoplus_{k=1}^n \mathcal{H}_R^{\tensor k})$.
With this definition, define $\mathcal{H}_Q$, as the Hilbert space which is spanned by $\ket{j} \tensor \ket{{\phi}_{\mathbf{i}^{(j-1)}}},$ for all $j\in [1,n]$ and $\mathbf{i}^{(j-1)} \in \mathcal{I}^{(j-1)}$. Therefore, $\mathcal{H}_Q$ is \textit{isometrically isomorphic} to the direct-sum $\bigoplus_k \mathcal{H}_R^{\tensor k}$. Note that $\mathcal{H}_Q$ can be viewed as a multi-particle Hilbert space, which is a truncated version of the so-called Fock space \cite{meyer2006quantum}.

Similarly, for Bob's rate we have
\begin{align*}
R_2+\epsilon \geq I(L_2;R|L_1,Q)_\sigma.
\end{align*} 
For the sum-rate, the following inequalities hold
\begin{align*}
n(R_1+R_2+2\epsilon)&\geq H(L_1,L_2)\geq I(L_1,L_2;R^n)_\tau\\
&=\sum_{j=1}^n  I(L_1, L_2; R_j|R^{j-1})_{\tau}\\
&= nI(L_1,L_2;R|Q)_\sigma. 
\end{align*}
In addition, the distortion of this q-c protocol satisfies $\bar{d}(\rho_{AB}^{\otimes n},  \mathcal{N}_{A^nB^n\mapsto \hat{X}^n})\leq D+\epsilon$, where $\mathcal{N}_{A^nB^n\mapsto \hat{X}^n}$ is the quantum channel associated with the protocol. Therefore, as the distortion observable is symbol-wise additive, we obtain 
\begin{align*}
D+\epsilon &\geq \frac{1}{n}\sum_{j=1}^n \tr\Big\{\left(\Delta_{R_j\hat{X}_j}\otimes I_{R\hat{X}}^{\otimes [n]\backslash j}\right) (\id\tensor \mathcal{N}_{A^nB^n\mapsto \hat{X}^n})(\phirhon)\Big\}\\
&= \frac{1}{n}\sum_{j=1}^n \tr\Big\{\left(\Delta_{R_j\hat{X}_j}\otimes I_{R_1^{j-1}}\tensor I_{R_{j+1}^n \hat{X}^n{\sim j}}\right) (\id\tensor \mathcal{N}_{A^nB^n\mapsto \hat{X}^n})(\phirhon)\Big\}\\
&=\frac{1}{n}\sum_{j=1}^n \tr\Big\{\left(\Delta_{R_j\hat{X}_j}\otimes I_{R_{1}^{j-1}}\right)\left(\tr_{R_{j+1}^n\hat{X}^n{\sim j}}\{ (\id\tensor \mathcal{N}_{A^nB^n\mapsto \hat{X}^n})(\phirhon)\}\right)\Big\}\\
&\stackrel{(a)}{=} \tr\{(\Delta\tensor I_Q)  \sigma^{RQ\hat{X}}\},
\end{align*}
where the third equality holds, because of the following argument. From \eqref{eq:sigma_L_1L_2RX}, one can show by partially tracing over $(L_1, L_2)$, that 
\begin{equation}
   \sigma^{RQ\hat{X}}=\tr_{L_1,L_2}\{\sigma^{L_1L_2RQ\hat{X}}\}=\bigoplus_{j} \frac{1}{n}\ketbra{j}\tensor\tr_{R_{j+1}^n\hat{X}^n{\sim j}}\{(\id\tensor \mathcal{N}_{A^nB^n\mapsto \hat{X}^n})(\phirhon)\}, 
\end{equation}
and $I_Q\deq \bigoplus_{j=1}^n \big(I_R^{\tensor (j-1)} \tensor \ketbra{j}\big)$. Then, $I_Q$ is the identity operator acting on $\mathcal{H}_Q$. Therefore, the right-hand side of the equality $(a)$ above can be written as   
\begin{align*}
 \tr\{(\Delta\tensor I_Q) \sigma^{RQ\hat{X}}\}=\tr\left\{ \Delta \sigma^{R\hat{X}}    \right\}.
\end{align*}

Let us identify the single-letter quantum test channel as given in the statement of the theorem. First, due to the distributive property of tensor product over direct sum operation, we can rewrite  $\sigma^{L_1L_2RQ\hat{X}}$ as
\begin{align*}
      & \sigma^{L_1L_2RQ\hat{X}}=\\\label{eq:sigma_L_1L_2RX}
    & \qquad \Big(\bigoplus_{j=1}^n \frac{1}{n} \sum_{l_1,l_2} \ketbra{l_1,l_2}\tensor 
   \big(\tr_{R_{j+1}^nA^nB^n} \Big\{(\id\tensor {\Lambda}^{A}_{l_1}\tensor{\Lambda}^{B}_{l_2})\phirhon\Big\} \tensor  \ketbra{j}\tensor \tr_{\hat{X}_{\sim j}}\{S_{l_1,l_2}\}\big)\Big).
\end{align*}
Next, we identify a quantum channel $\mathcal{N}_{AB \mapsto L_1L_2Q\hat{X}}: \rho_{AB}\mapsto \sigma^{L_1L_2Q\hat{X}}$. For that and for any $j$ define the following intermediate quantum channels:
\begin{align*}
    \mathcal{N}^{(j)}_{AB\mapsto L_1L_2 R^{(j-1)} \hat{X}}(\omega_{AB})\!\deq\!\!\! \sum_{l_1,l_2}\! \ketbra{l_1,l_2}\tensor  \big(\!\tr_{R_{j+1}^nA^nB^n} \!\Big\{\!(\id_{R^{n}{\sim j}}\!\tensor\! {\Lambda}^{A}_{l_1}\tensor{\Lambda}^{B}_{l_2})(\omega_{AB} \tensor E_j)\!\Big\} \!\tensor \!\tr_{\hat{X}^{n}{\sim j}}\{S_{l_1,l_2}\}\!\big),
\end{align*}
where $E_j=\Psi^{\rho_{AB}}_{(RAB)^{n}_{\sim j}}$. One can verify that $\mathcal{N}^{(j)}_{AB\mapsto L_1L_2 R^{(j-1)} \hat{X}}$ is indeed a quantum channel. With these definitions, let 
\begin{align*}
    \mathcal{N}_{AB \mapsto L_1L_2Q\hat{X}}(\omega_{AB}) \deq \bigoplus_j \frac{1}{n}\left(\mathcal{N}^{(j)}_{AB\mapsto L_1L_2 R^{(j-1)} \hat{X}}(\omega_{AB})\tensor \ketbra{j}\right).
\end{align*}
Using the property of direct-sum operation, one can verify that $\mathcal{N}_{AB \mapsto L_1L_2Q\hat{X}}$ is a valid quantum channel, moreover, 
\[
\sigma^{L_1L_2RQ\hat{X}}=(\id\tensor \mathcal{N}_{AB \mapsto L_1L_2Q\hat{X}})(\phirho).
\]
Lastly, we show that the condition $I(R;Q)_\sigma=0$ is also satisfied. By taking the partial trace of $\sigma$ over $(L_1, L_2, \hat{X})$ we obtain the following state
\begin{align*}
 \sigma^{RQ}&=\tr_{L_1L_2\hat{X}}(\sigma^{L_1L_2RQ\hat{X}})=
\bigoplus_{j=1}^n \frac{1}{n} \sum_{l_1,l_2} 
   \Big(\tr_{R_{j+1}^nA^nB^n} \Big\{(\id\tensor {\Lambda}^{A}_{l_1}\tensor{\Lambda}^{B}_{l_2})\phirhon\Big\} \Big)\tensor \ketbra{j}\\
   &= \bigoplus_{j=1}^n \frac{1}{n}  \Big(\tr_{R_{j+1}^nA^nB^n} \Big\{\phirhon\Big\} \Big)\tensor \ketbra{j}\\
     &= \bigoplus_{j=1}^n \frac{1}{n}  \Big(\tr_{AB}\{\phirho\}\Big)^{\tensor j} \tensor \ketbra{j}\\
     &= \tr_{AB}\{\phirho\} \tensor \left(\bigoplus_{j=1}^n \frac{1}{n}  \Big(\tr_{AB}\{\phirho\}\Big)^{\tensor (j-1)} \tensor \ketbra{j}\right),
\end{align*}
where the last equality is  due to the distributive property of tensor product over direct sum operation. Hence, $\sigma^{RQ}$ is in a tensor product of the form $\sigma^R\tensor \sigma^Q$, and therefore, $I(R;Q)_\sigma=0$. 
\begin{remark}
One may question the computability of the outer bound provided in Theorem \ref{thm:dist str outer}. The computability of this bound depends on the dimensionality of the auxiliary space $\mathcal{H}_Q$ defined in the theorem. Currently, we are unable to bound the dimension of the Hilbert space  $\mathcal{H}_Q$, but aim to provide one in our future work. As a matter of fact, the current outer bounds for the equivalent classical distributed rate distortion problem still suffers from the computability issue. The  first outer bound to the classical problem was provided in \cite{berger} and a recent substantial improvement was made by authors in \cite{wagner2008improved}. Both of these bounds suffer from the absence of cardinality bounds on at least one of the variables used and hence cannot be claimed to be computable using finite resources.
\end{remark}
\end{proof}

\section{Simulation of POVMs with Stochastic Processing} \label{sec:nf_faithfulsimulation}
We now provide an extension of the Winter's point-to-point measurement compression scheme \cite{winter} (discussed in Section \ref{sec:Winter_compression}) with stochastic processing. We assume that the receiver (Bob) has access to additional private randomness, and he 
is allowed to use this additional resource to perform any stochastic mapping of the received classical bits.
In fact, the overall effect on the quantum state can be assumed to be a measurement which is a concatenation of the POVM Alice performs and the stochastic map Bob implements. Hence, Alice in this case, does not remain aware of the measurement outcome. It is for this reason that \cite{wilde_e} describes this as a non-feedback problem, with the sender not required to know the outcomes of the measurement. With the availability of additional resources, 
such a formulation is expected to help reduce the overall resources needed.

\subsection{Problem Formulation}
\begin{definition}
	For a given finite set $\mathcal{X}$, and a Hilbert space $\mathcal{H}_{A}$, a measurement simulation protocol with stochastic processing with parameters $(n,\Theta,N)$ 
	is characterized by 
	\\
	\textbf{1)} a collections of Alice's sub-POVMs  $\tilde{M}^{(\mu)},\mu \in [1,N]$ each acting on $\mathcal{H}_A^{\tensor n}$ and with outcomes in a subset $\mathcal{L}$ satisfying $|\mathcal{L}|\leq \Theta$.\\
	\textbf{2)} a  Bob's classical stochastic map $P^{(\mu)}(x^n|l)$ for all $l\in \mathcal{L}$, $ x^n\in \mathcal{X}^n$ and $\mu \in [1,N]$. \\
	The overall  sub-POVM of this distributed protocol, given by $\tilde{M}$, is characterized by the following operators: 
	\begin{equation}\label{eq:p2p overall POVM opt}
	\tilde{\Lambda}_{x^n}\deq \frac{1}{N}\sum_{\mu, l} P^{(\mu)}(x^n|l)~ \Lambda^{(\mu)}_{l}, \quad \forall x^n \in \mathcal{X}^n,
	\end{equation}
	where $\Lambda^{(\mu)}_{l}$ are the operators corresponding to the sub-POVMs $\tilde{M}^{(\mu)}$.
	%
	%
	%
	%
	%
\end{definition}
In the above definition,  $\Theta$ characterizes the amount of classical bits communicated from Alice to Bob, and the amount of common randomness is determined by  $N$, with $\mu$ being the common randomness bits distributed among the parties. The classical stochastic mappings induced by  $P^{(\mu)}$ represents the action of Bob on the received classical bits.  
\begin{definition}
	Given a POVM $M$ acting on  $\mathcal{H}_{A}$, and a density operator $\rho \in \mathcal{D}(\mathcal{H}_{A})$, a pair $(R,C)$ is said to be achievable, if for all $\epsilon>0$ and for all sufficiently large $n$, there exists a measurement simulation protocol with stochastic processing with parameters  $(n, \Theta, N)$ such that its overall sub-POVM $\tilde{M}$  is $\epsilon$-faithful to $M^{\tensor n}$ with respect to $\rho^{\tensor n}$ (see Definition \ref{def:faith-sim}), and
	\begin{align*}
	\frac{1}{n}\log_2 \Theta \leq R+\epsilon, \quad
	\frac{1}{n}\log_2 N \leq C+\epsilon.
	\end{align*}
The set of all achievable pairs is called the achievable rate region. 
\end{definition}

The following theorem characterizes the achievable rate region.

\begin{theorem}\label{thm:nf_faithfulSim}
For any density operator $\rho\in \mathcal{D}(\mathcal{H}_A)$ and any POVM ${M} \deq \{{\Lambda}_x\}_{x \in \mathcal{X}}$ acting on the Hilbert space $\mathcal{H}_A$, a pair $(R,C)$ is achievable if and only if there exist a POVM $\bar{M}_A \deq \{\bar{\Lambda}_w^{A}\}_{w \in \mathcal{W}} $, with $\mathcal{W}$ being a finite set, and a stochastic map $P_{X|W}:\mathcal{W}\rightarrow\mathcal{X}$ such that
\begin{align*}
    R &\geq I(R;W)_\sigma \quad \mbox{and} \quad R+C \geq I(RX;W)_{\sigma}, \\
    &\Lambda_x = \sum_{w\in\mathcal{W}}P_{X|W}(x|w)\bar{\Lambda}_w^A, \quad \forall x \in \mathcal{X}.
\end{align*}
where  $\sigma_{RWX} \deq \sum_{w,x}\sqrt{\rho}\bar{\Lambda}^A_{w}\sqrt{\rho} \tensor P_{X|W}(x|w)\ketbra{w}\tensor\ketbra{x}.$
\end{theorem}

\begin{remark}
An alternative characterization of the above rate region can also be obtained in terms of Holevo information. For this, we define the following ensemble $\left\{\lambda_{x},\hat{\rho}_x\right\}$ 
as
\begin{align}
    \lambda_{x} = \sum_{w\in \mathcal{W}}\lambda_w^A P_{X|W}(x|w) \quad \mbox{and} \quad \hat{\rho}_x = \sum_{w\in \mathcal{W}} P_{W|X}\hat{\rho}_w^A, \nonumber
\end{align}  
for $\left\{\lambda_{w}^A,\hat{\rho}_w^A\right\}$  being the canonical ensemble  associated with the POVM ${M}$ and the state $\rho$ as defined in \eqref{eq:dist_canonicalEnsemble}. With this ensemble, we have
$$I(R;W)_\sigma = \chi\left(\left\{\lambda_{w}^A,\hat{\rho}_w^A\right\}\right) \quad \mbox{and} \quad I(RX;W)_\sigma = I(X;W)_\sigma + \chi\left(\left\{\lambda_{w}^A,\hat{\rho}_w^A\right\}\right) - \chi\left(\left\{\lambda_{x},\hat{\rho}_x\right\}\right).$$

\end{remark}

\input{nf_faithfulSimProof.tex}

\section{Simulation of Distributed POVMs with Stochastic Processing}\label{sec:nf_dist_POVM}
In this section, we develop a stochastic processing  variant of the distributed POVM simulation problem described in Section \ref{sec:Appx_POVM}. 
Let $\rho_{AB}$ be a density operator acting on a composite Hilbert Space  $\mathcal{H}_A\tensor \mathcal{H}_B$.
Consider two measurements $M_A$ and $M_B$ on sub-systems $A$ and $B$, respectively. 
Imagine again that we have three parties, named Alice, Bob and Eve, that are trying to collectively simulate a given measurement $M_{AB}$
acting on the state $\rho_{AB}$, as shown in Fig. \ref{fig:NF_Distributed}. In this version of distributed simulation, Eve additionally has access to unlimited private randomness.
The problem is defined in the following.
		 
\begin{figure}[hbt]
	\begin{center}
	  \includegraphics[scale=0.4]{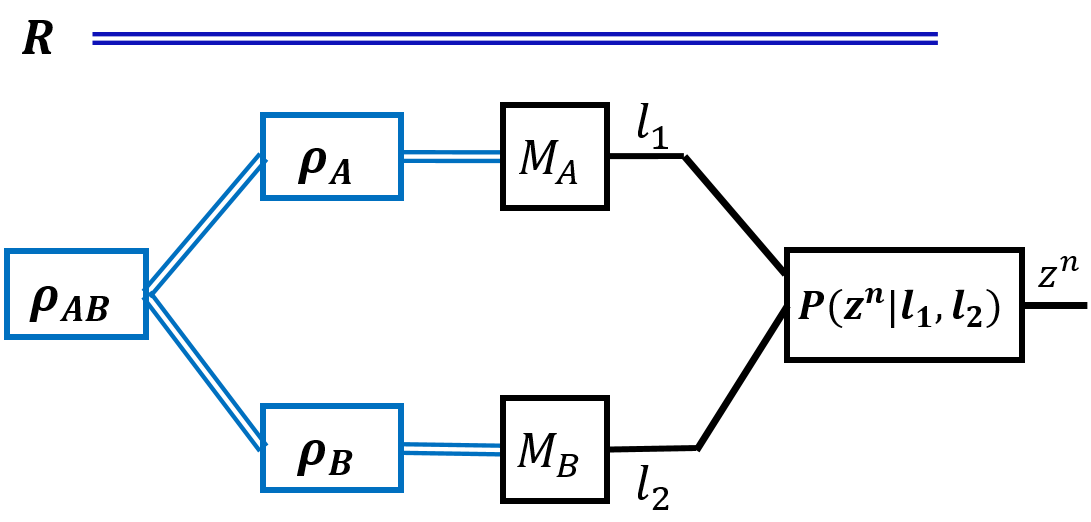}
    \caption{The diagram depicting the distributed POVM simulation problem with stochastic processing. In this setting, Eve additionally has access to unlimited private randomness.}
    \label{fig:NF_Distributed}  
	\end{center}
\end{figure}
\begin{definition}
	For a given finite set $\mathcal{Z}$, and a Hilbert space $\mathcal{H}_{A}\tensor \mathcal{H}_B$, a distributed protocol with stochastic processing with parameters $(n,\Theta_1,\Theta_2,N)$ 
	is characterized by 
	\\
		\textbf{1)} a collections of Alice's sub-POVMs  $\tilde{M}_A^{(\mu)},\mu \in [1, N]$ each acting on $\mathcal{H}_A^{\tensor n}$ and with outcomes in a subset $\mathcal{L}_1$ satisfying $|\mathcal{L}_1|\leq \Theta_1$.\\
		\textbf{2)} a collections of Bob's sub-POVMs  $\tilde{M}_B^{(\mu)},\mu \in [1, N]$ each acting on $\mathcal{H}_B^{\tensor n}$ and with outcomes in a subset $\mathcal{L}_2$, satisfying $|\mathcal{L}_2|\leq \Theta_2$.\\
		\textbf{3)} Eve's   classical stochastic map $P^{(\mu)}(z^n|l_1,l_2)$ for all $l_1\in \mathcal{L}_1, l_2\in \mathcal{L}_2, z^n\in \mathcal{Z}^n$ and $\mu \in [1,N]$. \\
		The overall  sub-POVM of this distributed protocol, given by $\tilde{M}_{AB}$, is characterized by the following operators: 
		\begin{equation}
		\tilde{\Lambda}_{z^n}\deq \frac{1}{N}\sum_{\mu, l_1, l_2} P^{(\mu)}(z^n|l_1,l_2)~ \Lambda^{A,(\mu)}_{l_1}\tensor \Lambda^{B,(\mu)}_{l_2}, \quad \forall z^n \in \mathcal{Z}^n, \nonumber
		\end{equation}
		where $\Lambda^{A,(\mu)}_{l_1}$ and  $\Lambda^{B,(\mu)}_{l_2}$ are the operators corresponding to the sub-POVMs $\tilde{M}_A^{(\mu)}$ and $\tilde{M}_B^{(\mu)}$, respectively.
	%
	%
	%
	%
	%
\end{definition}
In the above definition,  $(\Theta_1,\Theta_2)$ determines the amount of classical bits communicated from Alice and Bob to Eve. The amount of common randomness is determined by  $N$. The classical stochastic maps $P^{(\mu)}(z^n|l_1,l_2)$ represent the action of Eve on the received classical bits.  
\begin{definition}
	Given a POVM $M_{AB}$ acting on  $\mathcal{H}_{A}\tensor \mathcal{H}_B$, and a density operator $\rho_{AB}\in \mathcal{D}(\mathcal{H}_{A}\tensor \mathcal{H}_B)$, a triple $(R_1,R_2,C)$ is said to be achievable, if for all $\epsilon>0$ and for all sufficiently large $n$, there exists a distributed protocol with stochastic processing with parameters  $(n, \Theta_1, \Theta_2, N)$ such that its overall sub-POVM $\tilde{M}_{AB}$  is $\epsilon$-faithful to $M_{AB}^{\tensor n}$ with respect to $\rho_{AB}^{\tensor n}$ (see Definition \ref{def:faith-sim}), and
	\begin{align*}
	\frac{1}{n}\log_2 \Theta_i \leq R_i+\epsilon, \quad i=1,2, \quad \mbox{and} \quad
	\frac{1}{n}\log_2 N \leq C+\epsilon.
	\end{align*}
	The set of all achievable triples $(R_1,R_2,C)$ is called the achievable rate region. 
\end{definition}
The following theorem provides an inner bound to the {achievable} rate region, which is proved in Section \ref{appx:nf_proof of thm dist POVM}.
\begin{theorem}\label{thm:nf_dist POVM appx}
		Given a density operator $\rho_{AB}\in \mathcal{D}(\mathcal{H}_{A}\tensor \mathcal{H}_B)$, and a POVM $M_{AB}=\{\Lambda^{AB}_{z}\}_{z \in \mathcal{Z}}$ acting on  $\mathcal{H}_{A}\tensor \mathcal{H}_B$ having a separable decomposition with stochastic integration (as in Definition \ref{def:Joint Measurements}), a triple $(R_1,R_2,C)$ is achievable if the following inequalities are satisfied:
		\begin{subequations}\label{eq:nf_dist POVm appx rates}
			\begin{align}
			R_1 &\geq I(U;RB)_{\sigma_1}-I(U;V)_{\sigma_3},\\
			R_2 &\geq I(V;RA)_{\sigma_2}-I(U;V)_{\sigma_3},\\
			R_1+R_2 &\geq I(U;RB)_{\sigma_1}+I(V;RA)_{\sigma_2}-I(U;V)_{\sigma_3},\label{eq:nfrate3}\\
			R_1 + C &\geq I(U;RZV)_{\sigma_3} - I(U;V)_{\sigma_3},\\
			R_2 + C &\geq I(V;RZ)_{\sigma_3} - I(U;V)_{\sigma_3}, \\
			R_1+R_2+C &\geq I(UV;RZ)_{\sigma_3}, \label{eq:nfrate4}
			\end{align}
		\end{subequations}
		for some decomposition with POVMs $\bar{M}_A=\{\bar{\Lambda}^A_{u}\}_{u \in \mathcal{U}}$ and $\bar{M}_B=\{ \bar{\Lambda}^B_{v}\}_{v \in \mathcal{V}}$ and a stochastic map $P_{Z|U,V}:\mathcal{U}\times \mathcal{V} \rightarrow \mathcal{Z}$,
		where $\Psi^{\rho_{AB}}_{R A B}$ is a purification of $\rho_{AB}$, and the above information quantities are computed for the auxiliary states
$\sigma_{1}^{RUB} \deq (\emph{id}_R\tensor \bar{M}_{A} \tensor \emph{id}_B) ( \Psi^{\rho_{AB}}_{R A B}), \sigma_{2}^{RAV} \deq (\emph{id}_R\tensor \emph{\id}_A \tensor \bar{M}_{B}) ( \Psi^{\rho_{AB}}_{R A B}),$ 
and	$\sigma_3^{RUVZ}   \deq \sum_{u,v,z}\sqrt{\rho_{AB}}\left (\bar{\Lambda}^A_{u}\tensor \bar{\Lambda}^B_{v}\right )\sqrt{\rho_{AB}} \tensor P_{Z|U,V}(z|u,v)\ketbra{u}\tensor\ketbra{v}\tensor\ketbra{z}.$

\end{theorem}

\begin{remark}
An alternative characterization of the above rate region can be obtained in terms of Holevo information. For this, we define the following ensemble $\left\{\lambda_{z},\hat{\rho}_z\right\}$ 
as
\begin{align}
    \lambda_{z} = \sum_{u\in \mathcal{U}}\sum_{v\in \mathcal{V}}\lambda^{AB}_{uv} P_{Z|UV}(x|u,v) \quad \mbox{and} \quad \hat{\rho}_z = \sum_{u\in \mathcal{U}}\sum_{v\in \mathcal{V}} P_{UV|Z}(u,v|z)\hat{\rho}^{AB}_{uv}, \nonumber
\end{align}  
with 
 $\left\{\lambda^{A}_{u},\hat{\rho}^{A}_{u}\right\}$, $\left\{\lambda^{B}_{v},\hat{\rho}^{B}_{v}\right\}$ and 
$\left\{\lambda^{AB}_{uv},\hat{\rho}^{AB}_{uv}\right\}$  being the canonical ensembles defined in \eqref{eq:dist_canonicalEnsemble}, and 
$    P_{UV|Z}(u,v|z) = \lambda^{AB}_{uv}\cdot P_{Z|UV}(z|u,v)/{\lambda_z}$ for all $(u,v,z) \in \mathcal{U}\times\mathcal{V}\times\mathcal{Z}$.
With this ensemble, we have
$I(U;RB)_{\sigma_1} = \chi\left(\left\{\lambda_{u}^{A},\hat{\rho}_u^{A}\right\}\right)$, $I(V;RA)_{\sigma_2} = \chi\left(\left\{\lambda_{v}^{B},\hat{\rho}_v^{B}\right\}\right)$, and  $I(UV;RZ)_{\sigma_3} = I(UV;Z) + \chi\left(\left\{\lambda_{uv}^{AB},\hat{\rho}^{AB}_{uv}\right\}\right) - \chi\left(\left\{\lambda_{z},\hat{\rho}_z\right\}\right).$

\end{remark}

\input{NF_POVM_Appx_Proof.tex}

	\section{Conclusion}\label{sec:conclusion}
	We have developed a distributed measurement compression protocol where we introduced the technique of mutual covering and random binning of distributed measurements. Using these techniques, a set of communication rate-pairs and common randomness rate is characterized for faithful simulation of distributed measurements. We further developed an approach for a distributed quantum-to-classical rate-distortion theory, and provided single-letter inner and outer bounds. As a part of future work, we intend to improve the outer bound by providing a dimensionality bound on the auxiliary Hilbert space involved in the expression. Further, we also desire to improve the achievable rate region by using structured POVMs based on algebraic codes. 
	\newline
	\newline
	\noindent \textbf{{\large Acknowledgement: }} We thank Mark Wilde for his valuable inputs on techniques needed to prove Theorem \ref{thm:nf_faithfulSim} and for referring us to the additional work performed in \cite{berta2014identifying} and \cite{martens1990nonideal}. We are also grateful to Arun Padakandla for his inputs on the classical analogue of the current work \cite{atif2020source}, which was very helpful in developing the proof techniques here.

\appendices

\input{Lemmas_Appx_proof.tex}

\input{Proof_of_Propositions.tex}

\bibliographystyle{IEEEtran}

\bibliography{IEEEabrv,references}
\end{document}

%% file: Math_Template_2019.tex
\usepackage{mathrsfs}
\usepackage{tikz}
\usepackage[utf8]{inputenc}
\usepackage{amsfonts} 
\usepackage{epstopdf}
\usepackage{graphicx}
\usepackage{bbm}
\usepackage{enumerate}
\usepackage{epsf,epsfig,verbatim,amssymb,amsmath,array,cite,amsthm,subfigure,multicol,multirow}  
\usepackage{psfrag,xspace}

\usepackage{bm}
\usepackage{hhline}
\usepackage{mathabx}
\usepackage{physics}
\newcommand\numberthis{\addtocounter{equation}{1}\tag{\theequation}}
\usepackage{xcolor}
\definecolor{darkblue}{rgb}{0.1,0.1,0.8}
\definecolor{brickred}{rgb}{0.8, 0.25, 0.33}
\definecolor{DarkGreen}{rgb}{0,0.6,0}

\newtheorem{theorem}{Theorem}
\newtheorem{lem}{Lemma}

\newtheorem{proposition}{Proposition}

\theoremstyle{definition}
\newtheorem{definition}{Definition}
\newtheorem{example}{Example}

\newtheorem*{prob*}{Problem}

\newtheorem{remark}{Remark}

\allowdisplaybreaks 
\makeatletter
\def\old@comma{,}
\catcode`\,=13
\def,{%
  \ifmmode%
    \old@comma\discretionary{}{}{}%
  \else%
    \old@comma%
  \fi%
}
\makeatother

\usepackage{etoolbox}
\providetoggle{Blue_revision}
\settoggle{Blue_revision}{true}

\newcommand{\add}[1]{\iftoggle{Blue_revision}{{\color{darkblue}#1}}{#1 }}

 \usepackage{hyperref}
\hypersetup{
colorlinks=true, %
 pdfstartview={FitH},
    linkcolor=red,
    citecolor=blue, 
    urlcolor={blue!80!black}
}

\usepackage{graphicx}

%% file: Quantum_header.tex
\global\long\def\EE{\mathbb{E}}
\global\long\def\PP{\mathbb{P}}

\global\long\def\11{\mathbbm{1}}

\global\long\def\+{\oplus}

\DeclareMathOperator*{\medcap}{\mathbin{\scalebox{1}{\ensuremath{\bigcap}}}}
\DeclareMathOperator*{\medcup}{\mathbin{\scalebox{1}{\ensuremath{\bigcup}}}}

\usepackage{stackengine}
\def\deq{\mathrel{\ensurestackMath{\stackon[1pt]{=}{\scriptstyle\Delta}}}}

\global\long\def\tensor{\otimes}
 \def\tr{\Tr}
 \def\id{\text{id}}
\global\long\def\bigtensor{\bigotimes}
\allowdisplaybreaks

\def\phirho{\Psi_{RAB}^{\rho_{AB}}}
\def\phirhon{\Psi_{R^nA^nB^n}^{\rho_{AB}}}

%% file: introduction2.tex
\newpage 
\section{Introduction}\label{S1}
    \IEEEPARstart{M}easurements interface the intricate quantum world with the perceivable macroscopic classical world by associating a classical attribute to a quantum state. However, quantum phenomena, such as superposition, entanglement and non-commutativity contribute to uncertainty in the measurement outcomes. A key concern, from an information-theoretic standpoint, is to quantify the amount of  ``relevant information" conveyed by a measurement about a quantum state.  
    
Winter's  measurement compression theorem \cite{winter} (also elaborated in \cite{wilde_e}) quantifies the ``relevant information" as the amount of resources needed to simulate the output of a quantum measurement applied on a given state in an asymptotic sense. Imagine that an agent (Alice) performs a measurement $M$ on a quantum state $\rho$, and sends a set of classical bits to a receiver (Bob). Bob intends to \textit{faithfully} recover the outcomes of Alice’s measurements without having access to $\rho$. The measurement compression theorem states that at least  quantum mutual information ($I(X;R)$) amount of classical information and conditional entropy ($S(X|R)$) amount of common shared randomness are needed to obtain a\textit{ faithful simulation}, where $R$ denotes a reference of the quantum state, and $X$ denotes the auxiliary register corresponding to the random measurement outcome.
  Wilde et al. \cite{wilde_e} extended  the measurement compression problem by considering 
additional resources available to each of the participating parties. One such formulation allows Bob to further process the information received from Alice using local private randomness.
In analogy with \cite{bennett2009quantum}, this problem formulation is referred to as non-feedback measurement simulation, while the former is termed as simulation with feedback. This quantified the benefit of private randomness in terms of enhancing the trade-off between classical bits communicated and common random bits consumed.
In particular, the use of private randomness increases the requirement of classical communication bits, while reducing the common randomness constraint.
     
The measurement compression theorem finds applications in several paradigms including local purity distillation \cite{wilde_e} and private classical communication over quantum channels \cite{pcc}. This theorem was later used by Datta, et al. \cite{qtc}   to develop a \ac{qc}
rate-distortion theory. The problem involved lossy compression of a quantum information source into classical bits, with the task of compression performed by applying a measurement on the source. In this problem, the objective is to minimize the storage of the classical outputs resulting from the measurement, while being able to recover the quantum state (from classical bits) within a fixed level of distortion as measured by an observable. To achieve this, the authors in \cite{datta} advocated the use of 
measurement compression protocol, and subsequently characterized the so-called rate-distortion function in terms of single-letter quantum mutual information quantities.  The authors further established that by employing a naive approach of measuring individual output of the quantum source, and then applying Shannon's rate-distortion theory to compress the classical data obtained is insufficient to achieve optimal rates. Further, the problem of measurement compression in the presence of quantum side information was studied in \cite{wilde_e}. The authors here combined the ideas from \cite{winter} and \cite{devetak2003classical} to reduce the classical communication rate and common randomness needed to simulate a measurement in presence of quantum side information. Recently, authors in \cite{anshu2019convex} came up with a completely different technique for analyzing the measurement simulation protocols, while considering the problem of quantum measurement compression with side information. They provide a protocol based on convex-split and position based decoding, and bound rates from above in terms of smooth max and hypothesis testing relative entropies (defined in \cite{anshu2019convex}).

    \begin{figure}[hbt]
	\centering
	{\includegraphics[scale=1.2]{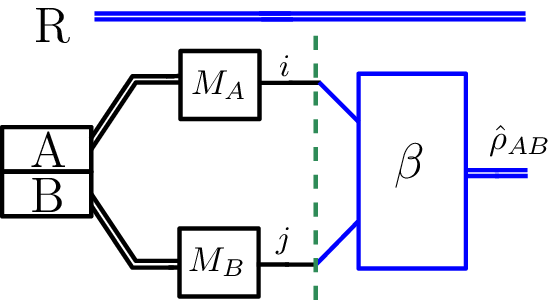}}
	\label{compressionDiagram}
    \caption{The diagram of a distributed quantum measurement applied to a bipartite quantum system $AB$. A tensor product measurement $M_A\tensor M_B$ is performed on many copies of the observed quantum state. The outcomes of the measurements are given by two classical bits. The receiver functions as a classical-to-quantum channel $\beta$ mapping the classical data to a quantum state.} 
    \label{fig:dist_measurements}
    \end{figure}

    In this work, we consider
     scenarios where the quantum measurements are performed in a distributed fashion on bipartite entangled states, and quantify ``relevant information" for these distributed quantum measurements in an asymptotic sense. As shown in Fig. \ref{fig:dist_measurements}, a composite bipartite quantum system $AB$ is made available to two agents, Alice and Bob, where they have access to the sub-systems $A$ and $B$, respectively. Two separate measurements, one for each sub-system, are performed in a distributed fashion with no communication taking place between Alice and Bob. Imagine that there is a third party, Eve, who is connected to Alice and Bob via two separate classical links. The objective of the three parties is to simulate the action of repeated independent measurements performed on many independent copies of the given composite state.
    To achieve this objective, Alice and Bob send classical bits to Eve at rate $R_1$ and $R_2$, respectively. Further, common randomness at rate $C$ is also shared amidst the three parties. Eve performs classical processing of the received bits and common randomness. We study two settings, based on whether or not Eve has access to private randomness. As an application of this quantification, we consider the quantum-to-classical distributed rate distortion problem where Eve is allowed to use classical-to-quantum channels. In this work, we focus on memoryless  quantum systems in finite dimensional Hilbert spaces. 
    We summarize the contributions of this work in the following: 
    \begin{itemize}
        \item We formulate the problem of faithful simulation of distributed quantum measurements  that can be decomposed as a convex-linear combination (incorporating Eve's stochastic processing) of separable measurements, as stated in Definition \ref{def:Joint Measurements}. The asymptotic  performance limit for this problem is given by the set of all communication rates $(R_1, R_2)$ and all common randomness rates $C$, referred to as the achievable rate region, under which the above-stated measurement is distributively simulated. 
        We devise a distributed simulation protocol for this problem, and provide a quantum-information theoretic inner bound to the achievable rate region in terms of computable single-letter information quantities  (see Theorem \ref{thm:nf_dist POVM appx}). This is the first main result of the paper.
        \item As an immediate application of our results on the simulation of distributed measurements, we 
        develop an approach for a distributed quantum-to-classical rate distortion theory, where the objective is to reconstruct a quantum state at Eve, with the quality of reconstruction measured using an additive distortion observable. The asymptotic performance limit is given by the set of all communication rate pairs $(R_1, R_2)$ at which the distortion $D$ is achieved. For the achievability part, we characterize an inner bound
        in terms of single-letter quantum mutual information quantities (see Theorem \ref{thm:q-c dist scr}). This is the second main result of the paper. The classical version of this result is called the Berger-Tung inner bound \cite{berger}. 
        
        \item We then develop a technique for deriving converse bounds  based on a combination of \textit{tensor-product} and \textit{direct-sum} Hilbert spaces (also referred to as  a multi-particle system). Using this technique, we derive a single-letter outer-bound on the optimal rate distortion region (see Theorem \ref{thm:dist str outer}), by converting a multi-letter expression into a single-letter expression. This is the third main result of the paper.
\end{itemize}

The organization of the paper is as follows. In Section \ref{sec:prelim}, we set the notation and state requisite definitions. 
    Instead of first presenting the above-stated inner bound to the performance limits of  the problem of simulation of distributed measurements in its full generality, for pedagogical reasons, in Section \ref{sec:Appx_POVM}, we first consider a special case, where the processing at Eve is restricted to a deterministic function, and provide a simple proof based on the application of Winter's measurement compression theorem.  In this setup, we compress each individual measurements $M_A$ and $M_B$, comprising the decomposing of $M_{AB}$. As a result, faithful simulation of $M_A$ is possible when at least $nI(U;R_A)$ classical bits of communication and $n S(U|R_A)$ bits of common randomness are available between Alice and Eve. Similarly, a faithful simulation of $M_B$ is possible with $nI(V;R_B)$ classical bits of communication and  $nS(V|R_B)$ bits of common randomness between Eve and Bob, where $R_A$ and $R_B$ are purifications of the sub-systems $A$ and $B$, respectively, and $U$ and $V$ denote the auxiliary registers corresponding to their measurement outcomes. The challenge here is that the direct use of single-POVM compression theorem for each individual POVMs, $M_A$ and $M_B$, does not necessarily ensure a “distributed” faithful simulation of the overall measurement, $M_{AB}$. To accomplish this,  we develop a Mutual Covering Lemma (see Lemma \ref{lem:mutual covering}), which also helps in converting the information quantities in terms of the reference $R$ of the joint state $\rho_{AB}$.
 \par Further, an interesting aspect about the distributed setting is that one can further reduce the amount of classical communication by exploiting the statistical correlations between Alice's and Bob's measurement outcomes. The challenge here is that the classical outputs of the approximating POVMs (operating on $n$ copies of the state $\rho_{AB}$) are not independent identically distributed (IID) sequences --- rather they are codewords generated from random coding. For this we develop a proposition for mutual packing (Proposition \ref{prop:Lemma for S_3}), that characterizes the binning rates in term of single-letter information quantities.
This issue also arises in classical distributed source coding problem which was addressed by Wyner-Ahlswede-K\"orner \cite{berger} by developing Markov lemma and Mutual packing lemma. 
The idea of binning in quantum setting has been explored from a different perspective in \cite{sideInf} and \cite{Quantum_Slepian_Wolf} for quantum data compression involving side information.
Toward the end of the section, we also provide an example to illustrate the inner bound to the achievable rate region. 

In Section \ref{sec:QC_dist scr coding}, we apply this special setting of the distributed measurement simulation with deterministic processing to the q-c distributed rate distortion problem. Since, the proof of the inner bound of this rate distortion problem requires only the special case of distributed measurement simulation, this is another reason for providing the special case in the previous section. 

In Section \ref{sec:nf_faithfulsimulation}, we consider the 
non-feedback measurement compression problem for the point-to-point setting. The authors in \cite{wilde_e} have discussed this formulation and provided a rate region with a proof of achievability and converse. 
However, the assumed equations  (53)  and (54) in proving the direct part (see \cite{wilde_e}) do not appear to be true, to the best to our knowledge, but only in an average sense. A stronger version of this theorem is also developed in \cite{berta2014identifying} using a different technique, wherein the authors have extended the Winter's measurement compression for fixed independent and identically distributed inputs \cite{winter} to arbitrary inputs.  Since the result is crucial for the distributed simulation problem with stochastic processing, to be described in the next section (Section \ref{sec:nf_dist_POVM}), we formally state the problem and provide an alternative proof of the direct part for completeness (see Theorem \ref{thm:nf_faithfulSim}).

Finally, the above proof of non-feedback simulation in the point-to-point setting provides us with necessary tools for the next task, namely, distributed quantum measurement simulation with stochastic processing. 
The objective of incorporating the additional processing at the decoder is to reduce the required shared randomness. Our objective in the distributed problem, considered in Section \ref{sec:Appx_POVM}, was to simulate $M_A\tensor M_B$. We achieve this by proving that a pair of POVMs that can faithfully simulate $M_A$ and $M_B$ individually, can also faithfully simulate $M_A\tensor M_B$ (Lemma \ref{lem:mutual covering}). However, it will be shown that, because of the presence of Eve's stochastic processing, decoupling the current problem into two symmetric point-to-point problems is not feasible. Therefore, we perform a non-symmetric partitioning  while being analytically tractable. Moreover, we provide a single-letter achievable inner bound that is symmetric with respect to Alice and Bob. We conclude the paper with a few remarks in Section \ref{sec:conclusion}.


%% file: POVM_Appx.tex
	\section{Simulation of Distributed POVMs with Deterministic Processing}\label{sec:Appx_POVM}
		Now, we develop an extension of Winter's measurement compression \cite{winter} to quantum measurements performed in a distributed fashion with deterministic processing.
    Consider a bipartite composite quantum system $(A,B)$ represented by Hilbert Space $\mathcal{H}_A\tensor \mathcal{H}_B$.  Let $\rho_{AB}$ be a density operator on
    $\mathcal{H}_A\tensor \mathcal{H}_B$. Consider two measurements $M_A$ and $M_B$ on sub-systems $A$ and $B$, respectively.   Imagine that three parties, named Alice, Bob and Eve, are trying to collectively simulate these two measurements, one applied to each sub-system. The three parties share some amount of common randomness.  Alice and Bob perform a measurement $\tilde{M}_A^{(n)}$ and $\tilde{M}_B^{(n)}$ on $n$ copies of sub-systems $A$ and $B$, respectively. 
%
		 The measurements are performed in a distributed fashion with no communication taking place between Alice and Bob. Based on their respective measurements and the common randomness, Alice and Bob send some classical bits to Eve. 
		 Upon receiving these classical bits, Eve applies a processing operation on them and then wishes to produce an $n$-letter classical sequence. The objective is to construct $n$-letter measurements $\tilde{M}_A^{(n)}$ and $\tilde{M}_B^{(n)}$ that minimize the classical communication and common randomness bits while ensuring that the overall measurement induced by the action of the three parties is close to $M^{\tensor n}_A \tensor M^{\tensor n}_B$.
%
%
%
The problem is formally defined in the following.
\begin{definition}
For a given finite set $\mathcal{Z}$, and a Hilbert space $\mathcal{H}_{A}\tensor \mathcal{H}_B$, a distributed protocol with parameters $(n,\Theta_1,\Theta_2,N)$ 
is characterized by 
\add{\\
\textbf{1)} a collections of Alice's sub-POVMs  $\tilde{M}_A^{(\mu)},\mu \in [1, N]$ each acting on $\mathcal{H}_A^{\tensor n}$ and with outcomes in a subset $\mathcal{L}_1$ satisfying $|\mathcal{L}_1|\leq \Theta_1$.\\
\textbf{2)} a collections of Bob's sub-POVMs  $\tilde{M}_B^{(\mu)},\mu \in [1, N]$ each acting on $\mathcal{H}_B^{\tensor n}$ and with outcomes in a subset $\mathcal{L}_2$, satisfying $|\mathcal{L}_2|\leq \Theta_2$.\\
\textbf{3)} a collection of Eve's decoding maps $ f^{(\mu)}:\mathcal{L}_{1}\times \mathcal{L}_{2} \rightarrow \mathcal{Z}^n $ for  $\mu \in [1,N]$. \\
The overall  sub-POVM of this distributed protocol, given by $\tilde{M}_{AB}$, is characterized by the following operators: 
\begin{equation}\label{eq:overall POVM opt}
\tilde{\Lambda}_{z^n}\deq \frac{1}{N}\sum_{\mu, l_1, l_2} \mathbbm{1}_{\{f^{(\mu)}(l_1,l_2)=z^n\}} \Lambda^{A,(\mu)}_{l_1}\tensor \Lambda^{B,(\mu)}_{l_2}, \quad \forall z^n \in \mathcal{Z}^n,
\end{equation}
where $\Lambda^{A,(\mu)}_{l_1}$ and  $\Lambda^{B,(\mu)}_{l_2}$ are the operators corresponding to the sub-POVMs $\tilde{M}_A^{(\mu)}$ and $\tilde{M}_B^{(\mu)}$, respectively.}
%
%
%
%
%
\end{definition}
In the above definition,  $(\Theta_1,\Theta_2)$ determines the amount of classical bits communicated from Alice and Bob to Eve. The amount of common randomness is characterized by  $N$, and $\mu$ can be viewed as the common randomness bits distributed among the parties. The mapping $ f^{(\mu)} $ represents the action of Eve on the received classical bits. 
\begin{definition}
Given a POVM $M_{AB}\delequal \{\Lambda^{AB}_{z}\}_{z\in \mathcal{Z}}$ acting on  $\mathcal{H}_{A}\tensor \mathcal{H}_B$,  and a density operator $\rho_{AB}\in \mathcal{D}(\mathcal{H}_{A}\tensor \mathcal{H}_B)$, a triplet $(R_1,R_2,C)$ is said to be achievable, if for all $\epsilon>0$ and for all sufficiently large $n$, there exists a distributed protocol with parameters  $(n, \Theta_1, \Theta_2, N)$ such that its overall {sub-}POVM $\tilde{M}_{AB}$  is $\epsilon$-faithful to $M_{AB}^{\tensor n}$ with respect to $\rho_{AB}^{\tensor n}$ (see Definition \ref{def:faith-sim}), and
\begin{align*}
    \frac{1}{n}\log_2 N &\leq C+\epsilon,\qquad 
    \frac{1}{n}\log_2 \Theta_i \leq R_i+\epsilon, \quad i=1,2.
\end{align*}
The set of all achievable triples $(R_1, R_2, C)$ is called the achievable rate region. 
\end{definition}
The following theorem provides an inner bound to the achievable rate region, which is proved in Section \ref{appx:proof of thm dist POVM},
\add{\begin{theorem}\label{thm: dist POVM appx}
 Given a density operator $\rho_{AB}\in \mathcal{D}(\mathcal{H}_{A}\tensor \mathcal{H}_B)$ and a POVM $M_{AB} \delequal \{\Lambda^{AB}_{z}\}_{z\in \mathcal{Z}}$ acting on  $\mathcal{H}_{A}\tensor \mathcal{H}_B$ having a separable decomposition with deterministic integration (as in Definition \ref{def:Joint Measurements}), a triple $(R_1,R_2,C)$ is achievable if the following inequalities are satisfied:
\begin{subequations}\label{eq:dist POVm appx rates}
	\begin{align}
	R_1 &\geq I(U;RB)_{\sigma_1}-I(U;V)_{\sigma_3},\label{eq:rate1}\\
	R_2 &\geq I(V;RA)_{\sigma_2}-I(U;V)_{\sigma_3},\label{eq:rate2}\\
	R_1+R_2 &\geq I(U;RB)_{\sigma_1}+I(V;RA)_{\sigma_2}-I(U;V)_{\sigma_3},\label{eq:rate3}\\
	R_1+C &\geq S(U|V)_{\sigma_3},\\
R_2+C &\geq S(V|U)_{\sigma_3},\\
	{R_1+R_2+C} &\geq S(U,V)_{\sigma_3},\label{eq:rate4}  
	\end{align}
\end{subequations}
for some decomposition
with POVMs $\bar{M}_A=\{\bar{\Lambda}^A_{u}\}_{u \in \mathcal{U}}$ and $\bar{M}_B=\{ \bar{\Lambda}^B_{v}\}_{v \in \mathcal{V}}$ and a function $ g: \mathcal{U}\times \mathcal{V} \rightarrow \mathcal{Z}$, where the information quantities are computed for the auxiliary states 
$\sigma_{1}^{RUB} \deq (\emph{\id}_R\tensor \bar{M}_{A} \tensor \emph{\id}_B) ( \Psi^{\rho_{AB}}_{R A B}),~ \sigma_{2}^{RAV} \deq (\emph{\id}_R\tensor \emph{\id}_A \tensor \bar{M}_{B}) ( \Psi^{\rho_{AB}}_{R A B}),$ and $\sigma_{3}^{RUV} \deq (\emph{\id}_R\tensor \bar{M}_{A} \tensor \bar{M}_{B}) ( \Psi^{\rho_{AB}}_{R A B})$, with $\Psi^{\rho_{AB}}_{R A B} $ being a purification of $\rho_{AB}$, and $\mathcal{U}$ and $\mathcal{V}$ are some finite sets.
\end{theorem}}

\begin{remark}
An alternative characterization of the above rate region can be obtained in terms of Holevo information.
For this, we use the canonical ensembles $\left\{\lambda^{A}_{u},\hat{\rho}^{A}_{u}\right\}$, $\left\{\lambda^{B}_{v},\hat{\rho}^{B}_{v}\right\}$ and $\left\{\lambda^{AB}_{uv},\hat{\rho}^{AB}_{uv}\right\}$   defined as
\begin{align}\label{eq:dist_canonicalEnsemble}
\lambda^A_u &\deq \tr\{\bar{\Lambda}^A_u \rho_A\}, \hspace{1.3cm} \lambda^B_v \deq \tr\{\bar{\Lambda}^B_v \rho_B\}, \hspace{1.4cm} \lambda^{AB}_{uv} \deq \tr\{(\bar{\Lambda}^A_u \tensor \bar{\Lambda}^B_v) \rho_{AB}\},\nonumber \\
\hat{\rho}^A_u &\deq \frac{1}{\lambda^A_u}\sqrt{\rho_A}  \bar{\Lambda}^A_u \sqrt{\rho_A}, ~\quad \hat{\rho}^B_v \deq \frac{1}{\lambda^Y_v}\sqrt{\rho_B}\bar{\Lambda}^B_v \sqrt{\rho_B}, \quad \hat{\rho}^{AB}_{uv} \deq \frac{1}{\lambda^{AB}_{uv}}\sqrt{\rho_{AB}} (\bar{\Lambda}^A_u \tensor \bar{\Lambda}^B_v) \sqrt{\rho_{AB}}.
\end{align}
Using this, we get $$I(U;RB)_{\sigma_1} = \chi\left(\left\{\lambda_{u}^{A},\hat{\rho}_u^{A}\right\}\right)\quad \mbox{and}\quad I(V;RA)_{\sigma_2} = \chi\left(\left\{\lambda_{v}^{B},\hat{\rho}_v^{B}\right\}\right). $$ Also, $I(U;V)_{\sigma_3}$, and  $S(U,V)_{\sigma_3}$ are equal to the classical mutual information and joint entropy with respect to the joint distribution $\{\lambda^{AB}_{uv}\}$, respectively.
\end{remark}

Before providing a proof in the next section, we briefly discuss two corner points of the rate region with respect to the common randomness available. Firstly, consider the regime where
the sum rate $(R_1+R_2)$ is at its  minimum achievable, i.e., equation \eqref{eq:rate3} is active. This requires the largest amount of common randomness, given by the constraint $C\geq S(U|RB)_{\sigma_1}+S(V|RA)_{\sigma_2}$. Next, let us consider the regime where $C=0$. This implies $R_1+R_2\geq S(U,V)_{\sigma_3}$. This regime corresponds to the quantum measurement $M_A\tensor M_B$ followed by classical Slepian-Wolf compression \cite{slepian}.  Fig. \ref{fig:rateRegion} demonstrates the achievable rate region in these cases. 


\begin{figure}[hbt]
	\centering
	\hspace{-10pt}
	{\includegraphics[scale=0.4]{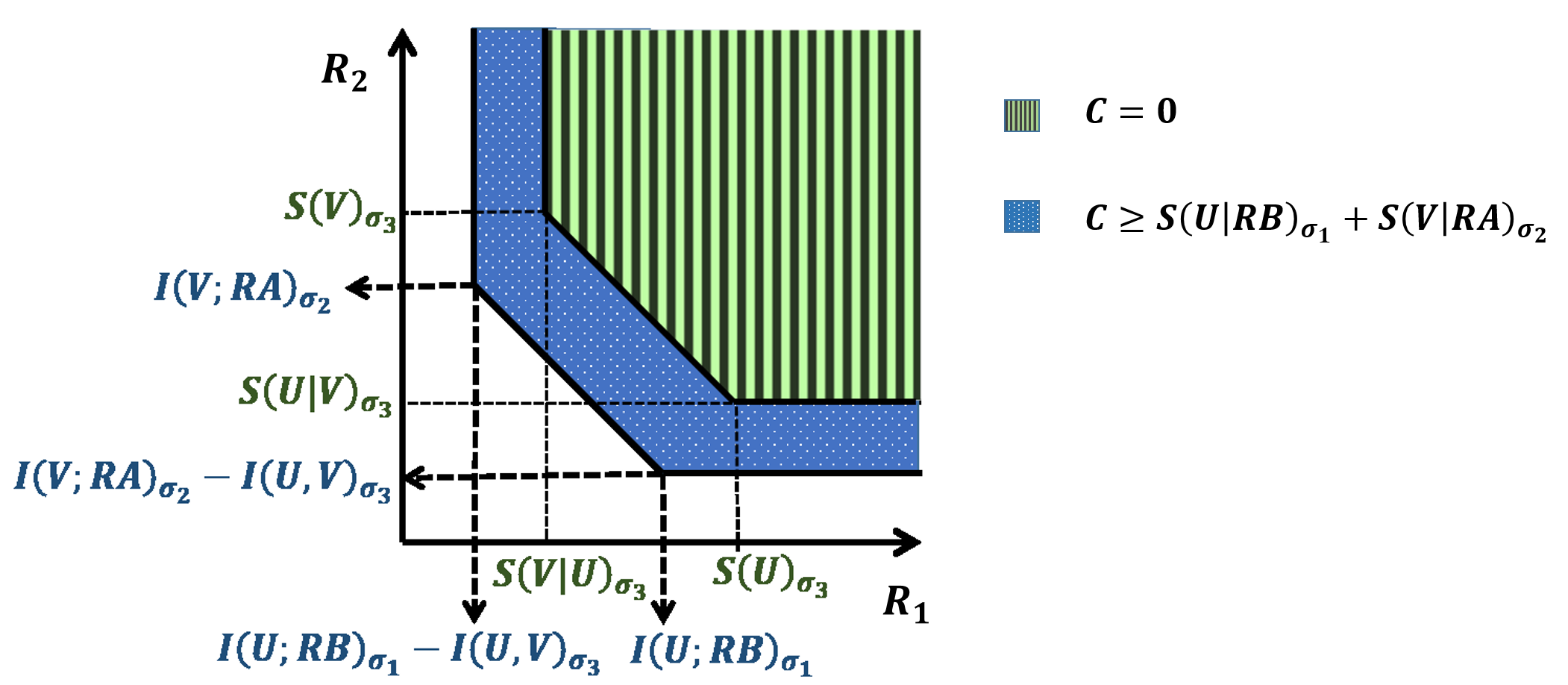}}
	
    \caption{ The inner bound to the achievable rate region given in Theorem \ref{thm: dist POVM appx} at two planes:  1) with no common randomness, i.e., $C=0$ (green color), and  2) with at least $S(U|RB)_{\sigma_1}+S(V|RA)_{\sigma_2}$ amount of common randomness (blue color). 
    As a result, the latter region contains the former.} 
\label{fig:rateRegion}
\end{figure}



We encounter two challenges in developing the single-letter inner bound to the achievable rate region as stated in Theorem \ref{thm: dist POVM appx}:
1) The direct use of single-POVM compression theorem, proved using random coding arguments as in \cite{winter}, for each individual POVMs, $M_A$ and $M_B$, does not necessarily ensure a “distributed” faithful simulation for the overall measurement, $M_A \tensor M_B$. This issue is unique to the quantum settings. One of the contributions of this work is to prove this when the two sources $A$ and $B$ are not necessarily independent, i.e., $\rho_{AB} \neq \rho_A \tensor \rho_B$ (see Lemma \ref{lem:mutual covering}).

2) The classical outputs of the approximating POVMs (operating on $n$ copies of the source) are not \textit{independently and identically distributed} (IID) sequences - rather they are codewords generated from random coding. The Slepian-Wolf scheme \cite{slepian} (also referred to as \textit{binning} in the literature) is developed for distributed compression of IID source sequences.  Applicability of such an approach to the problem requires that the classical outputs produced from the two approximating POVMs are jointly typical with high probability. This issue also arises in classical distributed source coding problem which was addressed by Wyner-Ahlswede-Korner by developing Markov Lemma and Mutual Packing Lemma (Lemma 12.1 and 12.2 in \cite{elGamal}). Building upon these ideas, we develop quantum-classical counterparts of these lemmas for the multi-user quantum measurement simulation problem (see the discussion in Section \ref{sec:POVM binning} and Proposition \ref{prop:Lemma for S_3}). Let us consider an example to illustrate the above inner bound.
\begin{example}
Suppose the composite state $\rho_{AB} $ is described using one of the Bell states on $\mathcal{H}_A\tensor\mathcal{H}_B$ as
\begin{align*}
    \rho^{AB} = \cfrac{1}{2}\left(\ket{00}_{AB} + \ket{11}_{AB} \right)\left( \bra{00}_{AB} + \bra{11}_{AB}\right).
\end{align*}
Since $\pi^A = \Tr_B{\rho^{AB}}$ and $\pi^B = \Tr_A{\rho^{AB}}$, Alice and Bob would perceive each of their particles in maximally mixed states $\pi^A = \frac{I^A}{2}$ and $\pi^B = \frac{I^B}{2}$, respectively. Upon receiving the quantum state, the two parties wish to independently measure their states, using identical POVMs $M_A$ and $M_B$, given by $\left\{ \cfrac{1}{2}\ketbra{0}, \cfrac{1}{2}\ketbra{1}, \cfrac{1}{2}\ketbra{+}, \cfrac{1}{2}\ketbra{-} \right\}$.
Alice and Bob together with Eve are trying to simulate the action of $M_A\tensor M_B$ using the classical communication and common randomness as the resources available to them (as described earlier). 

We compute the constraints given in Theorem \ref{thm: dist POVM appx}. Considering the first constraint from $\eqref{eq:rate1}$,
we evaluate $\sigma_1^{UB}$ as
\begin{align}
    \sigma_1^{UB} = \cfrac{1}{4} \big(\ketbra{0}_U\tensor\ketbra{0}_B + \ketbra{1}_U\tensor\ketbra{1}_B + \ketbra{2}_U\tensor\ketbra{+}_B + \ketbra{3}_U\tensor\ketbra{-}_B \big), \nonumber
\end{align}
where 
the vectors $\{\ket{0}_U,\ket{1}_U,\ket{2}_U,\ket{3}_U\}$ denote a set of orthogonal states on the space $\mathcal{H}_U$. Based on this state, we get
\begin{align*}
    S(\sigma_1^{RUB}) = S(\sigma_1^{UB}) = 2, \quad S(\sigma_1^{RB}) = S(\sigma_1^{B}) = 1 \quad \mbox{and} \quad S(\sigma_1^{U}) = 2.
\end{align*}
This gives $I(U;RB)_{\sigma_1}$ to be equal to 1 bit. Similarly, from the symmetry of the example, we also get $I(V;RA)_{\sigma_2}$ to be equal to 1 bit. Similarly, we can evaluate $\sigma_3^{UV}$ as 
\begin{align*}
 \sigma_3^{UV} = \left(\cfrac{1}{8}\sum_{i=0}^3\ketbra{i}_U\tensor\ketbra{i}_V + \cfrac{1}{16}\sum_{i=0}^3\sum_{j=i+2}^{i+4}\ketbra{i}_U\tensor\ketbra{\hspace{-10pt}\mod(j,4)}_V\right),   
\end{align*}
which gives
\begin{align*}
    S(U,V)_{\sigma_3} = 3.5 \quad \mbox{and} \quad I(U;V)_{\sigma_3} = 0.5.
\end{align*}
Therefore, we can write the constraints given in Theorem \ref{thm: dist POVM appx} as
\begin{align*}
    R_1 \geq 0.5, \quad R_2 \geq 0.5, \quad R_1 + R_2 \geq 1.5, \quad R_1+C \geq 1.5, \quad R_2 +C \geq 1.5 \quad \mbox{and} \quad R_1 +R_2 + C \geq 3.5.
\end{align*}
  Consider the case when $C\geq 2$ is available. By approximating $M_A$ and $M_B$ individually, we receive a gain of 1 bit, decreasing the rate from $S(U)_{\sigma_1}=2$ bits to $I(U;RB)_{\sigma_1}=1$ bit and similarly from $S(V)_{\sigma_2}=2$ bits  to $I(V;RA)_{\sigma_2}=1$ bit. Binning of these approximating POVMs (as discussed in Section (\ref{sec:POVM binning})), gives an additional gain of half a bit, which is characterized by $I(U;V)_{\sigma_3}=0.5$, thus giving us the achievable sum-rate of $1.5$ bits.
   
\end{example}

In the next section we provide the proof for this inner bound.

%% file: POVM_Appx_Proof.tex
\section{Proof of Theorem \ref{thm: dist POVM appx}}\addtocontents{toc}{\protect\setcounter{tocdepth}{1}}
\label{appx:proof of thm dist POVM}
Assume that the operators of the original POVM $M_{AB}$ are decomposed as 
\begin{equation}\label{eq:LambdaAB decompos}
\Lambda^{AB}_{z}=\sum_{u,v} \mathbbm{1}_{\{g(u,v)=z\}} \bar{\Lambda}^A_{u}\tensor \bar{\Lambda}^B_{v}, ~ \forall z \in \mathcal{Z}, 
\end{equation}
for some POVMs   $\bar{M}_{A}$ and $\bar{M}_B$ with operators denoted by $\{\bar{\Lambda}^A_u\}_{u\in \mathcal{U}}$ and $\{\bar{\Lambda}^B_v\}_{v\in \mathcal{V}}$, respectively, and for some function $ g:\mathcal{U}\times \mathcal{V} \rightarrow \mathcal{Z} $ where $\mathcal{U}, \mathcal{V}$ and $ \mathcal{Z} $ are three finite sets. The proof follows by constructing a protocol for faithful simulation of $\bar{M}^{\tensor n}_A\tensor \bar{M}^{\tensor n}_B$. We start by generating the canonical ensembles corresponding to $\bar{M}_{A}$ and $\bar{M}_B$, 
as given in \eqref{eq:dist_canonicalEnsemble}. 
With this notation, corresponding to each of the probability distributions, we can associate a $\delta$-typical set. Let us denote $\mathcal{T}_{\delta}^{(n)}(U)$, $\mathcal{T}_{\delta}^{(n)}(V)$ and $\mathcal{T}_{\delta}^{(n)}(UV)$ as the $\delta$-typical sets defined for $\{\lambda^{A}_{u}\}$, $\{\lambda^{B}_{v}\}$ and $\{\lambda^{AB}_{uv}\}$, respectively.

Let $\PiA$ and $\PiB$ denote the $\delta$-typical projectors (as in \cite{holevo}) for marginal density operators $\rho_A$ and $\rho_B$, respectively. Also, for any $u^n\in \mathcal{U}^n$ and $v^n \in \mathcal{V}^n$, let $\PiuA$ and $\PivB$ denote the conditional typical projectors (as in \cite{holevo}) for the canonical ensembles $\{\lambda^A_u, \hat{\rho}^A_u\}$ and $\{\lambda^B_v, \hat{\rho}^B_v\}$, respectively. For each $u^n\in \mathcal{U}^n$ and $v^n \in \mathcal{V}^n$ define 
\begin{equation}\label{eq:Lambdau'^n v^n}
\Lambda^{A'}_{u^n} \deq \PiA \PiuA \rhohatuA \PiuA \PiA, \quad 
\Lambda^{B'}_{v^n} \deq \PiB \PivB \rhohatvB \PivB \PiB,
\end{equation}
where $\rhohatuA \deq \bigotimes_{i} \hat{\rho}^A_{u_i}$ and $\rhohatvB \deq \bigotimes_{i} \hat{\rho}^B_{v_i}$ \footnote{Note that $\Lambda^{A'}_{u^n}$ and $\Lambda^{B'}_{v^n}$ are not tensor products operators.}.  

With the notation above, define $\sigma^{A'}$ and $\sigma^{B'}$as
\begin{align}
\sigma^{A'}\deq \sum_{u^n \in \mathcal{T}_{\delta}^{(n)}(U)}\cfrac{\lambdauA}{(1-\varepsilon)}\Lambda^{A'}_{u^n} \quad \mbox{and} \quad  \sigma^{B'} \deq \sum_{v^n \in \mathcal{T}_{\delta}^{(n)}(V)}\cfrac{\lambdavB}{(1-\varepsilon')}\Lambda^{B'}_{v^n},
\end{align} 
where $\varepsilon = \sum_{u^n \in \mathcal{T}_{\delta}^{(n)}(U)}\lambdauA$ and $\varepsilon' = \sum_{v^n \in \mathcal{T}_{\delta}^{(n)}(V)}\lambdavB$. Note that $\sigma^{A'}$ and $\sigma^{B'}$ defined above are expectations with respect to the pruned distribution \cite{Wilde_book}.
Let $\PihatA$ and $\PihatB$ be the projectors onto the subspaces spanned by the eigen-states of $\sigma^{A'}$ and $\sigma^{B'}$ corresponding to eigenvalues that are larger than $\epsilon 2^{-n(S(\rho_A)+\delta_1)}$ and $\epsilon 2^{-n(S(\rho_B)+\delta_2)}$,  for some $\delta_1,\delta_2 \geq 0$.  Lastly, define\footnote{Note that $\Lambda^{A}_{u^n}$ and $\Lambda^{B}_{v^n}$ are not tensor products operators.}.
\begin{align}\label{eq:Lambda u^n}
\LambdauA \deq \PihatA\Lambda^{A'}_{u^n}\PihatA,\quad \mbox{and} \quad \LambdavB \deq \PihatB\Lambda^{B'}_{v^n}\PihatB.
\end{align}
\subsection{Construction of Random POVMs}\label{sec:POVM construction}
In what follows, we construct two random POVMs one for each encoder. Fix a positive integer $N$ and positive real numbers $\tilde{R}_1$ and $\tilde{R}_2$ satisfying $\tilde{R}_1 < S(U)_{\sigma_3}$ and $\tilde{R}_2< S(V)_{\sigma_3}$, where ${\sigma_3}$ is defined as
\begin{align*}
\sigma_{3}^{RUV} \deq (\id^R\tensor \bar{M}_{A} \tensor \bar{M}_{B}) ( \Psi^{\rho_{AB}}_{RAB}), 
\end{align*} with $ \Psi^{\rho_{AB}}_{RAB} $ being any purification of $ \rho_{AB}$\footnote{The information theoretic quantities calculated with respect to $ \sigma_{3}^{RUV} $ remain independent of the purification used in its definition.}. 
{Let $ \mu_1 \in [1,N_1]$ denote the common randomness shared between the first encoder and the decoder, and let $ \mu_2 \in [1,N_2]$ denote the common randomness shared between the second encoder and the decoder, with log($ N_1 $) + log($ N_2 $) $ \leq $ log($ N $)}.
For each $\mu_1 \in [1, N_1]$ and $\mu_2 \in [1, N_2]$,
 randomly and independently select $2^{n\tilde{R}_1}$ and $2^{n\tilde{R}_2}$ sequences $(U^{n,(\mu_1)}(l), V^{n,(\mu_2)}(k))$ according to the pruned distributions, i.e.,
 \begin{align}\label{def:prunedDist}
     \PP\left((U^{n,(\mu_1)}(l), V^{n,(\mu_2)}(k)) = (u^n,v^n)\right) = \left\{\begin{array}{cc}
          \dfrac{\lambdauA}{(1-\varepsilon)}\dfrac{\lambdavB}{(1-\varepsilon')} \quad & \mbox{for} \quad u^n \in \mathcal{T}_{\delta}^{(n)}(U), v^n \in \mathcal{T}_{\delta}^{(n)}(V)\\
           0 & \quad \mbox{otherwise}
     \end{array} \right. .
 \end{align} 
Construct operators
\begin{align}\label{eq:A_uB_v}
A^{(\mu_1)}_{u^n} \deq  \gamma^{(\mu_1)}_{u^n} \bigg(\sqrt{\rho_{A}}^{-1}\LambdauA\sqrt{\rho_{A}}^{-1}\bigg)\quad  \text{ and  } \quad
B^{(\mu_2)}_{v^n} \deq \zeta^{(\mu_2)}_{v^n} \bigg(\sqrt{\rho_{B}}^{-1}\LambdavB\sqrt{\rho_{B}}^{-1}\bigg),
\end{align}
where 
\begin{align}\label{eq:gamma_mu}
 \gamma^{(\mu_1)}_{u^n}\deq  \frac{1-\varepsilon}{1+\eta}2^{-n\tilde{R}_1}|\{l: U^{n,(\mu_1)}(l)=u^n\}|\quad  \text{ and  } \quad
 \zeta^{(\mu_2)}_{v^n}\deq  \frac{1-\varepsilon'}{1+\eta}2^{-n\tilde{R}_2}|\{k: V^{n,(\mu_2)}(k)=v^n\}|, 
\end{align}
where $\eta \in (0,1)$ is a parameter to be determined.
 Then, for each $\mu_{1} \in [1,N_{1}]$ and $ \mu_{2} \in [1,N_{2}]$, construct $M_1^{( n, \mu_1)}$ and $M_2^{( n, \mu_2)}$ as in the following
\begin{align}
   M_1^{( n, \mu_1)}& \deq \{A^{(\mu_1)}_{u^n} : u^n \in  \mathcal{T}_{\delta}^{(n)}(U)\}, \quad  M_2^{(n, \mu_2)} \deq \{B^{(\mu_2)}_{v^n} : v^n \in  \mathcal{T}_{\delta}^{(n)}(V)\}.
\label{eq:POVM-14}
\end{align} 

We show in the later part of the proof (Lemma \ref{lem:M_XM_Y POVM}) that $M_1^{(n,\mu_1)}$ and $M_2^{(n,\mu_2)}$ form sub-POVMs, with high probability, for all $\mu \in [1,N_1]$ and $\mu_2 \in [1,N_2]$, respectively. These collections $ M_1^{( n, \mu_1)}$ and $ M_2^{( n, \mu_2)}$ are completed using the operators $I - \sum_{u^n \in \mathcal{T}_{\delta}^{(n)}(U)}A^{(\mu_1)}_{u^n}$ and $I - \sum_{v^n\in \TDeltan(V)}B^{(\mu_2)}_{v^n}$, and these operators are associated with sequences $u^n_0$ and $v^n_0$, which are chosen arbitrarily from $\mathcal{U}^n\backslash\mathcal{T}_{\delta}^{(n)}(U) $ and $\mathcal{V}^n\backslash \mathcal{T}_{\delta}^{(n)}(V)$, respectively.

\subsection{Binning of POVMs}\label{sec:POVM binning}
We introduce the quantum counterpart of the so-called \textit{binning} technique which has been widely used in the context of classical distributed source coding. Fix binning rates $(R_1, R_2)$ and choose a $(\mu_1,\mu_2)$ pair. For each sequence $u^n\in \mathcal{T}_{\delta}^{(n)}(U)$ assign an index  from $[1,2^{nR_1}]$ randomly and uniformly, such that the assignments for different sequences are done independently. Perform a similar random and independent assignment for all $v^n\in \mathcal{T}_{\delta}^{(n)}(V)$ with indices chosen from  $[1,2^{nR_2}]$. Repeat this assignment for every $\mu_1 \in [1,2^{nC_1}]$ and  $\mu_2 \in [1,2^{nC_2}]$. 
%
  For each $i\in [1,2^{nR_1}]$ and $j\in [1,2^{nR_2}]$, let $\mathcal{B}^{(\mu_1)}_1(i)$ and $\mathcal{B}^{(\mu_2)}_2(j)$ denote the $i^{th}$ and the  $j^{th}$ bins, respectively. More precisely, $\mathcal{B}^{(\mu_1)}_1(i)$ is the set of all $ u^n$ sequences with assigned index equal to $ i$, and similar is $ \mathcal{B}^{(\mu_2)}_2(j)$. 
  Define the following operators:
\begin{align*}
\Gamma^{A, ( \mu_1)}_i \deq \sum_{u^n \in \mathcal{B}^{(\mu_1)}_1(i)}A^{(\mu_1)}_{u^n}, \qquad \text{and} \qquad \Gamma^{B,( \mu)}_j \deq \sum_{v^n \in \mathcal{B}^{(\mu_2)}_2(j)}B^{(\mu_2)}_{v^n},
\end{align*}
for all  $i\in [1,2^{nR_1}]$ and $j\in [1,2^{nR_2}]$.  
Using these operators, we form the following collection:
\begin{align}
	    M_A^{( n, \mu_1)} \deq  \{\Gamma^{A,  (\mu_1)}_i\}_{ i \in [1, 2^{nR_1}] }, \quad 
	    M_B^{( n, \mu)}\deq  \{\Gamma^{B, ( \mu_2)}_j \}_{j \in [1, 2^{nR_2}] }.
	    \label{eq:POVM-16}
\end{align}
Note that if  $M_1^{( n, \mu_1)}$ and $M_2^{( n, \mu_2)}$ are sub-POVMs, then so are $M_A^{( n, \mu_1)}$ and $M_B^{( n, \mu_2)}$. This is due to the  relations
\begin{equation*}
\sum_{i} \Gamma^{A, ( \mu_1)}_i=\sum_{u^n \in \mathcal{T}_{\delta}^{(n)}(U)}A^{(\mu_1)}_{u^n},\quad \text{and} \quad
\sum_{j} \Gamma^{B, ( \mu_2)}_j=\sum_{v^n\in \mathcal{T}_{\delta}^{(n)}(V)}B^{(\mu_2)}_{v^n}.
\end{equation*}
To make $ M_A^{( n, \mu_1)} $ and $M_B^{( n, \mu_2)}$ complete, we define $\Gamma^{A, ( \mu_1)}_0$ and $ \Gamma^{B, ( \mu_2)}_0 $ as $\Gamma^{A, ( \mu_1)}_0=I-\sum_i \Gamma^{A, ( \mu_1)}_i$ and $\Gamma^{B, ( \mu_2)}_0=I-\sum_j \Gamma^{B, ( \mu_2)}_j$, respectively\footnote{Note that $\Gamma^{A, ( \mu_1)}_0=I-\sum_i \Gamma^{A, ( \mu_1)}_i = I - \sum_{u^n \in T_{\delta}^{(n)}(U)}A^{(\mu_1)}_{u^n}$ and $\Gamma^{B, ( \mu_2)}_0=I-\sum_j \Gamma^{B, ( \mu_2)}_j = I - \sum_{v^n \in T_{
\delta}^{(n)}(V)}B^{(\mu_2)}_{v^n}$}. Now, we intend to use the completions $[M_A^{( n, \mu_1)}]$ and $[M_B^{( n, \mu_2)}]$ as the POVMs for each encoder.
Also, note that the effect of the binning is in reducing the communication rates from $(\tilde{R}_1, \tilde{R}_2)$ to $(R_1,R_2)$. 
  
\subsection{Decoder mapping }
Note that the operators $A_{u^n}^{(\mu_1)}\tensor B_{v^n}^{(\mu_2)}$ are used to simulate $\bar{M}_A\tensor \bar{M}_B$. Binning can be viewed as partitioning of the set of classical outcomes into bins. Suppose an outcome $(U^n,V^n)$ occurred after the measurement. Then, if the bins are small enough, one might be able to recover the outcomes by knowing the bin numbers. For that we create a decoder that takes as an input a pair of bin numbers and produces a pair of sequences  $(U^n,V^n)$. More precisely, we define a mapping $F^{(\mu)}$, for $ \mu = (\mu_1,\mu_2)$, acting on the outputs of $[M_A^{( n, \mu_1)}] \tensor [M_B^{( n, \mu_2)}]$ as follows. Let $\mathcal{C}^{(\mu)}$ denote the codebook containing all  pairs of  codewords  $(U^{n,(\mu_1)}(l),V^{n,(\mu_2)}(k))$. 
{On observing $ \mu $ and the classical indices $ (i,j) \in [1:2^{nR_1}]\times [1:2^{nR_2}] $ communicated by the encoder, the decoder first deduces $ (\mu_1,\mu_2) $ from $ \mu $} and then populates, 
\begin{equation}
D^{(\mu_1,\mu_2)}_{i,j} \deq \left \{(u^n, v^n)\in \mathcal{C}^{(\mu)}: (u^n,v^n) \in \mathcal{T}_{\delta}^{(n)}(UV) \text{ and } (u^n, v^n) \in \mathcal{B}_1^{(\mu_1)}(i)\times \mathcal{B}_2^{(\mu_2)}(j)\right \}.
%
%
\label{eq:POVM-17}
\end{equation}
For every $ \mu \in [1:N] $, $i\in [1:2^{nR_1}]$ and $j\in [1, 2^{nR_2}]$ define the function $F^{(\mu)}(i,j)=(u^n, v^n)$ if  $(u^n,v^n)$ is the only element of $D^{(\mu_1,\mu_2)}_{i,j}$; otherwise $F^{(\mu)}(i,j)=(u_0^n, v_0^n)$
Further, $F^{(\mu)}(i,j)=(u_0^n, v_0^n)$ for $i=0$ or $j=0$. Finally, the decoder produces $ z^n \in \mathcal{Z}^n $ according to the map $ g^n(F^{(\mu)}(i,j))$. 
With this mapping, we form the following collection of operators, denoted by $\tilde{M}_{AB}^{(n)} $,
\begin{equation*}
\tilde{\Lambda}_{u^n,v^n}^{AB} \delequal \frac{1}{N_{1}N_{2}}\sum_{\mu_{1}=1}^{N_{1}}\sum_{\mu_{2}=1}^{N_{2}}\sum_{\substack{(i,j): F^{(\mu)}(i,j)=(u^n,v^n)}}\hspace{-10pt} \Gamma^{A, ( \mu_1)}_i\tensor \Gamma^{B, ( \mu_2)}_j, \qquad \forall (u^n,v^n) \in \mathcal{U}^n\times \mathcal{V}^n.
\end{equation*}
Note that for $\tilde{\Lambda}_{u^n,v^n}^{AB} = 0$ for $(u^n,v^n) \notin (\TDeltan(U)\times\TDeltan(V))\medcup\{(u^n_0,v^n_0)\}.$
We show that $\tilde{M}_{AB}^{(n)} $ is a POVM that is $\epsilon$-faithful to the intermediate POVM $\bar{M}^{\tensor n}_A\tensor \bar{M}^{\tensor n}_B$, with respect to $\rho_{AB}^{\tensor n}$. For faithful simulation of the original POVM $M_{AB}$, we apply the deterministic mapping $g^n(u^n,v^n)$ to the classical outputs of $\tilde{M}_{AB}^{(n)}$. More precisely, we construct the POVM $\hat{M}_{AB}^{(n)}$ with the following operators: 
$$\hat{\Lambda}^{AB}_{z^n}=\sum_{u^n,v^n}  \tilde{\Lambda}_{u^n,v^n}^{AB}\mathbbm{1}_{\{g^n(u^n,v^n)=z^n\}}, ~ \forall z^n\in \mathcal{Z}^n .$$
\subsection{Analysis of POVM and Trace Distance}
\label{sec:traceDistance_POVM}
\add{In what follows, we show that $\hat{M}_{AB}^{(n)}$ is a POVM, and is $\epsilon$-faithful with respect to $\rho_{AB}$ (according to Definition \ref{def:faith-sim}) to $M_{AB}$, where $\epsilon>0$ can be made arbitrarily small for sufficiently large $ n$. More precisely, we show that, with probability sufficiently close to 1,
\begin{equation}\label{eq:faithful M hat}
\sum_{z^n} \|\sqrt{\rho_{AB}^{\tensor n}} \left(\Lambda_{z^n}^{AB}-\hat{\Lambda}_{z^n}^{AB}\right) \sqrt{\rho_{AB}^{\tensor n}}\|_1\leq \epsilon.
\end{equation}
According to the decomposition of $\Lambda_{z}$, given in \eqref{eq:LambdaAB decompos}, the above inequality is equivalent to 
\begin{align*}
\sum_{z^n} \norm{\sum_{u^n,v^n} \mathbbm{1}_{\{g(u^n,v^n)=z^n\}} \left( \sqrt{\rho_{AB}^{\tensor n}} (\bar{\Lambda}^A_{u^n}\tensor \bar{\Lambda}^B_{v^n}-\tilde{\Lambda}_{u^n, v^n}^{AB}) \sqrt{\rho_{AB}^{\tensor n}}\right)}_1\leq \epsilon.
\end{align*}
From triangle inequality, the left-hand side of the above inequality does not exceed the following 
\begin{align*}
\sum_{z^n} & \sum_{u^n,v^n} \mathbbm{1}_{\{g(u^n,v^n)=z^n\}} \norm{ \sqrt{\rho_{AB}^{\tensor n}} (\bar{\Lambda}^A_{u^n}\tensor \bar{\Lambda}^B_{v^n}-\tilde{\Lambda}_{u^n, v^n}^{AB}) \sqrt{\rho_{AB}^{\tensor n}}}_1=\sum_{u^n,v^n} \norm{ \sqrt{\rho_{AB}^{\tensor n}} (\bar{\Lambda}^A_{u^n}\tensor \bar{\Lambda}^B_{v^n}-\tilde{\Lambda}_{u^n, v^n}^{AB}) \sqrt{\rho_{AB}^{\tensor n}}}_1.
\end{align*}
Hence, it is sufficient to show that the above quantity is no greater than $\epsilon$, with probability sufficiently close to 1. This is equivalent to showing that $\tilde{M}^{(n)}_{AB}$ is $\epsilon$-faithful to $\bar{M}^{\tensor n}_A\tensor \bar{M}^{\tensor n}_B$ with respect to $\rho^{\tensor n}_{AB}$. 
Alternatively, using Lemma 1, we prove the following inequality 
\begin{align}\label{eq:actual trace dist}
    G \deq\left\| (\text{id}\tensor \bar{M}_A^{\tensor n}\tensor \bar{M}_B^{\tensor n}) (\Psi^\rho_{R^nA^nB^n})- (\text{id}\tensor \tilde{M}_{AB}^{(n)} ) (\Psi^\rho_{R^nA^nB^n}) \right\|_1\leq \epsilon.
\end{align}}
We characterize the conditions on $(n,N, R_1, R_2)$ under which the inequality given in \eqref{eq:actual trace dist} holds, using the following steps.\\
\noindent \textbf{Step 1: $M_1^{(n,\mu_1)}$ and $M_2^{(n,\mu_2)}$ are sub-POVMs and individually approximating}\\
As a first step, one can show that with probability sufficiently close to one, $M_1^{( n, \mu_1)}$ and $M_2^{( n, \mu_2)}$ form sub-POVMs for all $\mu_1\in [1,N_1]$ and $\mu_2\in [1,N_2]$, and also individually approximate the corresponding tensor product POVMs. More precisely the following Lemma holds. 

\begin{lem}\label{lem:M_XM_Y POVM}
For any two positive integers $N_1$ and $ N_2 $, and $\epsilon, \varepsilon, \varepsilon', \eta \in (0,1)$, as in \eqref{eq:gamma_mu}, and any $\zeta\in (0,1)$, there exists $n(\epsilon, \varepsilon, \varepsilon', \eta,\zeta)$ such that for all $n\geq n(\epsilon, \varepsilon, \varepsilon', \eta,\zeta)$, the collection of operators $ M_1^{( n, \mu_1)}$ and $ M_2^{( n, \mu_2)}$ form sub-POVMs for all $\mu_1 \in [1,N_1]$ and $\mu_2 \in [1,N_2]$ with probability at least $(1-\zeta)$, provided that 
\begin{equation*}
\tilde{R}_1> I(U;RB)_{\sigma_1}, \quad \text{and} \quad \tilde{R}_2> I(V;RA)_{\sigma_2},
\end{equation*} 
where $\sigma_1, \sigma_2$ are defined as in the statement of the theorem. 
In addition, if
\begin{align}
    \frac{1}{n}\log_2 N_1 +\tilde{R}_1> S(U)_{\sigma_1},~~~
    \frac{1}{n}\log_2 N_2 +\tilde{R}_2> S(V)_{\sigma_2},
\end{align}
then with probability at least $(1-\zeta)$ the collection of average operators $M_i^{( n)} \deq\frac{1}{N_i}\sum_{\mu_i} [M_i^{( n, \mu_i)}], i=1,2$  are $\epsilon$-faithful to $M_A^{\tensor n}$ with respect to $\rho_A^{\tensor n}$ and $M_B^{\tensor n}$ with respect to $\rho_B^{\tensor n}$, respectively. 
\end{lem}
\begin{proof}
The proof uses a similar argument as in that of Theorem 2 in \cite{winter}. Hence it is omitted. 
\end{proof}
As a result of the lemma, $\tilde{M}_{AB}^{(n)}$ and $\hat{M}_{AB}^{(n)}$ are valid POVMs with high probability. 

\noindent \textbf{Step 2: Isolating the effect of  un-binned approximating measurements}\\
In this step, we separate out the effect of un-binned approximating measurements from $G$ in \eqref{eq:actual trace dist}. This is done by adding and subtracting an appropriate term within the trace norm and applying triangle inequality, which bounds $G$ as $G\leq S_1 +S_2$, where 
\begin{align}
    & S_1 \deq    \left \| (\text{id}\tensor \bar{M}_A^{\tensor n}\tensor \bar{M}_B^{\tensor n}) (\Psi^\rho_{R^nA^nB^n})- \frac{1}{N_1N_2}\sum_{\mu_1,\mu_2}(\text{id}\tensor  [M_1^{( n, \mu_1)}]\tensor [M_2^{( n, \mu_2)}]) (\Psi^\rho_{R^nA^nB^n})\right \|_1,\nonumber \\
    &S_2 \deq \left \| \frac{1}{N_1N_2}\sum_{\mu_1,\mu_2}(\text{id}\tensor  [M_1^{( n, \mu_1)}]\tensor [M_2^{( n, \mu_2)}]) (\Psi^\rho_{R^nA^nB^n})-(\text{id}\tensor \tilde{M}_{AB}^{(n)} ) (\Psi^\rho_{R^nA^nB^n}) \right \|_1, \label{eq:trace-distance}
\end{align}
{where the $S_1$ captures the effect of using approximating POVMs $ M_1^{( n, \mu_1)} $ and $ M_2^{( n, \mu_2)} $ instead of the actual POVMs $ \bar{M}_A^{\tensor n} $ and $ \bar{M}_B^{\tensor n} $, while $S_2$ captures the error introduced by binning these approximating POVMs.} 
Before we proceed further, we provide the following lemma which will be useful in the rest of the paper.
\begin{lem}\label{lem:Separate}
	Given a density operator $ \rho_{AB} \in \mathcal{D}(\mathcal{H}_{AB}) $, a sub-POVM
$M_Y \deq \left \{\Lambda_y^B: y \in \mathcal{Y}\right \} $ acting on $ \mathcal{H}_B,  $ for some set $\mathcal{Y}$, 
and any Hermitian operator $\Gamma^A$ acting on $ \mathcal{H}_A $, we have
	\begin{align}\label{eq:lemSeparate1}
	\sum_{y \in \mathcal{Y}}\left\| \sqrt{\rho_{AB}}\left (\Gamma^A\tensor \Lambda_y^B\right )\sqrt{\rho_{AB}}\right \|_1 \leq \left \| \sqrt{\rho_A} \Gamma^A\sqrt{\rho_{A}} \right \|_1,
	\end{align} with equality if $ \displaystyle \sum_{y \in \mathcal{Y}}\Lambda_y^B = I  $, where $ \rho_A = \Tr_{B}\{\rho_{AB}\}$.
\end{lem}
\begin{proof}
 The proof is provided in Appendix \ref{appx:proofLemmaSeparate}. 
\end{proof}
Next, we show $S_1$ is sufficiently small using the following Mutual Covering Lemma.
\begin{lem}[Mutual Covering Lemma]\label{lem:mutual covering}
Suppose the sub-POVM $\hat{M}_X$ is $\epsilon$-faithful to $M_X$ with respect to $\rho_X$, and the sub-POVM $\hat{M}_Y$ is $\epsilon$-faithful to ${M}_Y$ with respect to $\rho_Y$, where $\rho_X = \Tr_Y{\{\rho_{XY}\}}$ and $\rho_Y = \Tr_X{\{\rho_{XY}\}}$. Then the sub-POVM $\hat{M}_X\tensor \hat{M}_Y$ is $2\epsilon$-faithful to the POVM ${M}_X \tensor {M}_Y$ with respect to $\rho_{XY}$.
\end{lem}
\begin{proof}
The proof is provided in the Appendix \ref{appx:proof of mutual covering}.
\end{proof}
Using Lemma \ref{lem:mutual covering} with $\rho_{XY} = \rho_{AB}^{\tensor n}$, $\hat{M}_X = \frac{1}{N_1}\sum_{\mu_1} [M_1^{( n, \mu_1)}]$, $\hat{M}_Y = \frac{1}{N_2}\sum_{\mu_2} [M_2^{( n, \mu_2)}]$, $M_X = \bar{M}_A^{\tensor n}$ and $M_Y = \bar{M}_B^{\tensor n}$, and Lemma \ref{lem:M_XM_Y POVM},   
 with probability at least $(1-\zeta)$, we have $S_1\leq 2 \epsilon$. 
 

\noindent\textbf{Step 3: Analyzing the effect of Binning}\\
In this step, we provide an upper bound on $S_2$.
For $(u^n, v^n) \in \mathcal{B}^{(\mu_1)}_1(i)\times \mathcal{B}^{(\mu_2)}_2(j)$
, define $e^{(\mu)}(u^n,v^n)\deq F^{(\mu)}(i,j)$. For any $(u^n, v^n) \notin \mathcal{C}^{(\mu)}$ define $e^{(\mu)}(u^n,v^n)=(u_0, v_0)$. Note that $e^{(\mu)}$ captures the overall effect of the binning followed by the decoding function $F^{(\mu)}$. For all $u^n\in \mathcal{U}^n$ and $v^n \in \mathcal{V}^n$, let $\Phi_{u^n,v^n}\deq\ketbra{u^n,v^n}$.  
With this notation, we simplify $S_2$ using the following proposition.
\begin{proposition}\label{prop:Lemma for Simplification1}
$S_2$ can be simplified as
\begin{align*}
    S_2= \frac{1}{N_1N_2}\sum_{\mu_1,\mu_2}  
\sum_{u^n \in \TDeltan(U)} \sum_{v^n \in \TDeltan(V)}  \big\|\Phi_{u^n, v^n} -\Phi_{e^{(\mu)}(u^n,v^n)} \big\|_1 \gamma^{(\mu_1)}_{u^n}\zeta^{(\mu_2)}_{v^n}\Omega_{u^n,v^n}, 
\end{align*}
where $\Omega_{u^n,v^n}$ is defined as
\begin{align}
    \Omega_{u^n,v^n}\deq \tr\Big\{ \sqrt{\rho^{\tensor n}_A\tensor \rho^{\tensor n}_B}^{-1}(\LambdauA\tensor\LambdavB)  \sqrt{\rho^{\tensor n}_A\tensor \rho^{\tensor n}_B}^{-1} \rho^{\tensor n}_{AB}\Big\}.\nonumber
\end{align}
\end{proposition}
\begin{proof}
The proof is provided in Appendix \ref{appx:proof for Simplification1}.
\end{proof}
In the next proposition we provide a bound on $S_2$. 
\begin{proposition}[Mutual Packing]\label{prop:Lemma for S_3}
There exist functions  $\epsilon_{S_2}(\delta) $ and $\delta_{{S}_{2}}(\delta)$, such that for  all sufficiently small $\delta$ and sufficiently large $n$, we have $\PP\left({S}_2 >\epsilon_{{S_2}}(\delta) \right)$ $\leq \delta_{{S}_2}(\delta)$, if $ \tilde{R}_{1} + \tilde{R}_{2} - R_1 -R_2 < I(U;V)_{\sigma_{2}}$,  where $\sigma_3$ is the auxiliary state defined in the theorem and $\epsilon_{{S}_2},\delta_{{S}_2}  \searrow 0$ as $\delta \searrow 0$.
\end{proposition}
\begin{proof}
The proof is provided in Appendix \ref{appx:proof of S_3}.
\end{proof}
 Using this, and from Step 2, with probability sufficiently close to one, we have $$G \leq S_1 + S_2 \leq 2\epsilon+ {\epsilon_{ S_2}}. $$ 
\subsection{Rate Constraints}
To sum-up, we showed that the trace distance inequality in \eqref{eq:actual trace dist} holds for sufficiently large $n$ and with probability sufficiently close to 1, if the following bounds hold:
\begin{subequations}\label{eq:rate-region 1}
\begin{align}
\tilde{R}_1& \geq I(U; RB)_{\sigma_1}\\
\tilde{R}_2& \geq I(V; RA)_{\sigma_2}\\
C_{1}+\tilde{R}_1  \geq S(U)_{\sigma_3},& \quad C_{2}+\tilde{R}_2 \geq S(V)_{\sigma_3}, \\
(\tilde{R}_1-R_1)+(\tilde{R}_2&-R_2)  < I(U;V)_{\sigma_3}\\
\tilde{R}_1\geq R_1\geq 0, &\quad \tilde{R}_2\geq R_2 \geq 0, \\
C_{1} +C_{2} \leq C, \quad & C_{1}\geq 0,\quad C_{2}\geq 0,
\end{align}
\end{subequations}
where $C_{i}\deq \frac{1}{n}\log_2 N_{i}; $  for $ i=1,2$. This implies that $\tilde{M}_{AB}^{(n)}$ is $\epsilon$-faithful to $\bar{M}^{\tensor n}_A\tensor \bar{M}^{\tensor n}_B$ with probability sufficiently close to one, and hence, $\hat{M}_{AB}^{(n)}$ is also $\epsilon$-faithful to $M_{AB}^{\tensor n}$ with respect to $\rho_{AB}^{\tensor n}$, i.e, \eqref{eq:faithful M hat} is satisfied. Therefore, there exists a distributed protocol with parameters $(n, 2^{nR_1}, 2^{nR_2}, 2^{nC})$ such that its overall POVM $\hat{M}_{AB}^{(n)}$ is $\epsilon$-faithful to $M_{AB}^{\tensor n}$ with respect to $\rho_{AB}^{\tensor n}$.
Lastly,  we complete the proof of the theorem using the following lemma.
\begin{lem}
Let $\mathcal{R}_1$ denote the set of all $(R_1,R_2,C)$ for which there exists $(\tilde{R}_1,\tilde{R}_2)$ such that the sextuple $(R_1,R_2, C_{1}, C_{2}, \tilde{R}_1, \tilde{R}_2)$ satisfies the inequalities in \eqref{eq:rate-region 1}. Let, $\mathcal{R}_2$ denote the set of all triples $(R_1, R_2, C)$ that satisfies  the inequalities in \eqref{eq:dist POVm appx rates} given in the statement of the theorem. Then, $
\mathcal{R}_1=\mathcal{R}_2$.
\end{lem}
\begin{proof}
The prove follows by Fourier-Motzkin elimination \cite{fourier-motzkin}.
\end{proof}

%% file: nf_faithfulSimProof.tex
\subsection{Proof of Achievability of Theorem \ref{thm:nf_faithfulSim}}
\label{app:nf_faithfulSim}
Suppose there exist a POVM $\bar{M}_A$ and a stochastic map $ P_{X|W}:\mathcal{W\rightarrow \mathcal{X}} $,  such that 
${M}$ can be decomposed as 
\begin{equation}\label{eq:Lambda decompos}
{\Lambda}_{x}=\sum_{w} P_{X|W}(x|w) \bar{\Lambda}_{w}^A, ~ \forall x \in \mathcal{X}.
\end{equation}
We begin by defining a canonical ensemble corresponding to $\bar{M}_A$ 
as $\{\lambda_w^A, \hat{\rho}_w^A\}_{w\in \mathcal{W}},$  
where 
\begin{align}\label{def:rhohat}
\lambda_w^A \delequal \tr\{\bar{\Lambda}_w^A \rho\}, \quad \text{and} \quad 
\hat{\rho}_w^A \delequal \frac{1}{\lambda_w^A}\sqrt{\rho}  \bar{\Lambda}_w^A \sqrt{\rho}.
\end{align}
%
Similarly. for each $w^n\in \mathcal{W}^n$ , we also define
\begin{equation}
\Lambda_{w^n}^{A} \deq \hat{\Pi}\Pi_{\rho} \Pi_{w^n}  \hat{\rho}_{w^n}^A \Pi_{w^n} \Pi_{\rho}\hat{\Pi}, \nonumber
\end{equation}
where $\hat{\rho}_{w^n}^A \deq \bigotimes_{i} \hat{\rho}_{w_i}^A$, and $ \hat{\Pi}, \Pi_{\rho} $ and $ \Pi_{w^n} $ are similar to the ones defined in Section \ref{appx:proof of thm dist POVM}.
Using the above definitions, we now construct the approximating POVM.


\subsubsection{Construction of Random POVMs}
In what follows, we construct a collection of random POVMs. 
Fix $ R $ and $ C $ as two positive integers. 
Let $ \mu \in [1,2^{nC}]$ denote the common randomness shared between the sender and receiver. For each $\mu \in [1,2^{nC}]$, randomly and independently select $2^{n{R}}$ sequences 
$W^{n,(\mu)}(l)$
according to the pruned distributions, i.e.,
 \begin{align}\label{def:WprunedDist}
     \PP\left((W^{n,(\mu)}(l)) = (w^n)\right) = \left\{\begin{array}{cc}
          \dfrac{\lambda_{w^n}^A}{(1-\varepsilon)}\quad & \mbox{for} \quad w^n \in \mathcal{T}_{\delta}^{(n)}(W)\\
           0 & \quad \mbox{otherwise}
     \end{array} \right. .
 \end{align} 
For $w^n \in \mathcal{T}_{\delta}^{(n)}(W)$, let the operators of the POVM $\tilde{M}^{(\mu)}_A$ be $\{A^{(\mu)}_{w^n}; w^n \in \mathcal{W}^n\}$ for each $\mu \in [1,2^{nC}]$, where $A^{(\mu)}_{w^n}$ is defined as
\begin{align}\label{eq:nfA_uB_v}
A^{(\mu)}_{w^n} \deq  \gamma^{(\mu)}_{w^n} \bigg(\sqrt{\rho_{A}}^{-1}\Lambda_{w^n}^A\sqrt{\rho_{A}}^{-1}\bigg)\quad  \text{ and } \quad \gamma^{(\mu)}_{w^n} =
\cfrac{1}{2^{nR}}\sum_{l=1}^{2^{nR}}\cfrac{(1-\varepsilon)}{(1+\eta)}\mathbbm{1}_{\{W^{n,(\mu)}(l)=w^n\}}, 
\end{align}
with $\eta \in (0,1)$  being a parameter to be determined. Now, for each $\mu \in [1,2^{nC}]$ construct $\tilde{M}^{(n,\mu)}_A$ as 
\begin{align}
\tilde{M}^{(n,\mu)}_A = \{A_{w^n}^{(\mu)} : w^n \in \mathcal{T}_{\delta}^{(n)}(W)\}.\nonumber
\end{align}
Since the construction is very similar to the one used in Section \ref{appx:proof of thm dist POVM}, we make a claim similar to the one in Lemma \ref{lem:M_XM_Y POVM}. This claim gives us the first constraint on the classical rate of communication $R$, which ensures that the operators constructed above for all $ \mu \in [1,2^{nC}] $ are valid sub-POVMs (characterized as the event $E_1$) with high probability. The claim is as follows. If $R > I(R;W)_{\sigma}$ then $\PP(E_1 ) \geq (1-\delta_1)$ for some $\delta_1 \in (0,{1}/{6})$; or in other words, with probability sufficiently close to one, $\tilde{M}^{( n, \mu)}_A$ forms a sub-POVM for all $\mu\in [1,2^{nC}]$. Note the definition of $\sigma_{RWX}$ follows from the statement of theorem. From this, let $ [\tilde{M}^{(n,\mu)}_A] $ denote the completion of the corresponding sub-POVM $ \tilde{M}^{(n,\mu)}_A $ for $ \mu \in [1,2^{nC}] $. Let the operators completing these POVMs, given by $ I-\sum_{w^n}A_{w^n}^{(\mu)} $, be denoted by $ A_{w_0^n}^{(\mu)} $ for some $ w_0^n \notin \TDeltan(W), $ for all $ \mu \in [1,2^{nC}]$, and $A_{w^n}^{(\mu)} = 0$ for $w^n \notin \TDeltan(W)\medcup \{w_0^n\}$. Using this construction, we define the intermediate POVM $\tilde{M}^{(n)}_A $ as $ \tilde{M}^{(n)}_A = \dfrac{1}{2^{nC}}\sum_{\mu}\tilde{M}^{(n,\mu)}_A $ and the operators of $  \tilde{M}^{(n)}_A $ as $ \tilde{\Lambda}_{w^n}^A = \dfrac{1}{2^{nC}}\sum_{\mu}A_{w^n}^{(\mu)}. $ Now, we define Bob's stochastic map as $P_{X|W}^n$, yielding the operators of the final approximating POVM  as
\begin{align}
    \sum_{w^n\in \mathcal{W}^n}P^n_{X|W}(x^n|w^n)\tilde{\Lambda}_{w^n}^A, \quad x^n \in \mathcal{X}^n. \nonumber
\end{align}
\subsubsection{Trace Distance}
 Now, we compare the action of this approximating POVM on the input state $ \rho^{\tensor n}$ with that of the given POVM $ M $, using the characterization provided in Definition \ref{def:faith-sim}.
Specifically, we show using the expressions for canonical ensemble that, with probability close to one,
\begin{align} \label{eq:main_lemma}
{G} \deq \sum_{x^n\in \mathcal{X}^n} &\left\|\sum_{w^n\in \mathcal{W}^n}P^n_{X|W}(x^n|w^n)\sqrt{\rho^{\tensor n}} (\bar{\Lambda}_{w^n}^A-\tilde{\Lambda}_{w^n}^A) \sqrt{\rho^{\tensor n}}\right\|_1 \leq \epsilon.
\end{align}
As a first step, we split and bound $G$ as $G \leq S_1 + S_2$, where
\begin{align} 
S_1 & \deq \sum_{x^n}\left\|\sum_{w^n}{\lambda}_{w^n}^A\hat{\rho}_{w^n}^A P^n_{X|W}(x^n|w^n) - \frac{1}{2^{nC}}\sum_{w^n\neq w^n_0}\sum_{\mu =1}^{2^{nC}}\gamma^{(\mu)}_{w^n}{\Lambda}_{w^n}^A P^n_{X|W}(x^n|w^n) \right\|_{1},  \nonumber \\ 
S_2 & \deq \sum_{x^n}\left\|P^n_{X|W}(x^n|w_0^n)\frac{1}{2^{nC}}\sum_{\mu =1}^{2^{nC}}\left[\sqrt{\rho^{\tensor n}} (I - \sum_{w^n}A_{w^n}^{(\mu)}) \sqrt{\rho^{\tensor n}} \right]\right\|_{1}. \nonumber
\end{align}



Now we bound $S_1$ by adding and subtracting an appropriate term and using triangle inequality as $S_1 \leq S_{11} + S_{12}$, where $S_{11} $ and $ S_{12}$ are given by
\begin{align}
S_{11}& \deq \left\|\sum_{x^n} \left[\sum_{w^n}{\lambda}_{w^n}^A\hat{\rho}_{w^n}^A P^n_{X|W}(x^n|w^n)\tensor \ketbra{x^n}    - \frac{1}{2^{nC}}\sum_{w^n\neq w_0^n}\sum_{\mu =1}^{2^{nC}}\gamma^{(\mu)}_{w^n}\hat{\rho}_{w^n}^A P^n_{X|W}(x^n|w^n)\tensor \ketbra{x^n}    \right]\right\|_1 \nonumber, \\
S_{12} & \deq \left\|\sum_{x^n}\sum_{w^n\neq w_0^n} \left[\frac{1}{2^{nC}}\sum_{\mu =1}^{2^{nC}}\gamma^{(\mu)}_{w^n}\hat{\rho}_{w^n}^A P^n_{X|W}(x^n|w^n)  - \frac{1}{2^{nC}}\sum_{\mu =1}^{2^{nC}}\gamma^{(\mu)}_{w^n}{\Lambda}_{w^n}^AP^n_{X|W}(x^n|w^n) \right]\tensor \ketbra{x^n} \right\|_{1} \nonumber. 
\end{align}
Note that in the above expressions, we have used an additional triangle inequality for block operators (which is in fact an equality) to move the summation over $\mathcal{X}^n$ inside the trace norm.
Firstly, we show $S_{11}$ is small with high probability. To simplify the notation, we define $\sigma_{w^n} = \sum_{x^n} P^n_{X|W}(x^n|w^n)  \ketbra{x^n}$ which gives $S_{11}$ as 
\begin{align}
S_{11} &  = \left\|\sum_{w^n}{\lambda}_{w^n}^A\hat{\rho}_{w^n}^A  \tensor \sigma_{w^n} - \cfrac{1}{2^{n(R+C)}}\cfrac{(1-\varepsilon)}{(1+\eta)}\sum_{l,\mu} \hat{\rho}_{W^{n,(\mu)}(l)}^A\tensor\sigma_{W^{n,(\mu)}(l)}\right\|_1. \nonumber 
\end{align}
We develop the following lemma to bound this term.
\begin{lem}\label{lem:nf_SoftCovering}
Consider an ensemble given by $\{\tilde{P}_{W^n}({w^n}),\mathcal{T}_{w^n}\}$, where $ \tilde{P}_{W^n}({w^n}) $ is the pruned distribution as defined in (\ref{def:WprunedDist}) and $  \mathcal{T}_{w^n}$ is any tensor product state of the form  $  \mathcal{T}_{w^n} = \bigtensor_{i=1}^n\mathcal{T}_{w}  $. Then, for any $\epsilon_2 > 0 $, $\delta>0$, there exists functions $\delta_s(\delta,\epsilon_2)$ and $\delta'(\delta)$,  such that  for all sufficiently large $n$, the inequality
\begin{align}
    \left\|\sum_{w^n}\lambda_{w^n}^A\mathcal{T}_{w^n} - \cfrac{1}{2^{n(R+C)}}\cfrac{(1-\varepsilon)}{(1+\eta)}\sum_{l,\mu} \mathcal{T}_{W^{n,(\mu)}(l)}\right\|_1\leq \epsilon_2, \label{eq:nf_SoftCovering_Term}
\end{align}  
holds with probability greater than $1-\delta_s(\delta,\epsilon_2)$, 
if  $R + C \geq S({\sum_{w}\lambda_w^A\mathcal{T}_{w}}) - \sum_{w}\lambda^A_{w} S(\mathcal{T}_{w}) + \delta' = \chi\left(\{\lambda_w^A,\mathcal{T}_w\}\right) + \delta'$, where 
 $\{W^{n,(\mu)}(l): l \in [1,2^{nR}],\mu \in [1,2^{nC}]\}$ are independent random vectors generated according to the pruned distribution given in \eqref{def:WprunedDist}, and $\delta_s \searrow 0$, $\delta' \searrow 0$
 as $\epsilon_2 \searrow 0$,
 $\delta \searrow 0$.
\end{lem}
\begin{proof}
 The proof of the lemma is provided in Appendix \ref{lem:nf_SoftCovering_Proof}
\end{proof}
Therefore, using the lemma above,
$S_{11}$ can be made arbitrarily small, for sufficiently large n, with high probability, given the constraints
$R + C \geq S(\sum_{w}\lambda_w^A\hat{\rho}_{w}^A  \tensor \sigma_{w}) - \sum_{w}\lambda_{w}^A S(\hat{\rho}_{w}^A  \tensor \sigma_{w}) + \delta' = \chi\left(\{\lambda_w^A\},\{\hat{\rho}_{w}^A\tensor\sigma_{w}\}\right) + \delta' = I(RX;W)_{\sigma} + \delta'$.
Secondly, we bound $S_{12}$ by applying expectation and using Gentle Measurement Lemma \cite{Wilde_book} as follows,
\begin{align}
\EE\left[S_{12}\right] & = \EE\Bigg[\left\|\sum_{x^n}\sum_{w^n\neq w_0^n}  \left[\frac{1}{2^{nC}}\sum_{\mu =1}^{M}\gamma^{(\mu)}_{w^n}\hat{\rho}_{w^n}^A P^n_{X|W}(x^n|w^n)  - \frac{1}{2^{nC}}\sum_{\mu =1}^{2^{nC}}\gamma^{(\mu)}_{w^n}{\Lambda_{w^n}^A}P^n_{X|W}(x^n|w^n) \right]\tensor \ketbra{x^n} \right\|_{1}\Bigg] \nonumber \\
& \stackrel{(a)}\leq \frac{1}{2^{nC}}\sum_{\mu =1}^{2^{nC}} \sum_{x^n}\sum_{w^n\neq w_0^n}P^n_{X|W}(x^n|w^n) \EE\left[\gamma^{(\mu)}_{w^n}  \left \| (\hat{\rho}_{w^n}^A -\Lambda_{w^n}^A) \right \|_{1} \right] \nonumber \\
& \stackrel{(b)}= \frac{1}{2^{nC}}\sum_{\mu =1}^{2^{nC}} \sum_{w^n \in \TDeltaN{(W)}}\cfrac{\lambda_{w^n}^A}{(1+\eta)} \left \|\hat{\rho}_{w^n} -\Lambda_{w^n}^A \right \|_{1} \nonumber \\
& \stackrel{(c)}= \cfrac{1}{(1+\eta)} \sum_{w^n \in \TDeltaN{(W)}}\lambda_{w^n}^A \left \|\hat{\rho}_{w^n}^A -\hat{\Pi}\Pi_{\rho}\Pi_{w^n}\hat{\rho}_{w^n}^A\Pi_{w^n}\Pi_{\rho}\hat{\Pi} \right \|_{1} \stackrel{(d)}\leq \cfrac{(1-\varepsilon)}{(1+\eta)} (2\sqrt{\varepsilon'} + 2\sqrt{\varepsilon''}) \deq \varepsilon_3, \label{eq:avg_gentle}
\end{align}
where $(a)$ is obtained by using triangle inequality and the linearity of expectation, $(b)$ is obtained by marginalizing over $x^n$ and using the fact that $\EE[\gamma^{(\mu)}_{w^n}] = \frac{\lambda_{w^n}}{(1+\eta)} $, $(c)$ is obtained by substitution, and finally $(d)$ uses repeated application of the average gentle measurement lemma, by setting $ \varepsilon_3 = \frac{(1-\varepsilon)}{(1+\eta)} (2\sqrt{\varepsilon'} + 2\sqrt{\varepsilon''}) $ with $\varepsilon_3 \searrow 0 $ as $ \delta \searrow 0 $ and, $ \varepsilon' = \varepsilon + 2\sqrt{\varepsilon} $ and $ \varepsilon'' = 2\varepsilon + 2\sqrt{\varepsilon} $ (see (35) in \cite{wilde_e} for details).
Finally, we show that the term corresponding to $S_2$ can also be made arbitrarily small. This term can be simplified as follows
\begin{align}
S_2 & \leq \frac{1}{2^{nC}}\sum_{\mu =1}^{2^{nC}}\sum_{x^n}P^n_{X|W}(x^n|w_0^n)\left\| \sum_{w^n}\lambda_{w^n}^A\hat{\rho}_{w^{n}}^A - \sum_{w^n\neq w_0^n}\sqrt{\rho^{\tensor n}}A_{w^n}^{(\mu)} \sqrt{\rho^{\tensor n}} \right\|_{1}, \nonumber \\
& \leq \frac{1}{2^{nC}}\sum_{\mu =1}^{2^{nC}}\left\| \sum_{w^n}\lambda_{w^n}^A\hat{\rho}_{w^{n}}^A - \sum_{w^n\neq w_0^n}\gamma_{w^n}^{(\mu)}\hat{\rho}_{w^n}^A \right\|_1 + \frac{1}{2^{nC}}\sum_{\mu =1}^{2^{nC}}\sum_{w^n\neq w_0^n}\gamma_{w^n}^{(\mu)}\left\| \hat{\rho}_{w^n}^A -  \Lambda_{w^n}^A \right\|_1 = S_{21} + S_{22}, \nonumber
\end{align}
where
\begin{align}\label{eq:mainTerm0Simplification}
S_{21} & \deq \frac{1}{2^{nC}}\sum_{\mu =1}^{2^{nC}}\left\| \sum_{w^n}\lambda_{w^n}^A\hat{\rho}_{w^{n}}^A - \cfrac{(1-\varepsilon)}{(1+\eta)}\cfrac{1}{2^{nR}}\sum_{l=1}^{2^{nR}}\hat{\rho}_{W^{n,(\mu)}_{l}}^A \right\|_1 \quad \mbox{and} \quad
S_{22} & \deq \frac{1}{2^{nC}}\sum_{\mu =1}^{2^{nC}}\sum_{w^n\neq w_0^n}\gamma_{w^n}^{(\mu)}\left\| \hat{\rho}_{w^n}^A -  \Lambda_{w^n}^A \right\|_1 .
\end{align}
Now, for the first term in \eqref{eq:mainTerm0Simplification} we use 
Lemma \ref{lem:nf_SoftCovering} and claim that for given any $ \epsilon_{\scriptscriptstyle w_0}, \delta_{\scriptscriptstyle w_0} \in (0,1)$, if  
\begin{align}
R > S\left (\sum_{w \in \mathcal{W}}\lambda_w^A\hat{\rho}_{w}\right ) + \sum_{w \in \mathcal{W}}\lambda_w^A S(\hat{\rho}_{w}) = I(R;W)_{\sigma}, \nonumber
\end{align} 
then the probability of this term being greater than $ \epsilon_{\scriptscriptstyle w_0} $ is bounded by $ \delta_{\scriptscriptstyle w_{0}} $ for sufficiently large n, where $ \sigma$ is as defined in the statement of the theorem .
Note that the  requirements we obtain on $ R$  were already imposed when claiming the collection of operators $ A_{w^n}^{(\mu)} $ forms a sub-POVM. As for the second term in \eqref{eq:mainTerm0Simplification} we again use the gentle measurement Lemma 
and bound its expected value as
\begin{align}
\EE\left[\frac{1}{2^{nC}}\sum_{\mu =1}^{2^{nC}}\sum_{w^n}\gamma_{w^n}^{(\mu)}\left\| \hat{\rho}_{w^n} -  \Lambda_{w^n} \right\|_1 \right]  = \sum_{w^n \in \TDeltaN(W)}\cfrac{\lambda_{w^n}}{(1+\eta)}\left\| \hat{\rho}_{w^n} -  \Lambda_{w^n} \right\|_1 \leq \varepsilon_{3},\nonumber
\end{align}
where $ \varepsilon_{3} $ is defined in \eqref{eq:avg_gentle}.

In summary, we have performed the following sequence of steps. Firstly, we argued that $\tilde{M}^{(n,\mu)}_A$ forms a valid sub-POVM for all $ \mu\in[1,2^{nC}]$, with high probability, when the rate $ R $ satisfies $ R > I(R;W)_{\sigma} $. Secondly, we moved onto bounding the trace norm between the states obtained after the action for these approximating POVMs when compared with those obtained from the action of actual POVM $ M $, characterized as $G$ using 
Definition \ref{def:faith-sim}. As a first step in establishing this bound, we showed that $G\leq S_1 + S_2.$
Considering $S_1$, we used triangle inequality and divided it into two terms: $S_{11}$ and $S_{12}$. Then, using Lemma \ref{lem:nf_SoftCovering} we showed that for any given $ \varepsilon_1 \in (0,1) $,  $S_{11}$ can be made smaller than $ \varepsilon_1, $
with high probability if $ R+C > I(RX;W)_{\sigma} $. As for $S_{12}$, we showed that it goes to zero in the expected sense using  (\ref{eq:avg_gentle}). Finally, for the term given by $S_2$, we bounded this as a sum of two trace norms $S_{21}$ and $S_{22}$ given in \eqref{eq:mainTerm0Simplification}. We showed that its first term can be made smaller than $ \epsilon_{\scriptscriptstyle w_0} $ with high probability if $ R > I(R;W) $ for sufficiently large $ n $, and the second term was shown to approach zero in expected sense. 

Now, using Markov inequality we argue the existence of at least one collection of POVMs that satisfies the statement of the Theorem  \ref{thm:nf_faithfulSim} as follows. Note that $S_{12}$ is same as $S_{22}$. Let $E_1$ be the event defined earlier in the proof. Let us define $ E_{2}, E_3 $ and $ E_{4} $ as the random variables corresponding to the terms $S_{11}, S_{12}$ and $S_{21}$, respectively.
Firstly, if $R > I(R;W)_{\sigma}$. then $\PP(E_1 ) \geq (1-\delta_1)$. Secondly, from Lemma \ref{lem:nf_SoftCovering},
for all $0<\epsilon_2<1 $, and 
for all sufficiently large n, if $$R + C \geq I(RX;W)_{\sigma} 
 + \delta', \quad \mbox{and} \quad R \geq I(W;R)_{\sigma} + \delta',$$ then  we have $\PP(E_2 \leq {\epsilon_2}) \geq 1-\delta_s(\epsilon_2)$, $\PP(E_4 \leq {\epsilon_{w_0}}) \geq 1-\delta_{w_0}$.
Thirdly, from (\ref{eq:avg_gentle}) we have  $\EE[E_3]  \leq \varepsilon_3$. This implies, from the Markov inequality, that
\begin{align}
    \PP(E_3 \geq 2\varepsilon_3) \leq \frac{\EE[E_3]}{2\varepsilon_3} \leq \frac{1}{2}.\nonumber
\end{align}
\noindent Using these bounds, we get
\begin{align}\label{eq:probExistence}
\PP\left((E_1 ) \medcap (E_2 < \varepsilon_2) \medcap (E_3<2\varepsilon_3) \medcap (E_4 < \epsilon_{w_0}) \right) & \geq \PP(E_1)+ \PP(E_2 < \varepsilon_2) + \PP(E_3 < 2\varepsilon_3)+ \PP(E_4 < \epsilon_{w_0}) -3 \nonumber \\
& \geq 3-(\delta_1+\delta_s+\delta_{w_0})+\cfrac{1}{2}  -3 > \cfrac{1}{4},
\end{align}
given that we choose $0 < \delta_1,\delta_s,\delta_{w_0}< \frac{1}{12}.$
Note that $G \leq E_2+ (2E_3) + E_4$, and the inequality in (\ref{eq:probExistence}) ensures that there exists a valid collection of sub-POVMs satisfying $G \leq \varepsilon_2 + 4\varepsilon_3 + \epsilon_{w_0}$, with non-vanishing probability.
Therefore, using random coding arguments, there exists at least one collection of sub-POVMs with the above construction satisfying the statement of  Theorem \ref{thm:nf_faithfulSim}.



%% file: NF_POVM_Appx_Proof.tex
\section{Proof of Theorem \ref{thm:nf_dist POVM appx}}\label{appx:nf_proof of thm dist POVM}
\addtocontents{toc}{\protect\setcounter{tocdepth}{1}}
\subsection{Construction of POVMs}\label{subsec:stochastic_dist_POVM}
Suppose there exist POVMs $\bar{M}_A \deq $ $\{\bar{\Lambda}^A_u\}_{u\in \mathcal{U}}$ and $\bar{M}_B\deq\{\bar{\Lambda}^B_v\}_{v\in \mathcal{V}}$ and a stochastic map $ P_{Z|UV}:\mathcal{U\times V\rightarrow \mathcal{Z}} $,  such that ${M}_{AB}$ can be decomposed as
\begin{equation}\label{eq:nf_LambdaAB decompos}
\Lambda^{AB}_{z}=\sum_{u,v} P_{Z|UV}(z|u,v) \bar{\Lambda}^A_{u}\tensor \bar{\Lambda}^B_{v}, ~ \forall z, 
\end{equation}
Note that the proof technique here is very different to the one used in Section \ref{appx:proof of thm dist POVM} for proving Theorem \ref{thm: dist POVM appx}. Recall that in Theorem \ref{thm: dist POVM appx} we initiated the proof by constructing a protocol to faithfully simulate $\bar{M}_A^{\tensor n}\tensor \bar{M}_B^{\tensor n}$. However, here we are not interested in faithfully simulating $\bar{M}_A^{\tensor n}\tensor \bar{M}_B^{\tensor n}$. Instead, by carefully exploiting the private randomness Eve possesses, manifested in terms of the stochastic processing applied by her on the classical bits received, i.e., $ P_{Z|U,V} $, we aim to strictly reduce the sum rate constraints compared to the ones obtained in (\ref{eq:rate4}) of Theorem \ref{thm: dist POVM appx}. This requires a considerably different methodology. More specifically, Lemma \ref{lem:faithful_equivalence} was employed in Theorem \ref{thm: dist POVM appx}, which guaranteed that any two point-to-point POVMs that can individually approximate their corresponding original POVMs, can also faithfully approximate a measurement formed by the tensor product of the original POVMs performed on any state in the tensor product Hilbert space. Such a lemma cannot be developed in the setting involving a stochastic decoder. This is due to the fact that bits received from Alice and Bob are jointly perturbed by the stochastic decoder which doesn't allow a straightforward segmentation into two point-to-point problems. However, the analysis performed in the Section \ref{appx:nf_proof of thm dist POVM} actually modularizes the problem, using an asymmetric partitioning.

Nevertheless, we use the same POVM construction and binning operation as in the proof of Theorem \ref{thm: dist POVM appx}, and hence we appeal to Section \ref{sec:POVM construction} and \ref{sec:POVM binning} for constructing the POVMs based on the codebook $\mathcal{C}^{(\mu)}$ and binning them, resulting in the sub-POVMs $M_1^{( n, \mu_1)}$ and $M_2^{( n, \mu_2)}$ (see \eqref{eq:POVM-14}), and $M_A^{( n, \mu_1)}$ and $M_B^{( n, \mu_2)}$ (see \eqref{eq:POVM-16}), and their completions. All the notations used subsequently can be found in these sections. Therefore, the main focus of the proof hereon is to describe the decoder which is distinct from the one with deterministic mapping, in the sense that it employs the additional stochastic map, and a thorough analysis of the achievability result.

To start with, one can show by using a result similar to Lemma \ref{lem:M_XM_Y POVM} that with probability sufficiently close to one, $M_1^{( n, \mu_1)}$ and $M_2^{( n, \mu_2)}$ form sub-POVMs for all $\mu_1\in [1,N_1]$ and $\mu_2\in [1,N_2]$ if $\tilde{R}_1> I(U;RB)_{\sigma_1}$ and $\tilde{R}_2> I(V;RA)_{\sigma_2}$.
where $\sigma_1, \sigma_2$ are defined as in the statement of the theorem. Further, from $D^{(\mu_1,\mu_2)}$ in \eqref{eq:POVM-17} and $F^{(\mu)}$, as defined subsequently,  we obtain the sub-POVM 
$\tilde{M}_{AB}$ with the following operators. 
\begin{align*}
  \tilde{\Lambda}_{u^n,v^n}^{AB} \delequal \frac{1}{N_1N_2}\sum_{\mu_1=1}^{N_1}\sum_{\mu_2=1}^{N_2} 
  \sum_{(i,j):F^{(\mu)}(i,j)=(u^n,v^n)}
  \Gamma^{A, ( \mu_1)}_i\tensor \Gamma^{B, (\mu_2)}_j, \quad 
\forall (u^n,v^n) \in \mathcal{U}^n \times \mathcal{V}^n.   
\end{align*}
Now, we use the stochastic mapping to define the approximating sub-POVM $\hat{M}^{(n)}_{AB} \deq \{\hat{\Lambda}_{z^n}\}$ as
  \begin{align*}
\hat{\Lambda}^{AB}_{z^n}=\sum_{u^n,v^n}  \tilde{\Lambda}_{u^n,v^n}^{AB}P^n_{Z|U,V}(u^n,v^n), ~ \forall z^n\in \mathcal{Z}^n.
\end{align*}

\subsection{Trace Distance}
In what follows, we show that $\hat{M}_{AB}^{(n)}$ is $\epsilon$-faithful to $M_{AB}^{\tensor n}$ with respect to $\rho_{AB}^{\tensor n}$ (according to Definition \ref{def:faith-sim}), where $\epsilon>0$ can be made arbitrarily small. More precisely, using \eqref{eq:nf_LambdaAB decompos}, we show that, with probability sufficiently close to 1, the following inequality holds
\begin{align}\label{eq:nf_actual trace dist}
{G} = \sum_{z^n} \norm{ \sum_{u^n,v^n}\sqrt{\rho_{AB}^{\tensor n}}\left( \bar{\Lambda}^A_{u^n}\tensor \bar{\Lambda}^B_{v^n} P^n_{Z|U,V}(z^n|u^n,v^n) - \tilde{\Lambda}_{u^n,v^n}^{AB}P^n_{Z|U,V}(u^n,v^n) \right )\sqrt{\rho_{AB}^{\tensor n}}}_1\leq \epsilon.
\end{align}

\noindent \textbf{Step 1: Isolating the effect of error induced by not covering}\\ Consider the second term within $ {G} $, which can be written as
\begin{align}
\sum_{u^n,v^n}&\sqrt{\rho_{AB}^{\tensor n}}\tilde{\Lambda}^{AB}_{u^n,v^n}\sqrt{\rho_{AB}^{\tensor n}}P^n_{Z|U,V}(u^n,v^n) \nonumber\\&= \frac{1}{N_1N_2}\sum_{\mu_1,\mu_2}\sum_{i,j} \sqrt{\rho_{AB}^{\tensor n}}\left (\Gamma^{A, ( \mu_1)}_i\tensor \Gamma^{B, (\mu_2)}_j\right )\sqrt{\rho_{AB}^{\tensor n}} P^n_{Z|U,V}(z^n|F^{(\mu)}(i,j))\underbrace{\sum_{u^n,v^n}\mathbbm{1}_{\{F^{(\mu)}(i,j) = (u^n,v^n)\}}}_{=1} \nonumber \\\vspace{-20pt}
& = T + \widetilde{T}, \nonumber
\end{align}
where \begin{align}
 T \deq & \frac{1}{N_1N_2}\sum_{\mu_1,\mu_2}\sum_{\{i>0\} \medcap \{j>0\}} \sqrt{\rho_{AB}^{\tensor n}}\left (\Gamma^{A, ( \mu_1)}_i\tensor \Gamma^{B, (\mu_2)}_j\right )\sqrt{\rho_{AB}^{\tensor n}} P^n_{Z|U,V}(z^n|F^{(\mu_1,\mu_2)}(i,j)), \nonumber \\
 \widetilde{T} \deq &\frac{1}{N_1N_2}\sum_{\mu_1,\mu_2}\sum_{\{i=0\}\medcup\{j=0\}} \sqrt{\rho_{AB}^{\tensor n}}\left (\Gamma^{A, ( \mu_1)}_i\tensor \Gamma^{B, (\mu_2)}_j\right )\sqrt{\rho_{AB}^{\tensor n}} P^n_{Z|U,V}(z^n|u^{n}_{0},v^{n}_{0}). \nonumber 
\end{align}
Hence, we have
\begin{align}
    G \leq S+ \widetilde{S}, \label{eq:G_separation_1}
\end{align}
where
\begin{align}
    S \deq \sum_{z^n} \norm{ \sum_{u^n,v^n}\sqrt{\rho_{AB}^{\tensor n}}\left( \bar{\Lambda}^A_{u^n}\tensor \bar{\Lambda}^B_{v^n} P^n_{Z|U,V}(z^n|u^n,v^n)\right)\sqrt{\rho_{AB}^{\tensor n}} - T }_1, \label{eq:def_S}
\end{align}
and $\widetilde{S} \deq \sum_{z^n}\|\widetilde{T}\|_1$. Note that $\widetilde{S}$ captures the error induced by not covering the state $\rho_{AB}^{\tensor n}.$
For the term corresponding to $\widetilde{S}$, we prove the following result. 
\begin{proposition}\label{prop:Lemma for S_234}
There exist functions  $\epsilon_{\widetilde{S}}(\delta), $ and $\delta_{\widetilde{S}}(\delta)$, such that for  all sufficiently small $\delta$ and sufficiently large $n$, we have $\PP\left(\widetilde{S} >\epsilon_{\widetilde{S}}(\delta) \right)\leq \delta_{\widetilde{S}}(\delta)$, if  $\tilde{R}_1 > I(U;RB)_{\sigma_1}$ and $\tilde{R}_2 > I(V;RA)_{\sigma_2},$ where $\sigma_1$ and $\sigma_2$ are auxiliary states defined in the theorem and $\epsilon_{\widetilde{S}},\delta_{\widetilde{S}}  \searrow 0$ as $\delta \searrow 0$. 
\end{proposition}
\begin{proof}
The proof is provided in Appendix \ref{appx:proof of S_234}.
\end{proof}
\begin{remark}
The terms corresponding to the operators that complete the sub-POVMs $M_A^{(n,\mu_1)}$ and $M_B^{(n,\mu_2)}$, i.e., $ I - \sum_{u^n \in \TDeltaN(U)}A_{u^n}^{(\mu_1)}$ and $ I - \sum_{v^n \in \TDeltaN(V)}B_{v^n}^{(\mu_2)}$ are taken care in $\widetilde{T}$. The expression $T$ excludes the completing operators. Therefore, we use $A_{u^n}^{(\mu_1)}$  and $B_{v^n}^{(\mu_2)} $ to denote the operators corresponding to $u^n \in \TDeltaN(U)$ and $v^n \in \TDeltaN(V)$, respectively.
\end{remark}

\noindent{\bf Step 2: Isolating the effect of error induced by binning}\\\
Recall the definition of $e^{(\mu)}(u^n,v^n)$ as   $e^{(\mu)}(u^n,v^n)\deq F^{(\mu)}(i,j)$, for each $(u^n, v^n) \in \mathcal{B}^{(\mu_1)}_1(i)\times \mathcal{B}^{(\mu_2)}_2(j)$ and $(u^n,v^n)\in \mathcal{C}^{(\mu)}$. For any $(u^n, v^n) \notin \mathcal{C}^{(\mu)}$ let $e^{(\mu)}(u^n,v^n)=(u_0^n, v_0^n)$.
This simplifies $ T $ as
\begin{align}
T = & \frac{1}{N_1N_2}\sum_{\mu_1,\mu_2}\sum_{\substack{i>0,\\j>0}} \sqrt{\rho_{AB}^{\tensor n}}\left (\sum_{u^n \in B^{(\mu_1)}_{1}(i)}A_{u^n}^{(\mu_1)}\tensor \sum_{v^n \in B^{(\mu_2)}_{2}(j)}B_{v^n}^{(\mu_2)}\right )\sqrt{\rho_{AB}^{\tensor n}} P^n_{Z|U,V}(z^n|F^{(\mu_1,\mu_2)}(i,j)) \nonumber \\
= & \frac{1}{N_1N_2}\sum_{\mu_1,\mu_2}\sum_{u^{n},v^{n}} \sqrt{\rho_{AB}^{\tensor n}}\left (A_{u^n}^{(\mu_1)}\tensor B_{v^n}^{(\mu_2)}\right )\sqrt{\rho_{AB}^{\tensor n}}\sum_{\substack{i>0,\\j>0}}\mathbbm{1}_{\left \{u^n \in B^{(\mu_1)}_{1}(i), v^n \in B^{(\mu_2)}_{2}(j) \right \}} P^n_{Z|U,V}(z^n|e^{(\mu)}(u^n,v^n)) \nonumber \\
= & \frac{1}{N_1N_2}\sum_{\mu_1,\mu_2}\sum_{u^{n},v^{n}} \sqrt{\rho_{AB}^{\tensor n}}\left (A_{u^n}^{(\mu_1)}\tensor B_{v^n}^{(\mu_2)}\right )\sqrt{\rho_{AB}^{\tensor n}} P^n_{Z|U,V}(z^n|e^{(\mu)}(u^n,v^n)), \nonumber
\end{align}
where we have used the fact that $\sum_{u^n \in B^{(\mu_1)}_{1}(i)}A_{u^n}^{(\mu_1)} = \sum_{u^{n}} A_{u^n}^{(\mu_1)}\mathbbm{1}_{\left \{u^n \in B^{(\mu_1)}_{1}(i)\right \}}  $ and $  \sum_{i>0}\mathbbm{1}_{\left \{u^n \in B^{(\mu_1)}_{1}(i)\right \}} = 1 $ for all $u^n \in \TDeltaN(U)$, and similar holds for the POVM $ \{B_{v^n}^{(\mu_2)}\} $. Note that the $(u^n,v^n)$ that appear in the above summation is confined to $(\TDeltaN(U)\times\TDeltaN(V))$, however for ease of notation, we do not make this explicit. We substitute the above expression into $S$ as in \eqref{eq:def_S} to obtain
\begin{align}
S =  \sum_{z^n} &\left \| \sum_{u^n,v^n}\sqrt{\rho_{AB}^{\tensor n}}\left( \bar{\Lambda}^A_{u^n}\tensor \bar{\Lambda}^B_{v^n} -  \frac{1}{N_1N_2}\sum_{\mu_1,\mu_2}A_{u^n}^{(\mu_1)}\tensor B_{v^n}^{(\mu_2)}\right)\sqrt{\rho_{AB}^{\tensor n}}P^n_{Z|U,V}(z^n|e^{(\mu)}(u^n,v^n))\right \|_1. \nonumber
\end{align}
 We add and subtract an appropriate term within $ S $ and apply triangle inequality to isolate the effect of binning as $ S \leq S_{1}  + S_{2},$ where 
\begin{align}
S_{1} & \deq \sum_{z^n} \left\|\sum_{u^n,v^n}\sqrt{\rho_{AB}^{\tensor n}}  \left (\bar{\Lambda}^A_{u^n}\tensor \bar{\Lambda}^B_{v^n} -\frac{1}{N_1N_2} \sum_{\mu_1,\mu_2} A_{u^n}^{(\mu_1)} \tensor B_{v^n}^{(\mu_2)}\right )\sqrt{\rho_{AB}^{\tensor n}} P^n_{Z|U,V}(z^n|u^n,v^n)\right\|_1  \text{ and  } \nonumber \\
S_{2} & \deq \sum_{z^n}\left \| \frac{1}{N_1N_2} \sum_{\mu_1,\mu_2} \sum_{u^n,v^n} \sqrt{\rho_{AB}^{\tensor n}}\left (A_{u^n}^{(\mu_1)} \tensor B_{v^n}^{(\mu_2)}\right ) \sqrt{\rho_{AB}^{\tensor n}} \left(P^n_{Z|U,V}(z^n|u^n,v^n)- P^n_{Z|U,V}\left (z^n|e^{(\mu)}(u^n,v^n)\right )\right)\right \|_1. \label{eq:binningIsolation}
\end{align}
Note that the term $ S_{1} $ characterizes the error introduced by approximation of the original POVM with  the collection of  approximating sub-POVMs $ M_1^{(n,\mu_1)} $ and $ M_2^{{(n,\mu_2)}} $, and the term $ S_2 $ characterizes the error caused by binning of these approximating sub-POVMs. In this step, we analyze $S_2$ and prove the following proposition. 

\begin{proposition}[Mutual Packing]\label{prop:Lemma for S_12}
There exist functions  $\epsilon_{S_{2}}(\delta) $ and $\delta_{{S}_{2}}(\delta)$, such that for  all sufficiently small $\delta$ and sufficiently large $n$, we have $\PP\left({S}_2 >\epsilon_{{S_2}}(\delta) \right)$ $\leq \delta_{{S}_2}(\delta)$, if $ \tilde{R}_{1} + \tilde{R}_{2} - R_1 -R_2 < I(U;V)_{\sigma_{3}}$,  where $\sigma_3$ is the auxiliary state defined in the theorem and $\epsilon_{{S}_2},\delta_{{S}_2}  \searrow 0$ as $\delta \searrow 0$.
\end{proposition}
\begin{proof}
The proof is provided in Appendix \ref{appx:proof of S_12}
\end{proof}
\noindent\textbf{Step 3: Isolating the effect of Alice's approximating measurement}\\
\noindent In this step, we separately analyze the effect of approximating measurements  at the two distributed parties in the term $S_1$. For that, we split $S_{1}$ as $ S_{1} \leq Q_1 + Q_2 $, where
\begin{align}
Q_1 &\deq \sum_{z^n} \left\|\sum_{u^n,v^n}\sqrt{\rho_{AB}^{\tensor n}}  \left (\bar{\Lambda}^A_{u^n}\tensor \bar{\Lambda}^B_{v^n} -\frac{1}{N_1} \sum_{\mu_1=1}^{N_1} A_{u^n}^{(\mu_1)} \tensor \bar{\Lambda}^B_{v^n}\right )\sqrt{\rho_{AB}^{\tensor n}} P^n_{Z|U,V}(z^n|u^n,v^n)\right\|_1, \nonumber \\
Q_2 &\deq \sum_{z^n} \left\|\frac{1}{N_1} \sum_{\mu_1=1}^{N_1}\sum_{u^n,v^n}\sqrt{\rho_{AB}^{\tensor n}}  \left ( A_{u^n}^{(\mu_1)} \tensor \bar{\Lambda}^B_{v^n} -\frac{1}{N_2} \sum_{\mu_2=1}^{N_2} A_{u^n}^{(\mu_1)} \tensor B_{v^n}^{(\mu_2)}\right )\sqrt{\rho_{AB}^{\tensor n}} P^n_{Z|U,V}(z^n|u^n,v^n)\right\|_1.\nonumber
\end{align} 
 With this partition, the terms within the trace norm of $ Q_1 $ differ only in the action of Alice's measurement. And similarly, the terms within the norm of $ Q_2 $ differ only in the action of Bob's measurement. Showing that these two terms are small forms a major portion of the achievability proof. 

\noindent{\bf Analysis of $ Q_1$:}  
To show $ Q_1 $ is small,
we compute rate constraints which ensure that an upper bound to $ Q_1 $ can be made to vanish in an expected sense.
Furthermore, this upper bound becomes convenient in obtaining a single-letter characterization for the rate needed to make the term
corresponding to $ Q_2 $ vanish. For this, we define $ J $ as
\begin{align}\label{def:J}
J \deq \sum_{z^n,v^n}\left\|\sum_{u^n}\sqrt{\rho_{AB}^{\tensor n}}  \left (\bar{\Lambda}^A_{u^n}\tensor \bar{\Lambda}^B_{v^n} -\frac{1}{N_1} \sum_{\mu_1=1}^{N_1} 	  A_{u^n}^{(\mu_1)} \tensor  \bar{\Lambda}^B_{v^n}\right )\sqrt{\rho_{AB}^{\tensor n}} P^n_{Z|U,V}(z^n|u^n,v^n)\right\|_1.
\end{align}
By defining $J$ and using triangle inequality for block operators (which holds with equality), we add the sub-system $V$ to $RZ$, resulting in the joint system $RZV$, corresponding to the state $\sigma_3$ as defined in the theorem. Then we approximate the joint system $RZV$ using an approximating sub-POVM $M_A^{(n)}$ producing outputs on the alphabet $\mathcal{U}^n$. To make $J$ small for sufficiently large n, we expect the sum of the rate of the approximating sub-POVM and common randomness, i.e., $\tilde{R}_1 + C_1$, to be larger than $I(U;RZV)_{\sigma_3}$. We seek to prove this in the following.

Note that from triangle inequality, we have $ Q_1 \leq J. $
Further, we add and subtract an appropriate term
within $ J $ and use triangle inequality obtain $ J \leq J_1 + J_2 $, where
\begin{align}
J_1 & \deq  \sum_{z^n,v^n}\left\|\sum_{u^n}\sqrt{\rho_{AB}^{\tensor n}}  \left (\bar{\Lambda}^A_{u^n}\tensor \bar{\Lambda}^B_{v^n} -\frac{1}{N_1}\sum_{\mu_1=1}^{N_1} 	  \frac{\gamma_{u^n}^{(\mu_1)}}{\lambdauA}\bar{\Lambda}^A_{u^n} \tensor  \bar{\Lambda}^B_{v^n}\right )\sqrt{\rho_{AB}^{\tensor n}} P^n_{Z|U,V}(z^n|u^n,v^n)\right\|_1 \mbox{ and } \nonumber \\
J_2 &\deq \sum_{z^n,v^n}\left\|\sum_{u^n}\sqrt{\rho_{AB}^{\tensor n}}  \left (\frac{1}{N_1}\sum_{\mu_1=1}^{N_1} \frac{\gamma_{u^n}^{(\mu_1)}}{\lambdauA}\bar{\Lambda}^A_{u^n}\tensor \bar{\Lambda}^B_{v^n} -\frac{1}{N_1} \sum_{\mu_1=1}^{N_1} 	  A_{u^n}^{(\mu_1)} \tensor  \bar{\Lambda}^B_{v^n}\right )\sqrt{\rho_{AB}^{\tensor n}} P^n_{Z|U,V}(z^n|u^n,v^n)\right\|_1. \nonumber
\end{align}
Now with the intention of employing Lemma \ref{lem:nf_SoftCovering}, we express $J_1$ as
\begin{align}
J_1 & = \left\|\sum_{z^n,u^n,v^n}\lambdaAB\rhohatuAvB \tensor P^n_{Z|U,V}(z^n|u^n,v^n)\ketbra{v^n}\tensor\ketbra{z^n}\right . \nonumber \\ 
& \left.\hspace{18pt}- \frac{(1-\varepsilon)}{(1+\eta)}\frac{1}{2^{n(\tilde{R}_1+C_1)}}\sum_{\mu_1,l}\sum_{z^n,u^n,v^n}\hspace{-10pt}\IndiU1\frac{\lambdaAB}{\lambdauA}\rhohatuAvB \tensor P^n_{Z|U,V}(z^n|u^n,v^n)\ketbra{v^n}\tensor\ketbra{z^n}\right\|_1, \nonumber 
\end{align}
where the equality above is obtained by using the definitions of $ \gamma_{u^n}^{(\mu_1)} $ and $ \rhohatuAvB $, followed by using the triangle inequality for the block diagonal operators, which in fact becomes an equality.

\noindent Let us define $ \mathcal{T}_{u^n} $ as
\begin{align}
\mathcal{T}_{u^n} = \sum_{z^n,v^n}\frac{\lambdaAB}{\lambdauA}\rhohatuAvB \tensor P^n_{Z|U,V}(z^n|u^n,v^n)\ketbra{v^n}\tensor\ketbra{z^n}. \nonumber
\end{align} 
Note that the above definition of $ \mathcal{T}_{u^n} $ contains all the elements in product form, and thus it can be written as $ \mathcal{T}_{u^n} = \bigtensor_{i=1}^{n}\mathcal{T}_{u_i}. $
This simplifies $ J_1 $ as
\begin{align}
J_1 &= \left\|\sum_{u^n}\lambdauA\mathcal{T}_{u^n} -  \frac{(1-\varepsilon)}{(1+\eta)}\frac{1}{2^{n(\tilde{R}_1+C_1)}}\sum_{\mu_1,l} \mathcal{T}_{U^{n,(\mu_1)}(l)} \right\|_1. \nonumber
\end{align}
Now, using Lemma \ref{lem:nf_SoftCovering} we get the following bound.
For any $ \epsilon_{\scriptscriptstyle J_1}, \delta_{\scriptscriptstyle J_1} \in (0,1)$, if  
\begin{align}
\tilde{R}_1+C_1 > S\left (\sum_{u \in \mathcal{U}}\lambda_u^A\mathcal{T}_{u}\right ) + \sum_{u \in \mathcal{U}}\lambda_u^A S(\mathcal{T}_{u}) = I(U;RZV)_{\sigma_{3}}, \label{constraint:nfDist1}
\end{align} 
then $ \PP(J_1 \geq \epsilon_{\scriptscriptstyle J_1} ) \leq \delta_{\scriptscriptstyle J_1} $ for sufficiently large n, where $ \sigma_3 = \sum_{u \in \mathcal{U}}\lambda_u^A\mathcal{T}_{u}\tensor \ketbra{u}$.

Now, we consider the term corresponding to $ J_2 $ and prove that its expectation with respect to the Alice's codebook is small. Recalling $ J_2 $, we get
\begin{align}
J_2  & \leq \frac{1}{N_1}\sum_{\mu_1=1}^{N_1}\sum_{u^n,v^n}\sum_{z^n} P^n_{Z|U,V}(z^n|u^n,v^n) \left\|\sqrt{\rho_{AB}^{\tensor n}}  \left ( \cfrac{\gamma_{u^n}^{(\mu_1)}}{\lambdauA}\bar{\Lambda}^A_{u^n}\tensor \bar{\Lambda}^B_{v^n} -	  A_{u^n}^{(\mu_1)} \tensor  \bar{\Lambda}^B_{v^n}\right )\sqrt{\rho_{AB}^{\tensor n}} \right\|_1, \nonumber \\
& = \frac{1}{N_1}\sum_{\mu_1=1}^{N_1}\sum_{u^n,v^n}\gamma_{u^n}^{(\mu_1)} \left\|\sqrt{\rho_{AB}^{\tensor n}}  \left ( \left ( \cfrac{1}{\lambdauA}\bar{\Lambda}^A_{u^n}-\sqrt{\rho_{A}^{\tensor n}}^{-1}\LambdauA\sqrt{\rho_{A}^{\tensor n}}^{-1}\right )\tensor \bar{\Lambda}^B_{v^n} \right )\sqrt{\rho_{AB}^{\tensor n}} \right\|_1, \nonumber
\end{align}
where the inequality is obtained by using triangle and the next equality follows from the fact that $ \sum_{z^n} P^n_{Z|U,V}(z^n|u^n,v^n) =1  $ for all $ u^n \in \mathcal{U}^n $ and $ v^n \in \mathcal{V}^n $ and using the definition of $ A_{u^n}^{(\mu_1)} $.
By applying expectation of $ J_2  $ over  the Alice's codebook, we get
\begin{align}
 \EE{\left[J_2 \right]} & \leq \frac{1}{(1+\eta)}\!\!\sum_{\substack{u^n \in \TDeltaN(U)}}\!\!\!\lambdauA\sum_{v^n} \left\|\sqrt{\rho_{AB}^{\tensor n}}  \left ( \left ( \cfrac{1}{\lambdauA}\bar{\Lambda}^A_{u^n}-\sqrt{\rho_{A}^{\tensor n}}^{-1}\LambdauA\sqrt{\rho_{A}^{\tensor n}}^{-1}\right )\tensor \bar{\Lambda}^B_{v^n} \right )\sqrt{\rho_{AB}^{\tensor n}} \right\|_1, \nonumber
\end{align}
where we have used the fact that $ \EE{[\gamma_{u^n}^{(\mu_1)}]} = \frac{\lambdauA}{(1+\eta)} $.
To simplify the above equation, we employ Lemma \ref{lem:Separate} from Section \ref{sec:traceDistance_POVM} that completely discards the effect of Bob's measurement. Since $ \sum_{v^n}\bar{\Lambda}^B_{v^n} = I$, from Lemma \ref{lem:Separate} we have for every $ u^n \in \TDeltaN(A) $,
\begin{align}
	& \sum_{v^n}\left\|\sqrt{\rho_{AB}^{\tensor n}}  \left ( \left ( \cfrac{1}{\lambdauA}\bar{\Lambda}^A_{u^n}-\sqrt{\rho_{A}^{\tensor n}}^{-1}\LambdauA\sqrt{\rho_{A}^{\tensor n}}^{-1}\right )\tensor \bar{\Lambda}^B_{v^n} \right )\sqrt{\rho_{AB}^{\tensor n}} \right\|_1 \nonumber \\ 
	& \hspace{50pt} = \left\|\sqrt{\rho_{A}^{\tensor n}}  \left ( \cfrac{1}{\lambdauA}\bar{\Lambda}^A_{u^n}-\sqrt{\rho_{A}^{\tensor n}}^{-1}\LambdauA\sqrt{\rho_{A}^{\tensor n}}^{-1} \right )\sqrt{\rho_{A}^{\tensor n}} \right\|_1.\nonumber
\end{align}
This simplifies $ \EE{\left[J_2 \right]} $ as 
\begin{align}
\EE{\left[J_2 \right]}  & \leq  \frac{1}{(1+\eta)}\!\!\!\sum_{\substack{u^n \in \TDeltaN(U) }}\!\!\lambdauA\left\|\sqrt{\rho_{A}^{\tensor n}}  \left ( \cfrac{1}{\lambdauA}\bar{\Lambda}^A_{u^n}-\sqrt{\rho_{A}^{\tensor n}}^{-1}\LambdauA\sqrt{\rho_{A}^{\tensor n}}^{-1} \right )\sqrt{\rho_{A}^{\tensor n}} \right\|_1 \nonumber \\
& =  \frac{1}{(1+\eta)}\!\!\!\sum_{\substack{u^n \in \TDeltaN(U) }}\!\!\lambdauA\left\| \left (\rhohatuA -\LambdauA \right ) \right\|_1 \leq \cfrac{(1-\varepsilon)}{(1+\eta)} (2\sqrt{\varepsilon'_A} + 2\sqrt{\varepsilon''_A}) = \epsilon_{\scriptscriptstyle J_2}\nonumber,
\end{align}
where the last inequality is obtained by the repeated usage of the average gentle measurement lemma by setting $ \epsilon_{\scriptscriptstyle J_2} = \frac{(1-\varepsilon)}{(1+\eta)} (2\sqrt{\varepsilon'_A} + 2\sqrt{\varepsilon''_A}) $ with  $ \epsilon_{\scriptscriptstyle J_2} \searrow 0 $ as $ \delta \searrow 0 $ and $ \varepsilon'_A = \varepsilon + 2\sqrt{\varepsilon} $ and $ \varepsilon''_A = 2\varepsilon + 2\sqrt{\varepsilon} $ ( see (35) in \cite{wilde_e} for details). 
Since $Q_1 \leq J \leq J_1 + J_2$, hence $J$, and consequently $ Q_1$, can be made arbitrarily small for sufficiently large n, if $\tilde{R}_1+C_1 > I(U;RZV)_{\sigma_3}$. Now we move on to bounding $ Q_2 $.

\noindent{\bf Step 4: Analyzing the effect of Bob's approximating measurement}\\
Step 3 ensured that the sub-system $RZV$ is close to a tensor product state in trace-norm. In this step, we approximate the state corresponding to the sub-system $RZ$ using the approximating POVM $M_B^{(n)}$, producing outputs on the alphabet $\mathcal{V}^n$. We proceed with the following proposition. 
\begin{proposition} \label{prop:Lemma for Q2}
There exist functions  $\epsilon_{Q_{2}}(\delta) $ and $\delta_{{Q}_{2}}(\delta)$, such that for  all sufficiently small $\delta$ and sufficiently large $n$, we have $\PP\left({Q}_2 >\epsilon_{{Q_2}}(\delta) \right)$ $\leq \delta_{{Q}_2}(\delta_3)$, if $ \tilde{R}_2+C_2 >  I(V;RZ)_{\sigma_3}$,  where $\sigma_3$ is the auxiliary state defined in the theorem and $\epsilon_{Q_2},\delta_{{Q}_2}  \searrow 0$ as $\delta \searrow 0$.
\end{proposition}

\begin{proof}
 The proof is provided in Appendix \ref{appx:proof of Q2}.
\end{proof}

 \subsection{Rate Constraints}
 To sum-up, we showed that the trace distance inequality in \eqref{eq:nf_actual trace dist} holds for sufficiently large $n$ and with probability sufficiently close to 1, if the following bounds hold:
 \begin{subequations}\label{eq:nf_rate-region 1}
 	\begin{align}
 	\tilde{R}_1& \geq I(U; RB)_{\sigma_1},\\
 	\tilde{R}_2& \geq I(V; RA)_{\sigma_2},\\
	\tilde{R}_1 + C_1 & \geq I(U; RZV)_{\sigma_3},\\
	\tilde{R}_2 + C_2 & \geq I(V; RZ)_{\sigma_3},\\
 	(\tilde{R}_1-R_1)+(\tilde{R}_2-R_2) & < I(U;V)_{\sigma_3},\\
 	\tilde{R}_1\geq R_1\geq 0, \quad \tilde{R}_2&\geq R_2 \geq 0, \\ C_1 + C_2 &\leq C,\quad C\geq 0,
 	\end{align}
 \end{subequations}
 where $C\deq \frac{1}{n}\log_2 N$. 
Therefore, there exists a distributed protocol with parameters $(n, 2^{nR_1}, 2^{nR_2}, 2^{nC})$ such that its overall POVM $\hat{M}_{AB}$ is $\epsilon$-faithful to $M_{AB}^{\tensor n}$ with respect to $\rho_{AB}^{\tensor n}$. 
 Lastly,  we complete the proof of the theorem using the following lemma.
 \begin{lem}
 	Let $\mathcal{R}_1$ denote the set of all $(R_1,R_2,C)$ for which there exists $(\tilde{R}_1,\tilde{R}_2)$ such that the septuple $(R_1,R_2, C, \tilde{R}_1, \tilde{R}_2, \tilde{C}_1, \tilde{C}_2)$ satisfies the inequalities in \eqref{eq:nf_rate-region 1}. Let, $\mathcal{R}_2$ denote the set of all triples $(R_1, R_2, C)$ that satisfies  the inequalities in \eqref{eq:nf_dist POVm appx rates} given in the statement of the theorem. Then, $
 	\mathcal{R}_1=\mathcal{R}_2$.
 \end{lem}
 \begin{proof}
 	This follows from Fourier-Motzkin elimination \cite{fourier-motzkin}.
 	\end{proof}

%% file: Lemmas_Appx_proof.tex
\section{Proof of Lemmas}
\addtocontents{toc}{\protect\setcounter{tocdepth}{1}}
\subsection{Proof of Lemma \ref{lem:Separate}} \label{appx:proofLemmaSeparate}
 Consider the LHS of \eqref{eq:lemSeparate1}. We define an operator $\Lambda_{y_0}$ which completes the sub-POVM $\{\Lambda_y\}_{y \in \mathcal{Y}}$ as $\Lambda_{y_0} \deq I - \sum_{y \in \mathcal{Y}}\Lambda_y$. Further, let the set $\mathcal{Y}^+ \deq \mathcal{Y}\medcup\{y_0\}$. Since trace norm is invariant to transposition with respect to $\rho_{AB}$, we can write for any $ y \in \mathcal{Y}^+ $,
\begin{align}
\left\| \sqrt{\rho_{AB}}\left (\Gamma^A\tensor \Lambda_y^B\right )\sqrt{\rho_{AB}}\right \|_1
=\left\| \left [\sqrt{\rho_{AB}}\left (\Gamma^A\tensor \Lambda_y^B \right )\sqrt{\rho_{AB}}\right]^T\right \|_1
=\left\| \sqrt{\rho_{AB}}\left ({(\Gamma^A)}^T\tensor {(\Lambda_y^B)}^T\right )\sqrt{\rho_{AB}}\right \|_1. \nonumber
\end{align}
One can easily prove for any $\Gamma_A$ (not necessarily positive) that
\begin{align}
\left(\sqrt{\rho_{AB}}\left ({(\Gamma^A)}^T\tensor {(\Lambda_y^B)}^T\right )\sqrt{\rho_{AB}}\right)^R = \Tr_{AB}\left \{ \left (\text{id}\tensor\Gamma^A\tensor \Lambda_y^B \right )\Psi_{RAB} \right \}, \label{eq:traceTranspose}
\end{align}
where $ \Psi_{RAB}$ is the canonical purification of $\rho_{AB} $ defined as $ \Psi_{RAB} \deq \sum_{x,x'}\sqrt{\lambda_x \lambda_{x'}}\ketbra{x}_{AB}\tensor\ketbra{x'}_R$ for the spectral decomposition of $\rho_{AB}$ given as $\rho_{AB} = \sum_{x}\lambda_x\ketbra{x}_{AB}$.  
Now, using \eqref{eq:traceTranspose} we perform the following simplification
\begin{align}
\sum_{y \in \mathcal{Y}}\left\| \sqrt{\rho_{AB}}\left (\Gamma^A\tensor \Lambda_y^B\right )\sqrt{\rho_{AB}}\right \|_1 & \leq \sum_{y \in \mathcal{Y}^+}\left\| \sqrt{\rho_{AB}}\left (\Gamma^A\tensor \Lambda_y^B\right )\sqrt{\rho_{AB}}\right \|_1 \nonumber \\ & = \sum_{y \in \mathcal{Y}^+} \Big \|\Tr_{AB}\left \{ \left (\text{id}_R\tensor\Gamma^A\tensor \Lambda_y^B \right )\Psi_{RAB} \right \}\Big \|_1  \nonumber \\
& = \Big \|\sum_{y \in \mathcal{Y}^+}\Tr_{AB}\left \{ \left (\text{id}_{RB}\tensor\Gamma^A \right )\left (\text{id}_{RA}\tensor \Lambda_y^B \right )\Psi_{RAB} \tensor \ketbra{y}\right \}\Big \|_1 \nonumber \\
& = \left \|\Tr_{A}\Big \{ \left (\text{id}_{RY}\tensor\Gamma^A \right )\left(\sum_{y \in \mathcal{Y}^+}\ketbra{y}\tensor \Tr_B\left \{\left (\text{id}_{RA}\tensor \Lambda_y^B \right )\Psi_{RAB}\right \}\right)  \Big \}\right \|_1 \nonumber \\ 
& = \left \|\Tr_{A}\left \{ \left (\text{id}_{RY}\tensor\Gamma^A \right )\sigma_{RAY}\right \}\right \|_1 \nonumber \\
& = \left \|\Tr_{AZ}\left \{ \left (\text{id}_{RY}\tensor\Gamma^A \tensor \text{id}_Z \right )\Phi_{RAYZ}^{\sigma_{RAY}}\right \}\right \|_1, \label{eq:traceequality1}
\end{align} 
 where the second equality uses the triangle inequality for block diagonal operators,  the third equality first uses the property that $\tr_{XY}\{\} = \tr_{X}\{\tr_Y\{\}\}$, followed by the definition of partial trace and its linearity, the fourth equality uses $ \sigma_{RAY} $  defined as
\begin{align}
\sigma_{RAY} = \sum_{y \in \mathcal{Y}^+}\ketbra{y}\tensor \Tr_B\left \{\left (\text{id}_{RA}\tensor \Lambda_y^B \right )\Psi_{RAB}\right \}, \nonumber
\end{align}
and finally, the last one uses$ \Phi_{RAYZ}^{\sigma_{RAY}}$ defined as the  canonical purification of $\sigma_{RY}.$
Note that the above inequality becomes an equality when $\sum_{y\in \mathcal{Y}}\Lambda_y = I$.
Using similar sequence of arguments as used in \eqref{eq:traceTranspose}, we have 
\begin{align}
\left \|\Tr_{AZ}\left \{ \left (\text{id}_{RY}\tensor\Gamma^A \tensor \text{id}_Z \right )\Phi_{RAYZ}^{\sigma_{RAY}}\right \}\right \|_1 
= \left \|\sqrt{\tr_{RYZ}\left\{\Phi_{RAYZ}^{\sigma_{RAY}}\right\}}\Gamma^A \sqrt{\tr_{RYZ}\left\{\Phi_{RAYZ}^{\sigma_{RAY}}\right\}} \right \|_1 = \|\sqrt{\rho_A}\Lambda_x^A \sqrt{\rho_A} \|_1. \nonumber 
\end{align}
This completes the proof.
\subsection{Proof of Lemma \ref{lem:mutual covering}}\label{appx:proof of mutual covering}
Let the operators of $\hat{M}_X$ and $\hat{M}_Y$ be denoted by $\{\hat{\Lambda}^X_i\}_{i\in \mathcal{I}}$ and $\{\hat{\Lambda}^Y_j\}_{j\in \mathcal{J}}$, respectively, and let the operators of $M_{X}$ and $M_Y$ be denoted by $\{\Lambda_i^X\}$ and $\{\Lambda_j^Y\}$, respectively, for some finite sets $\mathcal{I}$ and $\mathcal{J}$. With this notation, we need to show the following inequality 
\begin{align*}
  G \deq \sum_{i,j} \Big\|\sqrt{\rho_{XY}} ({\Lambda}_{i}^X\tensor{\Lambda}_{j}^Y-\hat{\Lambda}_{i}^X\tensor\hat{\Lambda}_j^Y) \sqrt{\rho_{XY}}\Big\|_1 + \text{Tr}\bigg\{\Big(I-\sum_{i,j} \hat{\Lambda}_i^X\tensor \hat{\Lambda}^Y_j\Big) \rho_{XY} \bigg\} \leq 2\epsilon,  
\end{align*}
where $\Psi^\rho_{RXY}$ is a purification of $\rho_{XY}$. Next, by adding and subtracting appropriate terms, we get
\begin{align}
 G & \leq  \sum_{i,j} \Big\|\sqrt{\rho_{XY}} (\Lambda_{i}^X\tensor\Lambda_{j}^Y-\hat{\Lambda}_{i}^X\tensor\Lambda_j^Y) \sqrt{\rho_{XY}}\Big\|_1+   \Tr\bigg\{\Big(I-\sum_{i} \hat{\Lambda}_i^X\Big) \rho_{X} \bigg\} \nonumber\\ 
& \hspace{20pt}+ \sum_{i,j} \Big\|\sqrt{\rho_{XY}} (\hat{\Lambda}_{i}^X\tensor\Lambda_{j}^Y-\hat{\Lambda}_{i}^X\tensor\hat{\Lambda}_j^Y) \sqrt{\rho_{XY}}\Big\|_1 + \Tr\bigg\{\Big(I-\sum_{j} \hat{\Lambda}_j^Y\Big) \rho_{Y} \bigg\} \nonumber \\
&\hspace{20pt}+\text{Tr}\bigg\{\Big(I-\sum_{i,j} \hat{\Lambda}_i^X\tensor \hat{\Lambda}^Y_j\Big) \rho_{XY} \bigg\}-\text{Tr}\bigg\{\Big(I-\sum_{i} \hat{\Lambda}_i^X\Big) \rho_{X}\bigg\} - \Tr\bigg\{\Big(I-\sum_{j} \hat{\Lambda}_j^X\Big) \rho_{Y} \bigg\} \nonumber \\
& \leq \sum_{i} \Big\|\sqrt{\rho_{X}} (\Lambda_{i}^X-\hat{\Lambda}_{i}^X) \sqrt{\rho_{X}}\Big\|_1+   \Tr\bigg\{\Big(I-\sum_{i} \hat{\Lambda}_i^X\Big) \rho_{X} \bigg\} \nonumber\\ 
& \hspace{20pt}+ \sum_{j} \Big\|\sqrt{\rho_{Y}} (\Lambda_{j}^Y-\hat{\Lambda}_j^Y) \sqrt{\rho_{Y}}\Big\|_1 + \Tr\bigg\{\Big(I-\sum_{j} \hat{\Lambda}_j^Y\Big) \rho_{Y} \bigg\} \nonumber \\
&\hspace{20pt} +\text{Tr}\bigg\{\Big(I-\sum_{i,j} \hat{\Lambda}_i^X\tensor \hat{\Lambda}^Y_j\Big) \rho_{XY} \bigg\}-\text{Tr}\bigg\{\Big(I-\sum_{i} \hat{\Lambda}_i^X\Big) \rho_{X}\bigg\} - \Tr\bigg\{\Big(I-\sum_{j} \hat{\Lambda}_j^Y\Big) \rho_{Y} \bigg\} \nonumber \\
& \leq 2\epsilon + \text{Tr}\bigg\{\Big(\sum_{i} \hat{\Lambda}_i^X\tensor (I-\sum_{j}\hat{\Lambda}^Y_j)\Big) \rho_{XY} \bigg\} - \Tr\bigg\{\Big(I-\sum_{j} \hat{\Lambda}_j^Y\Big) \rho_{Y} \bigg\} \leq 2\epsilon, \nonumber
\end{align}
where the second inequality follows by applying Lemma \ref{lem:Separate} twice, the third inequality follows from the hypotheses of the lemma, and the final inequality uses the fact that $\hat{M}^X$ and $\hat{M}^Y$ are sub-POVMs. This completes the proof of the lemma.

\subsection{Proof of Lemma \ref{lem:nf_SoftCovering}} \label{lem:nf_SoftCovering_Proof}
\begin{proof}
Consider the trace norm expression given in \eqref{eq:nf_SoftCovering_Term}. This expression can be upper bounded using triangle inequality as 
\begin{align}
         & \left\|\sum_{w^n}\lambda_{w^n}\mathcal{T}_{w^n} - \cfrac{1}{2^{n(R+C)}}\cfrac{(1-\varepsilon)}{(1+\eta)}\sum_{l,\mu} \mathcal{T}_{W^{n,(\mu)}(l)}\right\|_1 \nonumber \\ & \leq \Bigg \|\sum_{w^n}\lambda_{w^n}\mathcal{T}_{w^n} - \cfrac{(1-\varepsilon)}{(1+\eta)}\!\!\!\sum_{\substack{w^n\in \\\TDeltaN(W)}}\hspace{-5pt}\tilde{P}_{W^n}(w^n)\mathcal{T}_{w^n}\Bigg \|_1 \hspace{-5pt}+  \cfrac{(1-\varepsilon)}{(1+\eta)}\Bigg\| \sum_{\substack{w^n\in \\\TDeltaN(W)}}\hspace{-10pt}\tilde{P}_{W^n}(w^n)\mathcal{T}_{w^n} - \cfrac{1}{2^{n(R+C)}}\sum_{l,\mu} \mathcal{T}_{W^{n,(\mu)}(l)}\Bigg \|_1. \label{eq:nf_SoftCovering_TermSplit}
    \end{align}
\noindent The first term in the right-hand side  is bounded from above as 
\begin{align}
     \Big\|\sum_{w^n}\lambda_{w^n}\mathcal{T}_{w^n}  - \cfrac{(1-\varepsilon)}{(1+\eta)}\sum_{w^n\in \TDeltaN(W)}\tilde{P}_{W^n}(w^n)\mathcal{T}_{w^n}\Big\|_1  & \leq \Big\|\hspace{-10pt}\sum_{w^n\in \TDeltaN(W)}\hspace{-10pt}\lambda_{w^n}\Bigg(1-\cfrac{1}{(1+\eta)}\Bigg)\mathcal{T}_{w^n} \Big\|_1 + \Big\| \hspace{-10pt}\sum_{w^n \notin \TDeltaN(W)}\hspace{-10pt}\lambda_{w^n}\mathcal{T}_{w^n}\Big\|_1 \nonumber\\
    & \leq \Bigg(\cfrac{\eta}{1+\eta}\Bigg)\sum_{w^n\in \TDeltaN(W)}\hspace{-10pt}\lambda_{w^n}\underbrace{\left \|\mathcal{T}_{w^n}\right \|_1}_{=1} + \sum_{w^n \notin \TDeltaN(W)}\hspace{-10pt}{\lambda_{w^n}}\underbrace{\left \|\mathcal{T}_{w^n}\right \|}_{=1}\nonumber \\ 
    & \leq \Bigg(\cfrac{\eta}{1+\eta}\Bigg) + {\varepsilon} \leq \eta + \varepsilon = \varepsilon_1, \label{eq:ensemble2}
\end{align}
where $ \varepsilon_1 $ can be made arbitrarily small for all sufficiently large $n$.
Now consider the second term in (\ref{eq:nf_SoftCovering_TermSplit}). Using covering lemma from \cite{Wilde_book}, this can be bounded as follows.
For $w^n \in \TDeltaN(W)$, let $\Pi$ and $\Pi_{w^n}$ denote the projectors onto the typical subspace of $\mathcal{T}^{\tensor n}$ and $\mathcal{T}_{w^n}$, respectively, where $ \mathcal{T} = \sum_{w^n}\lambda_{w^n}\mathcal{T}_{w^n}. $
From the definition of typical projectors, for any $ \epsilon_1 \in (0,1) $ we have for sufficiently large $ n $, the following inequalities satisfied for all $w^n \in \TDeltaN(W):$
\begin{align}
    \Tr{\Pi\mathcal{T}_{w^n}} &\geq 1-\epsilon_1, \nonumber \\
    \Tr{\Pi_{w^n}\mathcal{T}_{w^n}} &\geq 1-\epsilon_1, \nonumber \\
    \Tr{\Pi} &\leq D, \nonumber \\
    \Pi_{w^n}\mathcal{T}_{w^n}\Pi_{w^n} & \leq \frac{1}{d}\Pi_{w^n},
    \label{eq:coveringHyp}
\end{align}
where $D = 2^{n(S(\mathcal{T})+\delta_1)}$ and $ \displaystyle d = 2^{n\left[\left(\sum_{w}\lambda_w S(\mathcal{T}_{w})\right)+\delta_2\right]}$, and $ \delta_1 \searrow 0, \delta_2 \searrow 0 $ as $ \epsilon_{1} \searrow 0 $.
From the statement of the covering lemma, we know that for an ensemble $\{\tilde{P}_{W^n}(w^n), \mathcal{T}_{w^n}\}_{w^n \in \mathcal{W}^n}$, if there exists projectors $\Pi$ and $\Pi_{w^n}$ such that they satisfy the set of inequalities in (\ref{eq:coveringHyp}), then for any $\epsilon_2 > 0$,  sufficiently large $n$ and $n(R+C) > \log_2{\frac{D}{d}}$, the obfuscation error, defined as 
\begin{align*}
    \Bigg \|\sum_{w^n}\tilde{P}^n_W(w^n)\mathcal{T}_{w^n} - \cfrac{1}{2^{R+C}}\sum_{l,\mu}\mathcal{T}_{W^{n,(\mu)}(l)}\Bigg \|_1,
\end{align*}
can be made smaller than $\varepsilon_2$ with high probability. 
This gives us the the following rate constraints $R + C \geq S({\sum_{w}\lambda_w\mathcal{T}_{w}}) - \sum_{w}\lambda_{w} S(\mathcal{T}_{w}) + \delta' = \chi\left(\{\lambda_w\},\{\hat{\rho}_{w}\tensor\sigma_{w}\}\right) + \delta'$. Using this constraint and the bound from \eqref{eq:ensemble2}, the result follows. 
\end{proof}

%% file: Proof_of_Propositions.tex
\section{Proof of Propositions}
\addtocontents{toc}{\protect\setcounter{tocdepth}{1}}
\subsection{Proof of Proposition \ref{prop:Lemma for Simplification1}}\label{appx:proof for Simplification1}
The second term in the trace distance in $S_2$ can be expressed as
\begin{align*}
(\text{id}\tensor &\tilde{M}_{AB}^{(n)} ) (\Psi^\rho_{R^nA^nB^n}) =\frac{1}{N_1N_2}\sum_{\mu_1,\mu_2}\sum_{i,j} \Phi_{F^{(\mu)}(i,j)}\tensor \tr_{AB} \bigg \{ (\id\tensor \Gamma^{A, ( \mu_1)}_i\tensor \Gamma^{B, ( \mu_2)}_j ) \Psi^\rho_{R^nA^nB^n}   \bigg \}\\
&= \frac{1}{N_1N_2}\sum_{\mu_1,\mu_2}\sum_{i,j\geq 1} \sum_{(u^n, v^n)\in {\mathcal{B}}^{(\mu_1)}_1(i) \times {\mathcal{B}}^{(\mu_2)}_2(j) } \Phi_{e^{(\mu)}(u^n,v^n)} \tensor \tr_{AB} \bigg \{ (\id\tensor A^{(\mu_1)}_{u^n} \tensor B^{(\mu_2)}_{v^n}) \Psi^\rho_{R^nA^nB^n}    \bigg \}\\
&+\frac{1}{N_1N_2} \sum_{\mu_1,\mu_2} \sum_{\substack{j\geq 1}} \sum_{ v^n\in {\mathcal{B}}^{(\mu_2)}_2(j) }  \Phi_{(u_0^n,v_0^n)} \tensor \tr_{AB} \bigg \{ (\id\tensor (I-\sum_{u^n\in \TDeltan(U)} A^{(\mu_1)}_{u^n}) \tensor B^{(\mu_2)}_{v^n}) \Psi^\rho_{R^nA^nB^n}    \bigg \}\\
&+\frac{1}{N_1N_2} \sum_{\mu_1,\mu_2} \sum_{\substack{i\geq 1}} \sum_{ u^n\in {\mathcal{B}}^{(\mu_1)}_1(i) }  \Phi_{(u_0^n,v_0^n)} \tensor \tr_{AB} \bigg \{ (\id\tensor A^{(\mu_1)}_{u^n} \tensor (I-\sum_{v^n\in\TDeltan(V)} B^{(\mu_2)}_{v^n})) \Psi^\rho_{R^nA^nB^n}    \bigg \}\\\numberthis \label{eq:A_term}
&+\frac{1}{N_1N_2} \sum_{\mu_1,\mu_2}  \Phi_{(u_0^n,v_0^n)} \tensor \tr_{AB}\bigg \{ (\id\tensor (I-\sum_{u^n\in\TDeltan(U)} A^{(\mu_1)}_{u^n}) \tensor (I-\sum_{v^n\in\TDeltan(V)} B^{(\mu_2)}_{v^n})) \Psi^\rho_{R^nA^nB^n}   \bigg \}. 
\end{align*}
Similarly, for the first term in the trace distance in  $S_2$, we have
\begin{align*}
\frac{1}{N_1N_2} & \sum_{\mu_1,\mu_2} (\text{id}\tensor  [M_1^{( n, \mu_1)}]\tensor [M_2^{( n, \mu_2)}]) (\Psi^\rho_{R^nA^nB^n})\\
&=\frac{1}{N_1N_2} \sum_{\mu_1,\mu_2} \sum_{u^n \in \TDeltan(U)}\sum_{v^n\in \TDeltan(V) } \Phi_{(u^n,v^n)} \tensor \tr_{AB}\Big\{ (\id \tensor A_{u^n}^{(\mu_1)}\tensor B_{v^n}^{(\mu_2)})\Psi^\rho_{R^nA^nB^n} \Big\} \\
&+\frac{1}{N_1N_2} \sum_{\mu_1,\mu_2} \sum_{ v^n\in \TDeltan(V)}  \Phi_{(u_0^n,v_0^n)} \tensor \tr_{AB}\Big\{ (\id\tensor (I-\hspace{-10pt}\sum_{u^n\in \TDeltan(U)}\hspace{-10pt} A^{(\mu_1)}_{u^n}) \tensor B^{(\mu_2)}_{v^n}) \Psi^\rho_{R^nA^nB^n}   \Big\}\\
&+\frac{1}{N_1N_2} \sum_{\mu_1,\mu_2}  \sum_{ u^n\in \TDeltan(U)}  \Phi_{(u_0^n,v_0^n)} \tensor \tr_{AB}\Big\{ (\id\tensor A^{(\mu_1)}_{u^n} \tensor (I-\hspace{-10pt}\sum_{v^n\in\TDeltan(V)}\hspace{-10pt} B^{(\mu_2)}_{v^n})) \Psi^\rho_{R^nA^nB^n}   \Big\}\\ \numberthis \label{eq:B_term}
&+\frac{1}{N_1N_2} \sum_{\mu_1,\mu_2}  \Phi_{(u_0^n,v_0^n)} \tensor \tr_{AB}\Big\{ (\id\tensor (I-\hspace{-10pt}\sum_{u^n\in\TDeltan(U)}\hspace{-10pt} A^{(\mu_1)}_{u^n}) \tensor (I-\hspace{-10pt}\sum_{v^n\in\TDeltan(V)}\hspace{-10pt} B^{(\mu_2)}_{v^n})) \Psi^\rho_{R^nA^nB^n}   \Big\}.
 \end{align*}
By replacing the terms in $S_2$ using the corresponding expansions from \eqref{eq:A_term} and \eqref{eq:B_term}, we observe that the second, third and fourth terms on the right hand side of \eqref{eq:A_term} get canceled with the corresponding terms on the right hand side of  \eqref{eq:B_term}.
This simplifies $S_2$ as
\begin{align}\nonumber
     & \frac{1}{N_1N_2}\sum_{\mu_1,\mu_2} \sum_{u^n \in \TDeltan(U)} \sum_{v^n \in \TDeltan(V)}  \norm{    ( \Phi_{(u^n,v^n)}-  \Phi_{e^{(\mu)}(u^n,v^n)}) \tensor \tr_{AB}\Big\{ (\id \tensor A_{u^n}^{(\mu_1)}\tensor B_{v^n}^{(\mu_2)})\Psi^\rho_{R^nA^nB^n} \Big\}}_1\\ \nonumber
     &= \frac{1}{N_1N_2}\sum_{\mu_1,\mu_2} 
     \sum_{u^n \in \TDeltan(U)} \sum_{v^n \in \TDeltan(V)} 
     \big\| \Phi_{(u^n,v^n)}- 
     \Phi_{e^{(\mu)}(u^n,v^n)}\big\|_1\times  \norm{\tr_{AB}\Big\{ (\id \tensor A_{u^n}^{(\mu_1)}\tensor B_{v^n}^{(\mu_2)})\Psi^\rho_{R^nA^nB^n} \Big\} }_1\\
&=\frac{1}{N_1N_2}\sum_{\mu_1,\mu_2} 
\sum_{u^n \in \TDeltan(U)} \sum_{v^n \in \TDeltan(V)} 
 \big\| \Phi_{(u^n,v^n)}-  \Phi_{e^{(\mu)}(u^n,v^n)}\big\|_1\times  \tr \Big\{ (\id \tensor A_{u^n}^{(\mu_1)}\tensor B_{v^n}^{(\mu_2)})\Psi^\rho_{R^nA^nB^n} \Big\} \nonumber\\
&=  \frac{1}{N_1N_2}\sum_{\mu_1,\mu_2}  \sum_{u^n \in \TDeltan(U)} \sum_{v^n \in \TDeltan(V)}  \big\|\Phi_{u^n, v^n} -\Phi_{e^{(\mu)}(u^n,v^n)} \big\|_1 \gamma^{(\mu_1)}_{u^n}\zeta^{(\mu_2)}_{v^n}\Omega_{u^n,v^n}, \nonumber
\end{align} 
where the first two equalities are obtained by using the definition of trace norm and the last equality follows from the definition of $  A_{u^n}^{(\mu_1)}$  and $ B_{v^n}^{(\mu_2)} $  as in \eqref{eq:A_uB_v},
with $\Omega_{u^n,v^n}\deq \tr\Big\{ \sqrt{\rho^{\tensor n}_A\tensor \rho^{\tensor n}_B}^{-1}(\LambdauA\tensor\LambdavB)  \sqrt{\rho^{\tensor n}_A\tensor \rho^{\tensor n}_B}^{-1} \rho^{\tensor n}_{AB}\Big\}$. This completes the proof.

\subsection{Proof of Proposition \ref{prop:Lemma for S_3}}\label{appx:proof of S_3} 
From  Proposition \ref{prop:Lemma for Simplification1}, and using the definitions $\gamma^{(\mu_1)}_{u^n}$ and $\zeta^{(\mu_2)}_{v^n}$,  $S_2$ can be simplified as
\begin{align}
 S_2  = \frac{(1-\varepsilon)(1-\varepsilon')}{(1+\eta)^22^{n(\tilde{R}_1+\tilde{R}_2)}N_1N_2}\hspace{-3pt}\sum_{u^n,v^n}\sum_{\mu_1,\mu_2} \sum_{l,k}  \11\{ U^{n, (\mu_1)}(l)\!=\!u^n, V^{n, (\mu_2)}(k)\!=\!v^n\}  \Omega_{u^n,v^n} \|\Phi_{u^n, v^n} -\Phi_{e^{(\mu)}(u^n,v^n)}\|_1. \nonumber
\end{align} 
For any $(u^n,v^n)$, the 1-norm above can be bounded from above by the following quantity
  \begin{align*}
2\times \11\bigg\{\exists (\tilde{u}^n, \tilde{v}^n, i, j): &({u}^n, {v}^n) \in \mathcal{B}^{(\mu_1)}_1(i)\times \mathcal{B}^{(\mu_2)}_2(j),\\
& (\tilde{u}^n, \tilde{v}^n) \in \mathcal{C}^{(\mu)}\medcap \mathcal{T}_\delta^{(n)}(UV), (\tilde{u}^n, \tilde{v}^n) \in \mathcal{B}^{(\mu_1)}_1(i)\times \mathcal{B}^{(\mu_2)}_2(j), (\tilde{u}^n, \tilde{v}^n)\neq (u^n, v^n)\bigg\}.
\end{align*}
Denoting such an indicator function by $\11^{(\mu_1,\mu_2)}(u^n, v^n)$,  $S_2$ can be bounded from above as $S_2 \leq S_3$, where  
\begin{equation}\label{eq:binning prob bound}
S_3 \deq \frac{(1-\varepsilon)(1-\varepsilon')}{(1+\eta)^2 2^{n(\tilde{R}_1+\tilde{R}_2)}}\sum_{l,k}  \sum_{(u^n,v^n)}\Omega_{u^n,v^n}\frac{2}{N_1N_2}\sum_{\mu_1,\mu_2}\11^{(\mu_1,\mu_2)}(u^n, v^n)  \11\{ U^{n, (\mu_1)}(l)=u^n, V^{n, (\mu_2)}(k)=v^n\}. \nonumber
\end{equation}
Next, we use the Markov inequality to show that $S_3\leq \epsilon$ with probability sufficiently close to 1. We first show that the expectation of $S_3$ can be made arbitrary small by taking $n$ large enough. For that we take the expectation of the indicator functions with respect to random variables $U^n$ and $V^n$ which are independent of each other and distributed according to the pruned distribution, defined in \eqref{def:prunedDist}. This gives us, for $u^n \in \TDeltan(U)$ and 
$v^n \in \TDeltan(V)$, 
\begin{align}
    \EE\Big[\11^{(\mu_1,\mu_2)} & (u^n, v^n)   \11\left\{ U^{n, (\mu_1)}(l)=u^n, V^{n, (\mu_2)}(k)=v^n\right\}\Big]\nonumber\\
 &\leq  \hspace{-10pt}\sum_{\substack{(\tilde{u}^n, \tilde{v}^n)\in \mathcal{T}_\delta^{(n)}(UV)\nonumber\\ (\tilde{u}^n, \tilde{v}^n)\neq (u^n,v^n) }}\sum_{i,j} \sum_{(\tilde{l},\tilde{k})\neq (l,k)} \hspace{-10pt} \EE\left[\11\left\{ ({u}^n, {v}^n) \in \mathcal{B}^{(\mu_1)}_1(i)\times \mathcal{B}^{(\mu_2)}_2(j)\right\} \11\Big\{ (\tilde{u}^n, \tilde{v}^n) \in \mathcal{B}^{(\mu_1)}_1(i)\times \mathcal{B}^{(\mu_2)}_2(j)\Big\}\right.\nonumber\\
 &\left. \hspace{20pt }\times \11\left\{ U^{n, (\mu_1)}(l)=u^n, V^{n, (\mu_2)}(k)=v^n\right\}\11\left\{ U^{n, (\mu_1)}(\tilde{l})=\tilde{u}^n, V^{n, (\mu_2)}(\tilde{k})=\tilde{v}^n\right\}\right]  \nonumber\\
    &\leq  \frac{\lambdauA\lambdavB}{(1-\varepsilon)^2(1-\varepsilon')^2}
    2^{-n(I(U;V)-\delta_1)} \Big[2^{n(\tilde{R}_1-R_1)}2^{n(\tilde{R}_2-R_2)}\nonumber\\
    &+2^{n(\tilde{R}_1-R_1)}+2^{n(\tilde{R}_2-R_2)}+2^{-n(S(U)-\delta_1)} 2^{n\tilde{R}_1}2^{n(\tilde{R}_2-R_2)}+2^{-n(S(V)-\delta_1)} 2^{n\tilde{R}_2}2^{n(\tilde{R}_1-R_1)}\Big]\nonumber\\
    &\leq 5\frac{\lambdauA\lambdavB}{(1-\varepsilon)^2(1-\varepsilon')^2} 2^{-n(I(U;V)-2\delta_1)}2^{n(\tilde{R}_1-R_1)}2^{n(\tilde{R}_2-R_2)}, \label{eq:expecIndics}
\end{align}
\add{where 
$\delta_1\ssearrow 0$ as $\delta\ssearrow 0$}. {The first inequality follows from the union bound}. {The second inequality follows by evaluating the expectation of the indicator functions and the last inequality follows from the inequalities $\tilde{R}_1< S(U)$ and $\tilde{R}_2< S(V)$}. This implies
\begin{align*}
\EE[S_3]& \leq  10\frac{2^{-n(I(U;V)-2\delta_1)}2^{n(\tilde{R}_1-R_1)}2^{n(\tilde{R}_2-R_2)}}{(1+\eta)^2(1-\varepsilon)(1-\varepsilon')}\sum_{u^n \in \TDeltan(U)} \sum_{v^n \in \TDeltan(V)}  \Omega_{u^n,v^n}\lambdauA\lambdavB . 
\end{align*}

We proceed using the following lemma.
\begin{lem}\label{lem:omega lambda}
	For $ \lambdauA $ and $ \lambdavB $ as defined in \eqref{eq:dist_canonicalEnsemble} and $ \Omega_{u^n,v^n} $ defined above,
we have $$ \sum_{u^n \in \TDeltan(U)} \sum_{v^n \in \TDeltan(V)} \Omega_{u^n,v^n}\lambdauA\lambdavB \leq 2^{n\delta_{AB}},$$ for some $ \delta_{AB} \searrow 0 $ as $ \delta \searrow 0. $
\end{lem}
\begin{proof}
\label{appx:proof of omega lambda}
Firstly, note that 
\begin{align}
\sum_{\substack{u^n\in\TDeltaN(U)\\ v^n\in\TDeltaN(V)}}\hspace{-10pt}\Omega_{u^n,v^n}\lambdauA\lambdavB = \text{Tr}\bigg\{\bigg [\sqrt{\rho_{A}^{\tensor n}}^{-1}\Big(\hspace{-15pt}\sum_{u^n\in\TDeltaN(U)}\hspace{-15pt}\lambdauA\LambdauA\Big)\sqrt{\rho_{A}^{\tensor n}}^{-1}\hspace{-10pt}\tensor\sqrt{\rho_{B}^{\tensor n}}^{-1}\Big(\hspace{-15pt}\sum_{v^n \in \TDeltaN(V)}\hspace{-10pt}\lambdavB\LambdavB\Big) \sqrt{\rho_{B}^{\tensor n}}^{-1} \bigg] \rho^{\tensor n}_{AB}\bigg\}. \label{eq:omegalambdaUlambdaV}
\end{align} 
Consider,
\begin{align}
\sum_{u^n\in\TDeltaN(U)}\lambdauA\LambdauA & = \PihatA \Pi_{\rho_A}\left (\sum_{u^n\in\TDeltaN(U)}\lambdauA\PiuA\rhohatuA\PiuA\right )\Pi_{\rho_A}\PihatA \nonumber \\
& \leq \PihatA \Pi_{\rho_A}\left (\sum_{u^n}\lambdauA\rhohatuA\right )\Pi_{\rho_A}\PihatA \nonumber \\ &  = \PihatA \Pi_{\rho_A}\rho_{A}^{\tensor n}\Pi_{\rho_A}\PihatA \leq 2^{-n(S(\rho_A)-\delta_{A})}\PihatA\Pi_{\rho_A}\PihatA  = 2^{-n(S(\rho_A)-\delta_{A})}\Pi_{\rho_A}\PihatA\Pi_{\rho_A}. \nonumber
\end{align}
where the first inequality is obtained by using $ \PiuA\rhohatuA\PiuA \leq  \rhohatuA$ for all $ u^n \in \TDeltaN(U) $ and then by adding terms belonging to $\mathcal{U}^n\backslash \TDeltaN(U)$ into the summation. 
The subsequent inequality and the equality, follows from the properties of a typical projector with $ \delta_{A} \searrow 0 $ as $ \delta \searrow 0 $, and the commutativity of $ \PihatA $ and $ \Pi_{\rho_A} $, respectively.
This implies, 
\begin{align}
\sqrt{\rho_{A}^{\tensor n}}^{-1}\left(\sum_{u^n\in \TDeltaN(U)}\lambdauA\LambdauA\right)\sqrt{\rho_{A}^{\tensor n}}^{-1} \leq  2^{-n(S(\rho_A)-\delta_{A})}\sqrt{\rho_{A}^{\tensor n}}^{-1}\Pi_{\rho_A}\PihatA\Pi_{\rho_A}\sqrt{\rho_{A}^{\tensor n}}^{-1} \leq 2^{2n\delta_{A}}\PihatA, \label{eq:ineqrhoA}
\end{align}
where the last inequality again appeals to  the fact that $ \Pi_{\rho_{A}} $ and $ \PihatA $ commute. 
Similarly, using the same arguments above for the operators acting on $ \mathcal{H}_B $, we have
\begin{align} 
\sqrt{\rho_{B}^{\tensor n}}^{-1}\left(\sum_{v^n\in\TDeltaN(V)}\lambdavB\LambdavB\right)\sqrt{\rho_{B}^{\tensor n}}^{-1} \leq 2^{2n\delta_{B}}\PihatB, \label{eq:ineqrhoB}
\end{align}
where $ \delta_{B} \searrow 0 $ as $ \delta \searrow 0 $.
Using \eqref{eq:ineqrhoA} and \eqref{eq:ineqrhoB} in \eqref{eq:omegalambdaUlambdaV}, gives 
\begin{align}
\sum_{u^n,v^n}\Omega_{u^n,v^n}\lambdauA\lambdavB  \leq 2^{2n(\delta_{A}+\delta_{B})}\Tr{\left (\PihatB\tensor\PihatB\right ) \rho^{\tensor n}_{AB}} \leq 2^{2n(\delta_{A}+\delta_{B})}\Tr{ \rho^{\tensor n}_{AB}} = 2^{2n(\delta_{A}+\delta_{B})}, \nonumber
\end{align}
substituting $ \delta_{AB} = 2(\delta_{A}+\delta_{B}) $ gives the result.

\end{proof}
As a result, given any  $\epsilon\in (0,1)$, the above expectation can be made less than $ \frac{10 \epsilon}{(1-\varepsilon)(1-\varepsilon')}$  for large enough $n$ provided that $(\tilde{R}_1-R_1)+(\tilde{R}_2-R_2) \leq I(U;V) - 2\delta_1 - \delta_{AB} - \delta$. From Markov-inequality this implies that $S_3\leq \sqrt{\epsilon}$ with probability at least $1-\frac{10\sqrt{\epsilon}}{(1-\varepsilon)(1-\varepsilon')}$.

\subsection{Proof of Proposition \ref{prop:Lemma for S_234}}\label{appx:proof of S_234} 
We bound $\widetilde{S}$ as $\widetilde{S}\leq \widetilde{S}_2 +\widetilde{S}_3 +\widetilde{S}_4$, where
\begin{align*}
    \widetilde{S}_2 \deq &\left\|\frac{1}{N_1N_2}\sum_{\mu_1,\mu_2}\sum_{i>0} \sqrt{\rho_{AB}^{\tensor n}}\left (\Gamma^{A, ( \mu_1)}_i\tensor \Gamma^{B, (\mu_2)}_0\right )\sqrt{\rho_{AB}^{\tensor n}} P^n_{Z|U,V}(z^n|u^{n}_{0},v^{n}_{0})\right\|_1, \nonumber \\
\widetilde{S}_3 \deq  &\left\|\frac{1}{N_1N_2}\sum_{\mu_1,\mu_2}\sum_{j>0} \sqrt{\rho_{AB}^{\tensor n}}\left (\Gamma^{A, ( \mu_1)}_0\tensor \Gamma^{B, (\mu_2)}_j\right )\sqrt{\rho_{AB}^{\tensor n}} P^n_{Z|U,V}(z^n|u^{n}_{0},v^{n}_{0})\right\|_1, \nonumber \\
  \widetilde{S}_4 \deq  &\left\|\frac{1}{N_1N_2}\sum_{\mu_1,\mu_2} \sqrt{\rho_{AB}^{\tensor n}}\left (\Gamma^{A, ( \mu_1)}_0\tensor \Gamma^{B, (\mu_2)}_0\right )\sqrt{\rho_{AB}^{\tensor n}} P^n_{Z|U,V}(z^n|u^{n}_{0},v^{n}_{0})\right\|_1. \nonumber  
\end{align*}
\noindent{\bf Analysis of $\widetilde{ S}_2 $:} We have 
\begin{align}
\widetilde{S}_{2} & \leq  \frac{1}{N_1N_2}\sum_{\mu_1,\mu_2}\sum_{z^n}P^n_{Z|U,V}(z^n|u_0^n,v_0^n)\left\| \sqrt{\rho_{AB}^{\tensor n}}\left (\sum_{i>0}\Gamma^{A, ( \mu_1)}_i\tensor \Gamma^{B, (\mu_2)}_0\right )\sqrt{\rho_{AB}^{\tensor n}} \right\|_1 \nonumber \\
& \leq  \frac{1}{N_1N_2}\sum_{\mu_1,\mu_2}\sum_{u^n}\left\| \sqrt{\rho_{AB}^{\tensor n}}\left (A_{u^n}^{(\mu_1)}\tensor \Gamma^{B, (\mu_2)}_0\right )\sqrt{\rho_{AB}^{\tensor n}} \right\|_1 \nonumber \\
& \stackrel{w.h.p}{\leq}  \frac{1}{N_1N_2}\sum_{\mu_1,\mu_2}\left\| \sqrt{\rho_{B}^{\tensor n}} \Gamma^{B, (\mu_2)}_0\sqrt{\rho_{B}^{\tensor n}} \right\|_1 \nonumber \\
& = \frac{1}{N_2}\sum_{\mu_2}\left\| \sum_{v^n}\lambdavB\rhohatvB - \sum_{v^n}\sqrt{\rho_{B}^{\tensor n}} B^{(\mu_2)}_{v^n}\sqrt{\rho_{B}^{\tensor n}} \right\|_1 \nonumber \\
& \leq \frac{1}{N_2}\sum_{\mu_2}\left\| \sum_{v^n}\lambdavB\rhohatvB - \cfrac{(1-\varepsilon')}{(1+\eta)}\cfrac{1}{2^{n\tilde{R}_2}}\sum_{k=1}^{2^{n\tilde{R}_2}}\hat{\rho}^{B}_{V^{n,(\mu_2)}_{k}} \right\|_1 +\underbrace{ \frac{1}{N_2}\sum_{\mu_2}\sum_{v^n}\zeta_{v^n}^{(\mu_2)}\left\| \rhohatvB -  \LambdavB \right\|_1}_{\widetilde{S}_{22}}, \label{eq:S2simplification}
\end{align}
where the first inequality uses triangle inequality, the second uses the fact that $ \sum_{i>0}\Gamma^{A, ( \mu_1)}_i = \sum_{u^n}A_{u^n}^{(\mu_1)} $. The next inequality follows by using Lemma \ref{lem:Separate} where we use the result that with high probability 
(letting $E_1$ denote this event) we have $ \sum_{u^n}A_{u^n}^{(\mu_1)} \leq I $, given that $ \tilde{R}_1 \geq I(U;RB)_{\sigma_{2}}. $
Finally, the last inequality follows again from triangle inequality. 


\noindent Regarding the first term in \eqref{eq:S2simplification} 
using Lemma \ref{lem:nf_SoftCovering} we claim that for all sufficiently large $ n $, the term can be made arbitrarily small with high probability
(letting $E_2$ denote this event), given the rate $ \tilde{R}_2 $ satisfies $ \tilde{R}_2 \geq I(V;RA)_{\sigma_{2}} $ where $ \sigma_{2} $ is as defined in the statement of the theorem. Note that the  requirements we obtain on $ \tilde{R}_1$ and $ \tilde{R}_2$ here were already imposed earlier
in Section \ref{subsec:stochastic_dist_POVM}.
And as for the second term we use the gentle measurement lemma (as in \eqref{eq:gentleMeasurementVn}) and bound its expected value as
\begin{align}
\EE\left[\frac{1}{N_2}\sum_{\mu_2}\sum_{v^n}\zeta_{v^n}^{(\mu_2)}\left\| \rhohatvB -  \LambdavB \right\|_1\right]  = \sum_{v^n \in \TDeltaN(V)}\cfrac{\lambdavB}{(1+\eta)}\left\| \rhohatvB -  \LambdavB \right\|_1 \leq \epsilon_{\scriptscriptstyle \widetilde{S}_{2}},\nonumber
\end{align}
where the inequality is based on the repeated usage of the average gentle measurement lemma by setting $ \epsilon_{\scriptscriptstyle\widetilde{S}_{2}} = \frac{(1-\varepsilon')}{(1+\eta)} (2\sqrt{\varepsilon'_B} + 2\sqrt{\varepsilon''_B}) $ with  $ \epsilon_{\scriptscriptstyle \widetilde{S}_{2}} \searrow 0 $ as $ \delta \searrow 0 $ and $ \varepsilon'_B = \varepsilon' + 2\sqrt{\varepsilon'} $ and $ \varepsilon''_B = 2\varepsilon' + 2\sqrt{\varepsilon'} $ (see (35) in \cite{wilde_e} for more details ).
Now, by using Markov inequality $\PP(E_3) \leq \sqrt{\epsilon_{\widetilde{S}_2}}$, where 
$E_3 \deq \{\widetilde{S}_{22} \geq \sqrt{\epsilon_{\widetilde{S}_2}}\}$. Hence, using union bound on the three events $E_1, E_2$ and $E_3$, $\widetilde{S}_2$ can be made arbitrarily small, for sufficiently large $n$, with high probability.

\noindent{\bf Analysis of $\widetilde{S}_3 $:} Due to the symmetry in $\widetilde{S}_2 $ and $ \widetilde{S}_3 $, the analysis of $ \widetilde{S}_{3} $ follows very similar arguments as that of $\widetilde{S}_2 $ and hence we skip it.

\noindent{\bf Analysis of $ \widetilde{S}_4$:} We have 
\begin{align}
\widetilde{S}_{4} & \leq  \frac{1}{N_1N_2}\sum_{\mu_1,\mu_2}\sum_{z^n}P^n_{Z|U,V}(z^n|u^{n}_{0},v^{n}_{0})\left\| \sqrt{\rho_{AB}^{\tensor n}}\left (\Gamma^{A, ( \mu_1)}_0\tensor \Gamma^{B, (\mu_2)}_0\right )\sqrt{\rho_{AB}^{\tensor n}} \right\|_1  \nonumber \\
& \leq  \frac{1}{N_1N_2}\sum_{\mu_1,\mu_2}\left\| \sqrt{\rho_{AB}^{\tensor n}}\left (\Gamma^{A, ( \mu_1)}_0\tensor I\right )\sqrt{\rho_{AB}^{\tensor n}} \right\|_1  +  \frac{1}{N_1N_2}\sum_{\mu_1,\mu_2}\sum_{v^n}\left\| \sqrt{\rho_{AB}^{\tensor n}}\left (\Gamma^{A, ( \mu_1)}_0\tensor B_{v^n}^{(\mu_{2})}\right )\sqrt{\rho_{AB}^{\tensor n}} \right\|_1, \label{eq:S4inequalities}
\end{align}
where the inequalities above are obtained by a straight forward substitution and use of triangle inequality.
With the above constraints on $ \tilde{R}_1 $ and $ \tilde{R}_2 $, we have $ 0 \le \Gamma^{A, ( \mu_1)}_0 \leq I $  and  $ 0 \leq \Gamma^{B, (\mu_2)}_0 \leq I $. This simplifies the first term in \eqref{eq:S4inequalities} as 
\begin{align}
\frac{1}{N_1N_2}\sum_{\mu_1,\mu_2}\left\| \sqrt{\rho_{AB}^{\tensor n}}\left (\Gamma^{A, ( \mu_1)}_0\tensor I\right )\sqrt{\rho_{AB}^{\tensor n}} \right\|_1  = \frac{1}{N_1}\sum_{\mu_1}\left\| \sqrt{\rho_{A}^{\tensor n}}\left (\Gamma^{A, ( \mu_1)}_0\right )\sqrt{\rho_{A}^{\tensor n}} \right\|_1. \nonumber
\end{align}
 Similarly, the second term in \eqref{eq:S4inequalities} simplifies using Lemma \ref{lem:Separate} as
 \begin{align}
 \frac{1}{N_1N_2}\sum_{\mu_1,\mu_2}\sum_{v^n}\left\| \sqrt{\rho_{AB}^{\tensor n}}\left (\Gamma^{A, ( \mu_1)}_0\tensor B_{v^n}^{(\mu_{2})}\right )\sqrt{\rho_{AB}^{\tensor n}} \right\|_1 \leq \frac{1}{N_1}\sum_{\mu_1}\left\| \sqrt{\rho_{A}^{\tensor n}}\left (\Gamma^{A, ( \mu_1)}_0\right )\sqrt{\rho_{A}^{\tensor n}} \right\|_1.\nonumber
 \end{align}
 Using these simplifications, we have
 \begin{align}
 \widetilde{S}_4 & \leq \frac{2}{N_1}\sum_{\mu_1}\left\| \sqrt{\rho_{A}^{\tensor n}}\left (\Gamma^{A, ( \mu_1)}_0\right )\sqrt{\rho_{A}^{\tensor n}} \right\|_1. \nonumber
 \end{align}
 The above expression is similar to the one obtained in the simplification of $\widetilde{S}_2 $ and hence we can bound $\widetilde{S}_4 $ using the same constraints as $ \widetilde{S}_2 $, for sufficiently large $ n $.
 
 \subsection{Proof of Proposition \ref{prop:Lemma for S_12} }\label{appx:proof of S_12}
Recalling $ S_{2} $, we have 
\begin{align}
S_{2}& \leq \frac{1}{N_1N_2} \sum_{\mu_1,\mu_2}\sum_{z^n}\sum_{u^n,v^n}\left|P^n_{Z|U,V}(z^n|u^n,v^n)- P^n_{Z|U,V}\left (z^n|e^{(\mu)}(u^n,v^n)\right )\right| \left \|   \sqrt{\rho_{AB}^{\tensor n}}\left (A_{u^n}^{(\mu_1)} \tensor B_{v^n}^{(\mu_2)}\right ) \sqrt{\rho_{AB}^{\tensor n}} \right \|_1 \nonumber \\
& \leq \frac{1}{N_1N_2} \sum_{\mu_1,\mu_2}\sum_{u^n,v^n}\sum_{z^n}\left|P^n_{Z|U,V}(z^n|u^n,v^n)- P^n_{Z|U,V}\left (z^n|e^{(\mu)}(u^n,v^n)\right )\right| \gamma_{u^n}^{(\mu_1)}\zeta_{v^n}^{(\mu_2)}\Omega_{u^n,v^n} \nonumber \\
& \leq \cfrac{(1-\varepsilon)(1-\varepsilon')}{(1+\eta)^2}\frac{2}{N_1N_2} 2^{-n(\tilde{R}_1+\tilde{R}_2)} \sum_{\mu_1,\mu_2}\sum_{l,k} \sum_{u^n,v^n}\Omega_{u^n,v^n} \mathbbm{1}^{(\mu_1,\mu_2)}(u^n,v^n)\IndiU1\IndiV1, \nonumber
\end{align}
where $ \Omega_{u^n,v^n} $ and $ \mathbbm{1}^{(\mu_1,\mu_2)}(u^n,v^n)$ are defined as
\begin{align*}
    \Omega_{u^n,v^n}& \deq \tr\Big\{ \sqrt{\rho^{\tensor n}_A\tensor \rho^{\tensor n}_B}^{-1}(\LambdauA\tensor\LambdavB)  \sqrt{\rho^{\tensor n}_A\tensor \rho^{\tensor n}_B}^{-1} \rho^{\tensor n}_{AB}\Big\},\\
     \mathbbm{1}^{(\mu_1,\mu_2)}(u^n,v^n) & \deq \11\bigg\{\exists (\tilde{u}^n, \tilde{v}^n, i, j): ({u}^n, {v}^n) \in \mathcal{B}^{(\mu_1)}_1(i)\times \mathcal{B}^{(\mu_2)}_2(j),\\
& \hspace{20pt} (\tilde{u}^n, \tilde{v}^n) \in \mathcal{C}^{(\mu)}\medcap \mathcal{T}_\delta^{(n)}(UV), (\tilde{u}^n, \tilde{v}^n) \in \mathcal{B}^{(\mu_1)}_1(i)\times \mathcal{B}^{(\mu_2)}_2(j), (\tilde{u}^n, \tilde{v}^n)\neq (u^n, v^n)\bigg\}.
\end{align*}
We know from the simplification in \eqref{eq:expecIndics} that  
\begin{align}
\EE\left [\mathbbm{1}^{(\mu_1,\mu_2)}(u^n,v^n) \IndiU1\IndiV1 \right ] \nonumber \leq \frac{5~\lambdauA\lambdavB}{(1-\varepsilon)^2(1-\varepsilon')^2} 2^{-n(I(U;V)-2\delta_1)}2^{n(\tilde{R}_1-R_1)}2^{n(\tilde{R}_2-R_2)}. \nonumber
\end{align}
Substituting this in the expression for $ S_{12} $ gives
\begin{align}
\EE{[S_{2}]} & \leq 10 \cfrac{2^{-n(I(U;V)-2\delta_1)}2^{n(\tilde{R}_1-R_1)}2^{n(\tilde{R}_2-R_2)}}{(1+\eta)^2(1-\varepsilon)^2(1-\varepsilon')^2} \sum_{u^n,v^n}\Omega_{u^n,v^n}\lambdauA\lambdavB \nonumber \\
& \leq 10 \cfrac{2^{-n(I(U;V)-2\delta_1-\delta_{AB})}2^{n(\tilde{R}_1-R_1)}2^{n(\tilde{R}_2-R_2)}}{(1+\eta)^2(1-\varepsilon)^2(1-\varepsilon')^2}, \nonumber 
\end{align}
where the second inequality above uses Lemma \ref{lem:omega lambda}. Therefore, if $ \tilde{R}_{1} + \tilde{R}_{2} - R_1 -R_2 \leq I(U;V)_{\sigma_{3}} - 2\delta_1 -\delta_{AB} -\delta $, then we have
$ \EE{[S_{2}]} \leq 10\frac{ 2^{-n\delta}}{(1+\eta)^2(1-\varepsilon)(1-\varepsilon')} $. The proposition follows from Markov Inequality.

\subsection{Proof of Proposition \ref{prop:Lemma for Q2}}\label{appx:proof of Q2} 
We start by adding and subtracting the following terms in $ Q_2 $ 
\begin{align}
(i) &\sum_{u^n,v^n}\sqrt{\rho_{AB}^{\tensor n}}  \left(\bar{\Lambda}^A_{u^n}\tensor\bar{\Lambda}^B_{v^n} \right )\sqrt{\rho_{AB}^{\tensor n}} P^n_{Z|U,V}(z^n|u^n,v^n) \nonumber \\
(ii)&\sum_{u^n,v^n}\frac{1}{N_2}\sum_{\mu_2=1}^{N_2}\sqrt{\rho_{AB}^{\tensor n}}  \left( \bar{\Lambda}^A_{u^n} \tensor \cfrac{\zeta^{(\mu_2)}_{v^n}}{\lambdavB}\bar{\Lambda}^B_{v^n} \right )\sqrt{\rho_{AB}^{\tensor n}} P^n_{Z|U,V}(z^n|u^n,v^n) \nonumber \\
(iii)&\sum_{u^n,v^n}  \frac{1}{N_1N_2} \sum_{\mu_1,\mu_2}\sqrt{\rho_{AB}^{\tensor n}}\left(  A_{u^n}^{(\mu_1)} \tensor \cfrac{\zeta^{(\mu_2)}_{v^n}}{\lambdavB}\bar{\Lambda}^B_{v^n} \right )\sqrt{\rho_{AB}^{\tensor n}} P^n_{Z|U,V}(z^n|u^n,v^n). \nonumber
\end{align}
This gives us $ Q_2 \leq Q_{21} + Q_{22} + Q_{23} + Q_{24} $, where
\begin{align}
Q_{21} & \deq \sum_{z^n} \left\|\sum_{u^n,v^n}\sqrt{\rho_{AB}^{\tensor n}}  \left ( \left (\frac{1}{N_1} \sum_{\mu_1=1}^{N_1}A_{u^n}^{(\mu_1)}\right ) \tensor \bar{\Lambda}^B_{v^n} -\bar{\Lambda}^A_{u^n}\tensor\bar{\Lambda}^B_{v^n} \right )\sqrt{\rho_{AB}^{\tensor n}} P^n_{Z|U,V}(z^n|u^n,v^n)\right\|_1, \nonumber\\
Q_{22} & \deq \sum_{z^n} \left\| \sum_{u^n,v^n}\sqrt{\rho_{AB}^{\tensor n}}  \left (\bar{\Lambda}^A_{u^n}\tensor\bar{\Lambda}^B_{v^n}  - \bar{\Lambda}^A_{u^n} \tensor\left (\frac{1}{N_2}\sum_{\mu_2=1}^{N_2} \cfrac{\zeta^{(\mu_2)}_{v^n}}{\lambdavB}\bar{\Lambda}^B_{v^n}\right )\right )\sqrt{\rho_{AB}^{\tensor n}} P^n_{Z|U,V}(z^n|u^n,v^n)\right\|_1, \nonumber	\\
Q_{23} & \deq \sum_{z^n} \left\| \sum_{u^n,v^n}\!\sqrt{\rho_{AB}^{\tensor n}}  \left (\bar{\Lambda}^A_{u^n} \tensor\left (\frac{1}{N_2}\!\sum_{\mu_2=1}^{N_2}\! \cfrac{\zeta^{(\mu_2)}_{v^n}}{\lambdavB}\bar{\Lambda}^B_{v^n}\right ) - \frac{1}{N_1N_2} \sum_{\mu_1,\mu_2}\!\!A_{u^n}^{(\mu_1)} \tensor \cfrac{\zeta^{(\mu_2)}_{v^n}}{\lambdavB}\bar{\Lambda}^B_{v^n} \!\right)\sqrt{\rho_{AB}^{\tensor n}} P^n_{Z|U,V}(z^n|u^n,v^n)\right\|_1, \nonumber\\
Q_{24} & \deq \sum_{z^n} \left\| \sum_{u^n,v^n}\frac{1}{N_1N_2} \sum_{\mu_1,\mu_2}\sqrt{\rho_{AB}^{\tensor n}}  \left( A_{u^n}^{(\mu_1)} \tensor \cfrac{\zeta^{(\mu_2)}_{v^n}}{\lambdavB}\bar{\Lambda}^B_{v^n} - A_{u^n}^{(\mu_1)} \tensor B_{v^n}^{(\mu_2)} \right)\sqrt{\rho_{AB}^{\tensor n}} P^n_{Z|U,V}(z^n|u^n,v^n)\right\|_1. \nonumber
\end{align}

We start by analyzing $ Q_{21} $. Note that $ Q_{21} $ is exactly same as $ Q_{1} $ and hence using the same rate constraints as $ Q_{1} $, this term can be bounded. Next, consider $ Q_{22} $. Substitution of $ \zeta^{(\mu_2)}_{v^n} $ gives
\begin{align}
Q_{22} & = \left\| \sum_{u^n,v^n,z^n}\lambdaAB\rhohatuAvB \tensor P^n_{Z|U,V}(z^n|u^n,v^n)\ketbra{z^n} \right .\nonumber \\ 
& \hspace{20pt} \left . - \cfrac{(1-\varepsilon')}{(1+\eta)}\cfrac{1}{2^{n(\tilde{R}_2 + C_2)}}\sum_{\mu_2,k} \sum_{u^n,v^n,z^n}\IndiV1\cfrac{\lambdaAB}{\lambdavB} \rhohatuAvB \tensor P^n_{Z|U,V}(z^n|u^n,v^n)\ketbra{z^n}\right\|_1, \nonumber
\end{align}
where the equality uses the triangle inequality for block operators.
From here on, we use Lemma \ref{lem:nf_SoftCovering} to bound $Q_{22}$.
\begin{align}
 \mathcal{T}_{v^n} = \sum_{u^n,z^n}\cfrac{\lambdaAB}{\lambdavB} \rhohatuAvB \tensor P^n_{Z|U,V}(z^n|u^n,v^n)\ketbra{z^n}, \nonumber
\end{align}
Note that $ \mathcal{T}_{v^n} $ can be written in tensor product form as
$ \mathcal{T}_{v^n} = \bigtensor_{i=1}^{n}\mathcal{T}_{v_{i}} $. This simplifies $ Q_{22} $ as
\begin{align}
Q_{22} & =  \left\| \sum_{v^n}\lambdavB\mathcal{T}_{v^n}  - \cfrac{(1-\varepsilon')}{(1+\eta)}\cfrac{1}{2^{n(\tilde{R}_2 + C_2)}}\sum_{\mu_2,k} \mathcal{T}_{V^{n,(\mu_{2})}(k)}\right\|_1.\nonumber
\end{align}
Lemma \ref{lem:nf_SoftCovering} gives us, for any given $ \epsilon_{\scriptscriptstyle Q_{22}}, \delta_{\scriptscriptstyle Q_{22}} \in (0,1)$, if  
\begin{align}
\tilde{R}_2+C_2 > S\left (\sum_{v \in \mathcal{V}}\lambda_v^B\mathcal{T}_{v}\right ) - \sum_{v \in \mathcal{V}}\lambda_v^B S(\mathcal{T}_{v}) = I(RZ;V)_{\sigma_4},
\end{align}
then $ \PP(Q_{22} \geq \epsilon_{\scriptscriptstyle Q_{22}} ) \leq \delta_{\scriptscriptstyle Q_{22}} $ for sufficiently large $n$.

\noindent Now, we move on to consider $ Q_{23} $. Taking expectation with respect to the codebook $\mathcal{C}^{(\mu)} = (\mathcal{C}_1^{(\mu_1)}, \mathcal{C}_2^{(\mu_2)})$  gives
%
\begin{align}
\EE\left[Q_{23}\right ] &\leq \EE_{\mathcal{C}}\left[\sum_{z^n, v^n}\frac{1}{N_2}\sum_{\mu_2=1}^{N_2} \cfrac{\zeta^{(\mu_2)}_{v^n}}{\lambdavB} \left\| \sum_{u^n}\sqrt{\rho_{AB}^{\tensor n}} \left (\bar{\Lambda}^A_{u^n} \tensor\bar{\Lambda}^B_{v^n}\right)\sqrt{\rho_{AB}^{\tensor n}}P^n_{Z|U,V}(z^n|u^n,v^n)   \right .\right . \nonumber \\ 
& \hspace{100pt} \left.\left . - \sum_{u^n}\sqrt{\rho_{AB}^{\tensor n}} \left(\frac{1}{N_1} \sum_{\mu_1}A_{u^n}^{(\mu_1)} \tensor \bar{\Lambda}^B_{v^n} \right)\sqrt{\rho_{AB}^{\tensor n}} P^n_{Z|U,V}(z^n|u^n,v^n)\right\|_1\right] \nonumber\\
&= \EE_{\mathcal{C}_1}\left[\sum_{z^n, v^n}\frac{1}{N_2}\sum_{\mu_2=1}^{N_2} \cfrac{\EE_{\mathbb{C}_2}\left[\zeta^{(\mu_2)}_{v^n}\right ]}{\lambdavB} \left\| \sum_{u^n}\sqrt{\rho_{AB}^{\tensor n}} \left (\bar{\Lambda}^A_{u^n} \tensor\bar{\Lambda}^B_{v^n}\right)\sqrt{\rho_{AB}^{\tensor n}}P^n_{Z|U,V}(z^n|u^n,v^n)   \right . \right .\nonumber \\ 
& \hspace{100pt} \left .\left. - \sum_{u^n}\sqrt{\rho_{AB}^{\tensor n}} \left(\frac{1}{N_1} \sum_{\mu_1}A_{u^n}^{(\mu_1)} \tensor \bar{\Lambda}^B_{v^n} \right)\sqrt{\rho_{AB}^{\tensor n}} P^n_{Z|U,V}(z^n|u^n,v^n)\right\|_1\right] \nonumber\\
&= \EE_{\mathcal{C}_1}\left[\sum_{z^n, v^n}\cfrac{1}{(1+\eta)} \left\| \sum_{u^n}\sqrt{\rho_{AB}^{\tensor n}} \left (\bar{\Lambda}^A_{u^n} \tensor\bar{\Lambda}^B_{v^n}\right)\sqrt{\rho_{AB}^{\tensor n}}P^n_{Z|U,V}(z^n|u^n,v^n)   \right . \right .\nonumber \\ 
& \hspace{100pt} \left .\left. - \sum_{u^n}\sqrt{\rho_{AB}^{\tensor n}} \left(\frac{1}{N_1} \sum_{\mu_1}A_{u^n}^{(\mu_1)} \tensor \bar{\Lambda}^B_{v^n} \right)\sqrt{\rho_{AB}^{\tensor n}} P^n_{Z|U,V}(z^n|u^n,v^n)\right\|_1\right] \nonumber\\
&= \EE\left[\cfrac{J}{(1+\eta)}\right], \nonumber
\end{align}
where the inequality above is obtained by using the triangle inequality, and the first equality follows as $\mathcal{C}_1^{(\mu_1)}$ and $\mathcal{C}_2^{(\mu_2)}$ are generated independently. 
The last equality follows from the definition of $ J $ as in \eqref{def:J}. Hence, we use the result obtained in bounding $ \EE[J]. $

Finally, we consider $ Q_{24} $.
\begin{align}
Q_{24} & \leq \sum_{u^n,v^n}\sum_{z^n} P^n_{Z|U,V}(z^n|u^n,v^n)\left\| \frac{1}{N_1N_2} \sum_{\mu_1,\mu_2}\sqrt{\rho_{AB}^{\tensor n}}  \left( A_{u^n}^{(\mu_1)} \tensor \cfrac{\zeta^{(\mu_2)}_{v^n}}{\lambdavB}\bar{\Lambda}^B_{v^n}\right )\sqrt{\rho_{AB}^{\tensor n}} \right .\nonumber \\
& \hspace{70pt} - \left .\frac{1}{N_1N_2} \sum_{\mu_1,\mu_2}\sqrt{\rho_{AB}^{\tensor n}}\left( A_{u^n}^{(\mu_1)} \tensor \zeta^{(\mu_2)}_{v^n} \left (\sqrt{\rho_{B}}^{-1}\LambdavB\sqrt{\rho_{B}}^{-1}\right ) \right)\sqrt{\rho_{AB}^{\tensor n}} \right\|_1 \nonumber \\
& \leq \frac{1}{N_2} \sum_{\mu_2}\sum_{u^n,v^n}\zeta^{(\mu_2)}_{v^n}\left\| \sqrt{\rho_{AB}^{\tensor n}}  \left( \frac{1}{N_1} \sum_{\mu_1}A_{u^n}^{(\mu_1)} \tensor \cfrac{1}{\lambdavB}\bar{\Lambda}^B_{v^n}\right )\sqrt{\rho_{AB}^{\tensor n}} \right .\nonumber \\
& \hspace{70pt} - \left .\sqrt{\rho_{AB}^{\tensor n}}\left( \frac{1}{N_1} \sum_{\mu_1}A_{u^n}^{(\mu_1)} \tensor  \left (\sqrt{\rho_{B}}^{-1}\LambdavB\sqrt{\rho_{B}}^{-1}\right ) \right)\sqrt{\rho_{AB}^{\tensor n}} \right\|_1, \nonumber
\end{align}
where the inequalities above are obtained by substituting in the definition of $ B_{v^n}^{(\mu_2)} $  and using multiple triangle inequalities.
Taking expectation of $ Q_{24} $ with respect to the second codebook generation, we get
\begin{align}
\EE_{\mathcal{C}_2}{\left[Q_{24}\right]} & 
\leq \sum_{u^n}\sum_{\substack{v^n\in \TDeltaN(V)}}\cfrac{\lambdavB}{(1+\eta)}\Bigg\| \sqrt{\rho_{AB}^{\tensor n}}  \Big( \frac{1}{N_1} \sum_{\mu_1}A_{u^n}^{(\mu_1)} \tensor \cfrac{1}{\lambdavB}\bar{\Lambda}^B_{v^n}\Big)\sqrt{\rho_{AB}^{\tensor n}} \nonumber \\
& \hspace{100pt} -  \sqrt{\rho_{AB}^{\tensor n}}\Big( \frac{1}{N_1} \sum_{\mu_1}A_{u^n}^{(\mu_1)} \tensor  \left (\sqrt{\rho_{B}}^{-1}\LambdavB\sqrt{\rho_{B}}^{-1}\Big ) \right)\sqrt{\rho_{AB}^{\tensor n}} \Bigg\|_1 \nonumber \\
& \stackrel{w.h.p}{\leq}  \sum_{v^n\in \TDeltaN(V)}\cfrac{\lambdavB}{(1+\eta)}  \left \| \sqrt{\rho_{B}^{\tensor n}}  \left (\cfrac{1}{\lambdavB}\bar{\Lambda}^B_{v^n} - \sqrt{\rho_{B}}^{-1}\LambdavB\sqrt{\rho_{B}}^{-1} \right)\sqrt{\rho_{B}^{\tensor n}} \right \|_1 \nonumber \\
& = \sum_{v^n\in \TDeltaN(V)}\cfrac{\lambdavB}{(1+\eta)}  \left \|\rhohatvB - \LambdavB \right \|_1\leq \cfrac{(1-\varepsilon')}{(1+\eta)} (2\sqrt{\varepsilon'_B} + 2\sqrt{\varepsilon''_B}) = \epsilon_{\scriptscriptstyle Q_{24}}, \label{eq:gentleMeasurementVn}
\end{align}
where the second inequality above follows by using Lemma \ref{lem:Separate} and the fact that $ \sum_{ u^n}\frac{1}{N_1} \sum_{\mu_1}A_{u^n}^{(\mu_1)} \leq I,$ with high probability, 
and the last inequality uses the result based on the average gentle measurement lemma by setting $ \epsilon_{\scriptscriptstyle Q_{24}} = \frac{(1-\varepsilon')}{(1+\eta)} (2\sqrt{\varepsilon'_B} + 2\sqrt{\varepsilon''_B}) $ with  $ \epsilon_{\scriptscriptstyle Q_{24}} \searrow 0 $ as $ \delta \searrow 0 $ and $ \varepsilon'_B = \varepsilon' + 2\sqrt{\varepsilon'} $ and $ \varepsilon''_B = 2\varepsilon' + 2\sqrt{\varepsilon'} $ (see (35) in \cite{wilde_e} for more details ). This completes the proof for $ Q_{24} $ and hence for all the terms corresponding to $ Q_{2} $.

%% file: main.bbl
\begin{thebibliography}{10}
\providecommand{\url}[1]{#1}
\csname url@samestyle\endcsname
\providecommand{\newblock}{\relax}
\providecommand{\bibinfo}[2]{#2}
\providecommand{\BIBentrySTDinterwordspacing}{\spaceskip=0pt\relax}
\providecommand{\BIBentryALTinterwordstretchfactor}{4}
\providecommand{\BIBentryALTinterwordspacing}{\spaceskip=\fontdimen2\font plus
\BIBentryALTinterwordstretchfactor\fontdimen3\font minus
  \fontdimen4\font\relax}
\providecommand{\BIBforeignlanguage}[2]{{%
\expandafter\ifx\csname l@#1\endcsname\relax
\typeout{** WARNING: IEEEtran.bst: No hyphenation pattern has been}%
\typeout{** loaded for the language `#1'. Using the pattern for}%
\typeout{** the default language instead.}%
\else
\language=\csname l@#1\endcsname
\fi
#2}}
\providecommand{\BIBdecl}{\relax}
\BIBdecl

\bibitem{winter}
A.~Winter, ``{''Extrinsic'' and ''intrinsic'' data in quantum measurements:
  asymptotic convex decomposition of positive operator valued measures},''
  \emph{Communication in Mathematical Physics}, vol. 244, no.~1, pp. 157--185,
  2004.

\bibitem{wilde_e}
M.~M. Wilde, P.~Hayden, F.~Buscemi, and M.-H. Hsieh, ``The
  information-theoretic costs of simulating quantum measurements,''
  \emph{Journal of Physics A: Mathematical and Theoretical}, vol.~45, no.~45,
  p. 453001, 2012.

\bibitem{bennett2009quantum}
C.~H. Bennett, I.~Devetak, A.~W. Harrow, P.~W. Shor, and A.~Winter, ``Quantum
  reverse shannon theorem,'' \emph{arXiv preprint arXiv:0912.5537}, 2009.

\bibitem{pcc}
I.~\vspace{0mm}Devetak, ``The private classical capacity and quantum capacity
  of a quantum channel,'' \emph{IEEE Transactions on Information Theory},
  vol.~51, no.~1, pp. 44--55, 2005b.

\bibitem{qtc}
N.~Datta, M.-H. Hsieh, M.~M. Wilde, and A.~Winter, ``Quantum-to-classical rate
  distortion coding,'' \emph{Journal of Mathematical Physics}, vol.~54, no.~4,
  p. 042201, 2013.

\bibitem{datta}
N.~Datta, M.~H. Hsieh, and M.~M. Wilde, ``{Quantum rate distortion, reverse
  Shannon theorems, and source-channel separation},'' \emph{IEEE Transactions
  on Information Theory}, vol.~59, pp. 615--630, 2013.

\bibitem{devetak2003classical}
I.~Devetak and A.~Winter, ``Classical data compression with quantum side
  information,'' \emph{Physical Review A}, vol.~68, no.~4, p. 042301, 2003.

\bibitem{anshu2019convex}
A.~Anshu, R.~Jain, and N.~A. Warsi, ``Convex-split and hypothesis testing
  approach to one-shot quantum measurement compression and randomness
  extraction,'' \emph{IEEE Transactions on Information Theory}, vol.~65, no.~9,
  pp. 5905--5924, 2019.

\bibitem{berger}
T.~Berger, ``{Multiterminal source coding},'' \emph{The Inform. Theory Approach
  to Communications, G. Longo, Ed., New York: Springer-Verlag}, 1977.

\bibitem{sideInf}
I.~Devetak and A.~Winter, ``Classical data compression with quantum side
  information,'' \emph{Physical Review A}, vol.~68, no.~4, p. 042301, 2003.

\bibitem{Quantum_Slepian_Wolf}
A.~Anshu, R.~Jain, and N.~A. Warsi, ``A generalized quantum {Slepian–Wolf},''
  \emph{IEEE Transactions on Information Theory}, vol.~64, no.~3, pp.
  1436--1453, March 2018.

\bibitem{berta2014identifying}
M.~Berta, J.~M. Renes, and M.~M. Wilde, ``Identifying the information gain of a
  quantum measurement,'' \emph{IEEE Transactions on Information Theory},
  vol.~60, no.~12, pp. 7987--8006, 2014.

\bibitem{holevo}
A.~S. Holevo, \emph{Quantum systems, channels, information: a mathematical
  introduction}.\hskip 1em plus 0.5em minus 0.4em\relax Walter de Gruyter,
  2012, vol.~16.

\bibitem{slepian}
D.~S. Slepian and J.~K. Wolf, ``{Noiseless coding of correlated information
  sources},'' \emph{IEEE Transactions on Information Theory}, vol.~19, no.~4,
  p. 471–480, July 1973.

\bibitem{elGamal}
A.~El~Gamal and Y.-H. Kim, \emph{Network information theory}.\hskip 1em plus
  0.5em minus 0.4em\relax Cambridge university press, 2011.

\bibitem{Wilde_book}
M.~M. Wilde, ``From classical to quantum shannon theory,'' \emph{arXiv preprint
  arXiv:1106.1445}, 2011.

\bibitem{fourier-motzkin}
G.~M. Ziegler, \emph{Lectures on polytopes}.\hskip 1em plus 0.5em minus
  0.4em\relax Springer Science \& Business Media, 2012, vol. 152.

\bibitem{Farhad_Dist}
F.~{Shirani} and S.~S. {Pradhan}, ``Finite block-length gains in distributed
  source coding,'' in \emph{2014 IEEE International Symposium on Information
  Theory}, June 2014, pp. 1702--1706.

\bibitem{Conway1985}
J.~B. Conway, \emph{A Course in Functional Analysis}.\hskip 1em plus 0.5em
  minus 0.4em\relax Springer New York, 1985.

\bibitem{meyer2006quantum}
P.~A. Meyer, \emph{Quantum probability for probabilists}.\hskip 1em plus 0.5em
  minus 0.4em\relax Springer, 2006.

\bibitem{wagner2008improved}
A.~B. Wagner and V.~Anantharam, ``An improved outer bound for multiterminal
  source coding,'' \emph{IEEE Transactions on Information Theory}, vol.~54,
  no.~5, pp. 1919--1937, 2008.

\bibitem{martens1990nonideal}
H.~Martens and W.~M. de~Muynck, ``Nonideal quantum measurements,''
  \emph{Foundations of Physics}, vol.~20, no.~3, pp. 255--281, 1990.

\bibitem{atif2020source}
T.~A. Atif, A.~Padakandla, and S.~S. Pradhan, ``Source coding for synthesizing
  correlated randomness,'' \emph{arXiv preprint arXiv:2004.03651}, 2020.

\end{thebibliography}
